\documentclass[10pt,aps,twocolumn,superscriptaddress,footinbib,floatfix]{revtex4-2}
\usepackage{graphicx}
\usepackage{indentfirst}
\usepackage{siunitx}
\usepackage{physics}
\usepackage{braket}
\usepackage{float}
\usepackage{mathtools}
\usepackage{footnote}
\usepackage{esint}
\usepackage{comment}
\usepackage{color}
\usepackage[T1]{fontenc}
\usepackage{amsfonts}
\usepackage{amsmath}
\usepackage{footmisc}
\usepackage{scrextend}
\usepackage{multirow}
\usepackage{diagbox}
\usepackage[hyperfootnotes=false]{hyperref}
\usepackage[english]{babel}
\usepackage{url}
\usepackage{bm}
\usepackage{etoolbox}
\usepackage{enumerate}
\usepackage{soul}
\usepackage{amssymb}
\definecolor{darkblue}{rgb}{0,0,0.5}
\hypersetup{
    colorlinks=true,
    linkcolor=black,
    filecolor=blue,
    citecolor=darkblue,  
    urlcolor=black,
}

\usepackage[normalem]{ulem}

\apptocmd{\sloppy}{\hbadness 9999\relax}{}{}

\DeclareMathOperator{\E}{\mathbb{E}}

\DeclarePairedDelimiterX\parvertent[2]{\lparen}{\rparen}%
{#1\,\delimsize\vert\,\mathopen{}#2}

\usepackage{stmaryrd}
\usepackage{trimclip}

\makeatletter
\DeclareRobustCommand{\shortto}{%
  \mathrel{\mathpalette\short@to\relax}%
}

\newcommand{\short@to}[2]{%
  \mkern2mu
  \clipbox{{.5\width} 0 0 0}{$\m@th#1\vphantom{+}{\shortrightarrow}$}%
  }
\makeatother

\newtheorem{theorem}{Theorem}

\newtheorem{corollary}[theorem]{Corollary}

\newtheorem{definition}[theorem]{Definition}

\newtheorem{lemma}[theorem]{Lemma}

\newtheorem{proposition}[theorem]{Proposition}

\newenvironment{proof}[1][Proof]{\noindent\textbf{#1.} }{\ \rule{0.5em}{0.5em}}

\newcommand{\calA}{{\cal A}}

\newcommand{\calL}{{\cal L}}
 
\newcommand{\calN}{{\cal N}}

\newcommand{\calS}{{\cal S}}

\newcommand{\1}{^{(1)}}

\newcommand{\bI}{\boldsymbol I}

\newlength{\ketketwidth}
\newlength{\ketwidth}

\def\be{\begin{equation}}
\def\ee{\end{equation}}
\def\ba{\begin{eqnarray}}
\def\ea{\end{eqnarray}}

\newif\ifappendixmaterial
\appendixmaterialfalse
\makeatletter
\let\arxiv@addcontentsline\addcontentsline
\renewcommand{\addcontentsline}[3]{%
  \ifstrequal{#1}{toc}{%
    \ifappendixmaterial\arxiv@addcontentsline{#1}{#2}{#3}\fi
  }{\arxiv@addcontentsline{#1}{#2}{#3}}%
}
\makeatother
\hypersetup{
  pdftitle={Generative learning of time series with holographic quantum circuit sampling},
  pdfauthor={Bingzhi Zhang, Matthew W. Lee, Tom Finke, Quntao Zhuang},
  pdfsubject={Main article and appendices}
}

\begin{document}

\title{Generative learning of time series with holographic quantum circuit sampling}

\author{Bingzhi Zhang}
\affiliation{Ming Hsieh Department of Electrical and Computer Engineering, University of Southern California, Los Angeles, California 90089, USA}
\affiliation{Applied Quantum Software, San Mateo, California 94402, USA}
\author{Matthew W. Lee}
\affiliation{Applied Quantum Software, San Mateo, California 94402, USA}
\affiliation{Triaxiom Capital, LLC, San Francisco, CA 94104, USA}
\author{Tom Finke}
\affiliation{Applied Quantum Software, San Mateo, California 94402, USA}
\author{Quntao Zhuang}
\email{quntao@appqsoft.com}
\affiliation{Ming Hsieh Department of Electrical and Computer Engineering, University of Southern California, Los Angeles, California 90089, USA}
\affiliation{Applied Quantum Software, San Mateo, California 94402, USA}
\affiliation{Department of Physics and Astronomy, University of Southern California, Los Angeles, California 90089, USA}

\begin{abstract}
Quantum generative models face growing qubit requirements and concentrated loss landscapes when learning long time series. Here we introduce holographic quantum circuit sampling (HQCS), which combines mid-circuit measurements, qubit reuse and independently parameterized temporal blocks to generate trajectories with a fixed number of qubits. For Haar-random HQCS and a normalized exponential kernel, we prove inverse-polynomial bounds on maximum mean discrepancy loss variance, contrasting with exponential concentration in global Haar-random generators. For generic members of a family of nonstationary Markov processes, temporal parameterization reduces exact-generation gate counts from quadratic to linear in sequence length relative to a specified class of repeated shared-block generators. Simulations demonstrate learning of Gaussian autoregressive processes and historical Volatility Index (VIX) trajectories, with the latter achieving over tenfold lower maximum mean discrepancy than the tested classical benchmarks. 
We implement HQCS on IBM Phoenix and retain lower sequence-space MMD than the tested classical benchmarks.
These results establish a framework for quantum time-series generation with fixed qubit requirements and provable avoidance of exponential loss concentration as sequence length grows. 
\end{abstract}

\maketitle

\section{Introduction}
Learning to generate time series requires capturing both the distribution of individual observations and the correlations between observations at different times. Such data arise in stochastic physical processes, sensor records, communication signals and financial markets. A generative model must learn a joint distribution over complete trajectories and produce new samples that reproduce its statistical structure. As trajectories become longer, the challenge is to represent their temporal dependence while keeping the resources required for generation and learning manageable.

Quantum generative learning approaches this task by training parameterized quantum circuits to reproduce distributions inferred from data. Quantum generative adversarial networks (QGANs)~\cite{lloyd2018quantum,dallaire2018quantum} and quantum circuit Born machines (QCBMs)~\cite{liu2018differentiable} provide frameworks for modeling multivariate joint distributions through quantum sampling~\cite{zhu2022generative,zhu2023copula,rudolph2022generation}. These approaches build on the ability of quantum circuits to generate complex correlated distributions, a capability also explored in random-circuit-sampling experiments~\cite{arute2019quantum,hangleiter2023,wu2021strong,zhu2022quantum,morvan2024phase,gao2025establishing,deCross2025,abanin2025observation,ransford2025helios}. Applying this capability to time-series learning requires architectures and training objectives that remain effective as sequence length increases.

Two fundamental challenges arise in conventional quantum circuit generators that encode an entire trajectory in a quantum register and measure it at the end. First, at fixed output precision, the number of required physical qubits grows linearly with sequence length. Earlier demonstrations of quantum copula learning
considered joint distributions of two to four
variables~\cite{zhu2022generative,zhu2023copula}, largely limited by physical qubits. Second, the training objective can become exponentially concentrated as the register grows, suppressing the differences in loss that guide optimization~\cite{mcclean2018barren,cerezo2021cost,zhang2025energy,rudolph2024trainability}. The severity of this concentration depends on both the circuit architecture and the objective although restricted ansatzes and suitable loss functions can alleviate some of these difficulties~\cite{rudolph2024trainability,kasture2023protocols,recio2025train,zhong2024many}. For time-series generation, the central question is whether a model can accommodate long trajectories and explicitly time-dependent dynamics while avoiding both growing physical qubit requirements and exponential loss concentration.

In this work, we propose holographic quantum circuit sampling
(HQCS) for generative learning of time-series distributions.
Building on sequential quantum evolution interleaved with register measurement and reuse~\cite{zhang2025holographic,zhang2025anticoncentrated},
HQCS generates length-$T$ trajectories using a fixed number of physical
qubits independent of $T$.
Unlike recurrent models that share circuit parameters across
time~\cite{takaki2021learning,li2023quantum}, HQCS assigns
independently trainable parameters to successive temporal blocks.
For a fixed block ansatz, its parameter count grows linearly
with $T$, allowing explicitly time-dependent transition dynamics.

Our main theoretical result is that this growing parameter
dimension does not entail exponential decay of the variance of the
maximum mean discrepancy (MMD) loss.
We first rigorously show that conventional quantum circuit generators using a single deep global quantum circuit suffer from exponentially suppressed MMD loss variance with system size thus sequence length $T$. In contrast, HQCS admits inverse-polynomial upper and lower bounds on this variance with sequence length. We establish these bounds for fixed register size and independently-sampled random unitary at each step. HQCS can thus accommodate increasing temporal flexibility while avoiding exponential loss concentration. We numerically verify the theory on various classes of time series. 
As a complementary result, we establish an exact-representation 
resource separation for a family of nonstationary Markov processes.
HQCS generates generic members using one qubit and $O(T)$ 
gates, whereas recurrent models sharing circuit parameters across
time~\cite{takaki2021learning,li2023quantum} require $\Omega(T^2)$ gates when the shared
block is executed at every step.

We demonstrate the generative learning performance of HQCS
on Gaussian autoregressive (AR) processes and historical
Volatility Index (VIX) trajectories.
For the AR processes, HQCS learns temporal correlations
with performance comparable to a Gaussian benchmark whose
model family contains the target distribution.
For VIX, HQCS learns the historical trajectory distribution
and achieves more than an order of magnitude lower MMD
than the tested Gaussian and approximately parameter-matched
TimeGAN benchmarks~\cite{yoon2019time}. We validate our model on IBM Quantum devices, and the sequence MMD performance advantage of HQCS persists over both Gaussian model and TimeGAN.
Together, these results establish a framework for quantum
generative learning that combines fixed physical qubit
requirements, explicitly time-dependent dynamics and
provable avoidance of exponential loss concentration
with sequence length, and demonstrate its potential for generative modeling on near-term quantum devices.

\begin{figure}[tbp]
    \centering
    \includegraphics[width=\columnwidth]{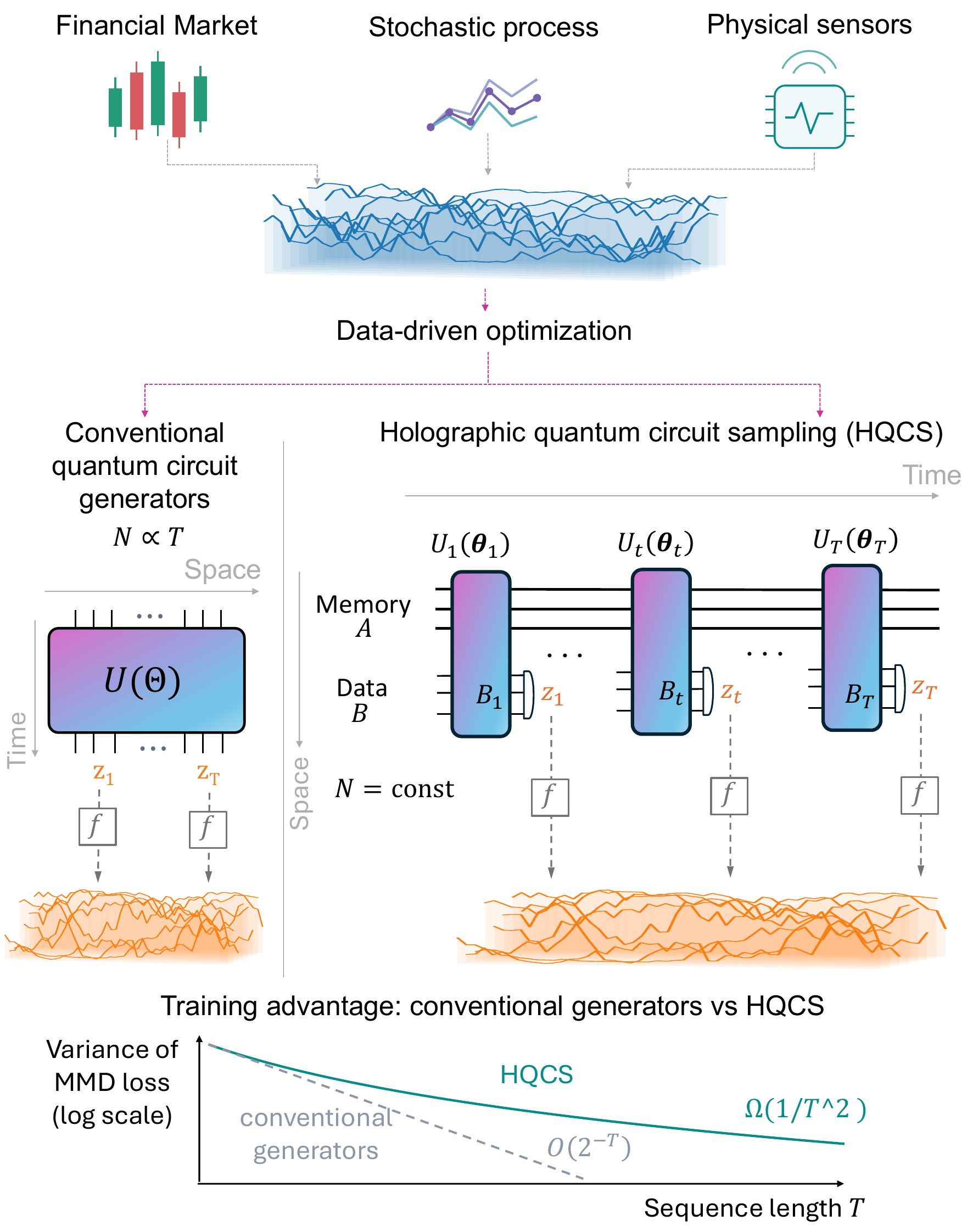}
    \caption{Summary of main results in this work. The target distribution of time series can arise from financial markets, stochastic processes and physical sensors. Quantum models use real data samples from target distribution to learn the distribution and generate new samples. Conventional quantum circuit generators (left) use a single parameterized circuit with system size linearly grows $N \propto T$. Holographic quantum circuit sampler (HQCS) instead applies parameterized unitary on a fixed number of qubits in each step followed by mid-circuit measurements and reuse. Measurement records are used to build time sequences (orange). In the bottom panel, we schematically show the main theoretical results. Conventional generators exhibits exponential decay of the variance of MMD loss, while HQCS only has a polynomial decay versus sequence length $T$.
    }
    \label{Q_circuit_schematic}
\end{figure}

\section{Generative learning task}
We are given a data set $\calS=\{\bm x^{(k)}\}_{k=1}^{N_s}$ drawn from an unknown target distribution $P$, where $\bm x^{(k)}=\left(x_1^{(k)},\ldots,x_T^{(k)}\right)$ is sampled from an unknown joint distribution $P$. The target distribution encodes both the marginal behavior at each time step and the dependence across time. A successful generator must reproduce both the single-time marginals and the temporal dependence. Our goal is to construct a parameterized distribution $Q_\Theta$ that approximates the target distribution and from which samples $\tilde{\bm x}=(\tilde{x}_1,\ldots,\tilde{x}_T)$ can be generated.

As shown in Ref.~\cite{zhu2022generative}, Sklar's theorem~\cite{sklar1959fonctions,durante2013topological} guarantees that, by applying the continuous cumulative distribution function (CDF) $F_t(x)$, random variables $\{x_t\}$ can be converted to $\{u_t=F_t(x_t)\}$, whose marginals are uniformly distributed on $[0,1]$. Therefore, sampling from the copula, i.e., the joint distribution of $\{u_t\}$, followed by applying inverse marginal CDF, enables direct sampling from the original distribution. With this conversion, we therefore formulate the generative learning task on time series with marginals that are uniform on $[0,1]$, without loss of generality.

To compare the target distribution $P$ and model distribution $Q_\Theta$, we adopt a sample-based measure, squared maximum mean discrepancy (MMD),
defined as
\begin{align}
\mathcal{L}_{\rm MMD}(P, Q_\Theta)
&\coloneqq
\mathbb{E}_{\bm a,\bm a^\prime\sim
P}
K_T(\bm a, \bm a^\prime)
+
\mathbb{E}_{\bm b, \bm b^\prime \sim Q_\Theta}
K_T(\bm b, \bm b^\prime)
\nonumber\\
&\quad
-
2\mathbb{E}_{\bm a \sim P,
\bm b \sim Q_\Theta}
K_T(\bm a, \bm b),
\label{eq:mmd_loss_main}
\end{align}
where $K_T(\bm x, \bm \tilde{\bm x})\in[0,1]$ is a bounded kernel function between two $T$-dimensional vectors $\bm x, \tilde{\bm x}$ providing a measure of similarity. Here, we consider bounded kernels of the exponential form
\be 
K_T(\bm x, \bm \tilde{\bm x})=\exp\left[-\frac{1}{T}\sum_{t=1}^T \varphi(|x_t - \tilde{x}_t|)\right].
\label{KT_general}
\ee 
We assume throughout that the resulting \(K_T\) is positive semidefinite.
The $1/T$ factor keeps typical kernel values independent of $T$, and the exponential form factorizes across coordinates as $K_T(\bm x, \bm \tilde{\bm x})=\prod_{t=1}^T k_T(x_t, \tilde{x}_t)$, where $k_T(x_t, \tilde{x}_t)=\exp[-\varphi(|x_t-\tilde{x}_t|)/T]$. Common choices of distance-based kernels such as the RBF (Gaussian) $K_T(\bm x, \bm \tilde{\bm x})=\exp(-\|\bm x - \bm \tilde{\bm x}\|^2/(2T\sigma^2))$ and Laplace kernels $K_T(\bm x, \bm \tilde{\bm x})=\exp[-\|\bm x - \bm \tilde{\bm x}\|_1/(T\sigma)]$ satisfy this assumption. In practice, the MMD loss in Eq.~\eqref{eq:mmd_loss_main} is estimated from a finite number of samples, explained in the section ``\hyperref[sec:numerical_simulations]{Learning performance in examples}''.

\begin{figure}[tbp]
    \centering
    \includegraphics[width=\columnwidth]{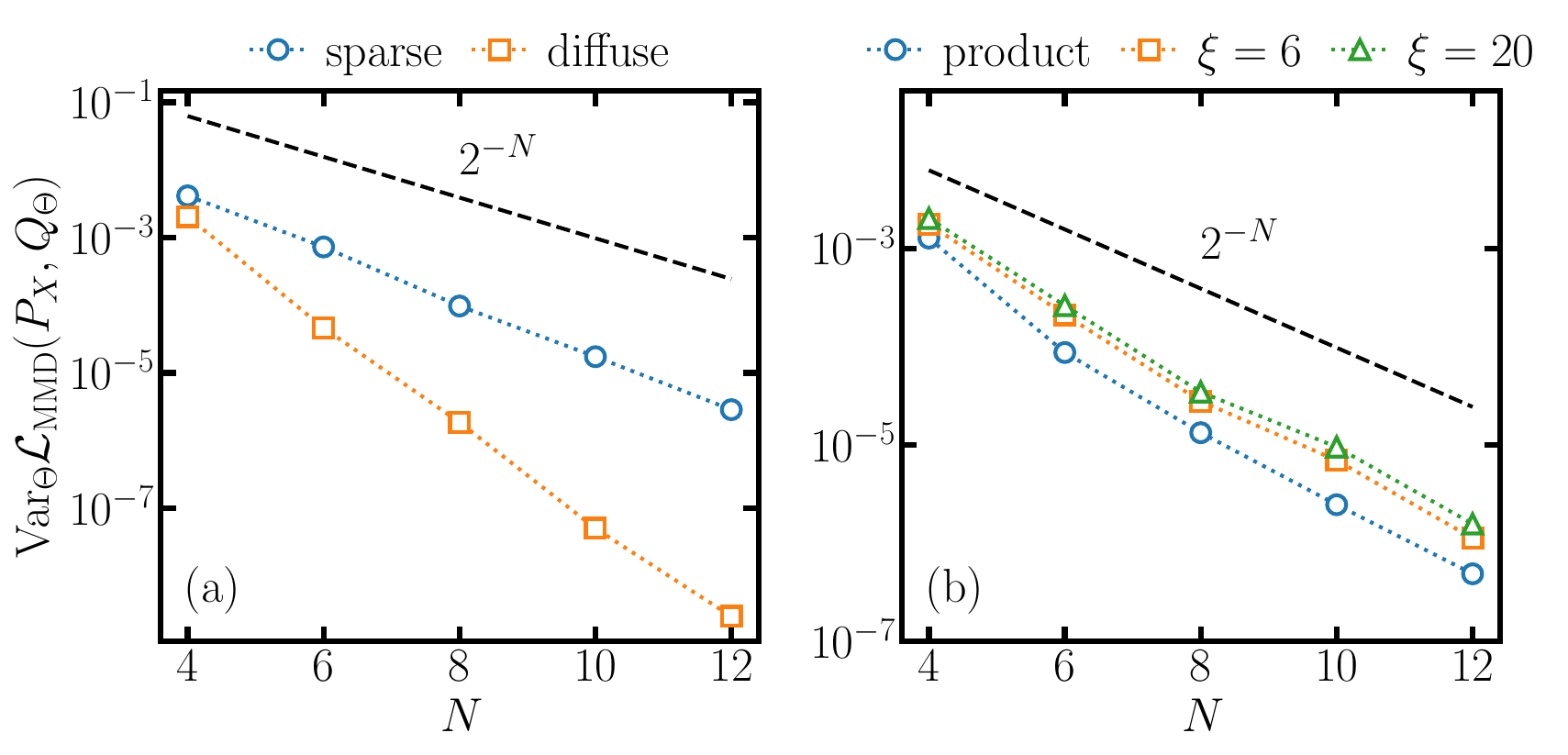}
    \caption{Variance of MMD over Haar-random states versus qubit number $N$. (a)
    Circles and squares show results for sparse and diffuse Dirichlet distributions ${\rm Dirichlet}(\alpha,\dots,\alpha)$ with $\alpha=2^{-N}$ and $\alpha=1$, respectively. (b) Circles show results for products of standard Gaussian distributions. Squares and triangles show AR(1) sequences with correlation length $\xi=6$ and $20$. Dashed lines in both subplots show the scaling of $2^{-N}$.}
    \label{fig:var_haar}
\end{figure}

\section{Challenges for the conventional quantum circuit generators}
We begin with the conventional quantum circuit generator built from a single parameterized unitary $U_\Theta$ and measures all qubits at the end to generate data, shown in Fig.~\ref{Q_circuit_schematic}. The generator takes a product state as input and applies a parameterized quantum circuit $U(\Theta)$ followed by projective measurements on the qubits to generate a data sequence.
Such conventional quantum generators include QCBMs~\cite{liu2018differentiable, rudolph2024trainability} and QGANs~\cite{lloyd2018quantum,zeng2019learning,zoufal2019quantum}. From Born's rule, the sampling probability for a generated sequence $\tilde{\bm x} = (\tilde{x}_1, \dots, \tilde{x}_N)$ is $q_\Theta(\tilde{\bm x}) = |\braket{\tilde{\bm x}|U(\Theta)|0^{\otimes N}}|^2$. Here each entry 
can be either a binary variable, or a discrete variable by grouping multiple qubits. However, the number of bits one can sample from the conventional approach cannot exceed the number of qubits. As we emphasize in the introduction, such a limitation hinders the application of quantum generative learning for long time series.

In addition to its growing qubit requirement, we further show that the conventional scheme also suffers from training challenges.
For Haar-random output states, we have the following result.  
\begin{theorem}
{\rm (Informal version of MMD variance of conventional generators)}
\label{thm:haar_mmd}
Let $Q_{\Theta}$ be the sampling probability distribution obtained from an $N$-qubit Haar-random state (or at least a 4-design) $\ket{\psi_\Theta}$, and let $P$ be the target data distribution. The variance of the MMD with kernel in Eq.~\eqref{KT_general} over the random initialization ensemble is
\be
{\rm Var}_\Theta[\mathcal L_{\rm MMD}]
= O\left(2^{-N}\right).
\label{eq:haar_mmd_main}
\ee
\end{theorem}
The formal theorem and its proof are presented in Appendix~\ref{app:mmd-haar-mean-variance}. Note that the variance of the MMD loss depends on the target data distribution and the choice of kernel and these dependencies are absorbed into the $O$ notation (see details in Appendix~\ref{app:mmd-haar-mean-variance}). The 4-design constraint of $\ket{\psi_\Theta}$ requires circuit depth at most $O(N)$.
Eq.~\eqref{eq:haar_mmd_main} shows the exponential suppression of the fluctuations of MMD loss in training a randomly-initialized global parameterized circuit, which is analog to the known {\em barren plateau}~\cite{mcclean2018barren, cerezo2021cost} for linear cost functions, and provides a rigorous characterization for training deep circuits in generative learning~\cite{rudolph2024trainability}.

We numerically verify this concentration in Fig.~\ref{fig:var_haar}. We first consider a randomly constructed discrete target distribution. For this purpose, we sample a target distribution of binary sequences from the Dirichlet ensemble $P \sim {\rm Dirichlet}(\alpha,\dots, \alpha)$. Once the distribution $P$ is fixed, we evaluate the MMD in Eq.~\eqref{eq:mmd_loss_main}. We consider a concentrated ensemble with $\alpha=2^{-N}$ and a diffuse ensemble with $\alpha=1$. Note that the diffuse ensemble also reproduces the sampling distribution from random circuit sampling. For the Gaussian kernel on binary sequences, the squared Euclidean distance in Eq.~\eqref{KT_general} reduces to the Hamming distance. We find that the numerical variance for the sparse target (circles) align with the scaling of $2^{-N}$ from Theorem~\ref{thm:haar_mmd}, and the variance for the diffuse target further decays as $2^{-2N}$. In Fig.~\ref{fig:var_haar}(b), we group every two qubits into a single variable in $[0,1]$ to obtain sequences of length $N/2$, and compare with CDF-normalized target distributions including products of standard Gaussian distributions (circles) and AR(1) sequences with various correlation lengths $\xi$ (squares and triangles). In all cases, we observe an exponential decay in variance of loss consistent with the theoretical scaling of $2^{-N}$ (dashed line).
Overall, the exponential suppression in MMD loss indicates the difficulty of training and therefore motivates architectures that avoid this training challenge in generative learning.

\section{Holographic quantum circuit sampling and its training advantage}

\subsection{Holographic quantum circuit sampling}
To resolve the physical limitations of sequence generation and the training issue identified in Theorem~\ref{thm:haar_mmd}, we propose the holographic quantum circuit sampling (HQCS).
The unitary $U_t(\bm\theta_t)$ can adopt various types of ansatzes, such as hardware-efficient-ansatz~\cite{kandala2017hardware} and fast-scrambling ansatz~\cite{belyansky2020minimal}. With a sufficient number of layers, we expect the performance of HQCS to be independent of the actual ansatz, whereas the capability will be limited with limited circuit depth or ansatz. In this work, we adopt the hardware-efficient ansatz, detailed in Appendix~\ref{app:details}.
HQCS realizes a sequential model for this normalized joint distribution. As depicted in Fig.~\ref{Q_circuit_schematic}, the quantum device contains $N_A$ memory qubits in register $A$ and $N_B$ data qubits in register $B$. At time step $t$, a parameterized joint unitary $U_t(\bm\theta_t)$ acts on $AB$, after which $B$ is measured to produce a bit string $z_t$. A fixed classical map $f:\{0,1\}^{N_B}\rightarrow[0,1]$ converts the outcome to $\tilde{x}_t=f(z_t)$. The post-measurement memory register is passed forward, while the data register in the measured state is reused at the next step. Repeating this procedure for $t=1,\ldots,T$ produces a generated trajectory $\tilde{\bm x}$. 

Because $A$ is the only coherent quantum register to carry information from one step to the next, increasing $N_A$ enlarges the family of accessible joint distributions, whereas a fixed $N_A$ in general imposes a finite-memory constraint~\cite{zhang2026scaling}.
Fortunately, as many real-world time series naturally have decaying temporal correlations and admits a finite memory, we expect that an appropriately chosen $N_A$ to match the effective-memory requirement of the target process is sufficient for learning. The minimum required system size can therefore be a constant $N_A+N_B$, independent of $T$.

We train HQCS to match the generated and target trajectory distributions using the MMD defined in Eq.~\eqref{eq:mmd_loss_main}. To compare these distributions at different temporal scales, we introduce the windowed sequence-space MMD
\be
    \calL_{\rm MMD}^{(\ell)}(\calS,\tilde{\calS}_\Theta)
    =\frac{1}{T-\ell+1}\sum_{s=1}^{T-\ell+1}
    \calL_{\rm MMD}(\calS_{s,\ell},\tilde{\calS}_{\Theta,s,\ell}),
    \label{eq:window_mmd_def}
\ee
where $1\le\ell\le T$, $\calS_{s,\ell}=\{(x_s^{(k)},\ldots,x_{s+\ell-1}^{(k)})\}_k$ is the target sample set for the length-$\ell$ window starting at $s$, and $\tilde{\calS}_{\Theta,s,\ell}$ is its generated counterpart. For $\ell=1$, this measure compares the marginal distributions at each time step; for $\ell=T$, it reduces to the full-sequence MMD in Eq.~\eqref{eq:mmd_loss_main}. 
However, in practice, with finite target and generated samples, minimizing the empirical sequence-space discrepancy may leave residual mismatches in local temporal variations such as single-step changes. To better capture these local temporal variations, we can also optionally adopt an additional increment-space MMD loss, defined by $\Delta x_t=x_t-x_{t-1}$ for $t=2,\ldots,T$. With increment-space dataset $\Delta\calS_{s,r}=\{(\Delta x_{s+1}^{(k)},\ldots,\Delta x_{s+r}^{(k)})\}_k$ and the corresponding generated set $\Delta\tilde{\calS}_{\Theta,s,r}$ for $1\le r\le T-1$, we introduce the  windowed increment-space MMD as
\begin{align}
    \calL_{\rm MMD}^{\Delta,(r)}(\calS,\tilde{\calS}_\Theta)
    =\frac{1}{T-r}\sum_{s=1}^{T-r} 
    \calL_{\rm MMD}(\Delta\calS_{s,r},
        \Delta\tilde{\calS}_{\Theta,s,r}).
    \label{eq:hqcs-increment-window-loss}
\end{align}
We can still choose the normalized exponential kernel in Eq.~\eqref{KT_general} with normalization factor of $T-1$.
The choice $r=1$ compares single-increment marginals at corresponding time steps, whereas $r=T-1$ compares the joint distributions of the complete increment vectors. Overall, we propose a combined training loss as
\begin{align}
    \calL(\Theta)
    =\calL_{\rm MMD}^{(\ell)}(\calS,\tilde{\calS}_\Theta)
    +\lambda \calL_{\rm MMD}^{\Delta,(r)}
        (\calS,\tilde{\calS}_\Theta),
    \label{eq:hqcs-training-loss}
\end{align}
where $\lambda\ge0$ is a hyperparameter to control their relative weight. These tunable parameters provide us flexibility to control the learning target through training. The MMD loss is also efficient for estimation in practice, as $N_s=O(\epsilon^{-2})$ independent trajectory samples from each distribution suffice to estimate the MMD loss to additive error $\epsilon$ at fixed confidence~\cite{gretton2012kernel}.

Before moving on, a few points are worth discussing. 
Generally, the measurement outcome $z_t$ may not be uniform random, although it may be close to uniform when $N_B$ is small and each unitary approaches a 2-design~\cite{zhang2025anticoncentrated}.
HQCS can also generate $n_V\ge1$ variables at each time step, by extending the classical map to multiple outputs, $f:\{0,1\}^{N_B}\rightarrow[0,1]^{n_V}$. This can be done in various ways. One possible approach is to group $N_B$ data qubits into $n_V$ groups and apply an individual function to each group for conversion into a single variable in $[0,1]$. A trainable readout map $f$ (e.g. by a classical neural network) is a possible extension.

\subsection{Variance of MMD}
We next present the main result for the variance of MMD loss in sequence and increment-space of HQCS with fixed physical-register sizes $N_A, N_B$ to characterize its landscape thus training hardness. Here we analyze the variance of loss for each MMD separately, and have the following result for the sequence-space MMD.
\begin{theorem}
\label{thm:hqcs-main-concentration}
{\rm (Informal MMD variance for HQCS)}
For randomly-initialized, fixed-size HQCS with independent Haar-random unitary at each step and the readout map $f$ satisfying that $\exists i\neq j, \ s.t.\ \varphi(|f(i)-f(j)|) > 0$,  the variance of MMD loss with normalized exponential kernel in Eq.~\eqref{KT_general} obeys
\be
	\frac{C_-}{T^2}
	\le
	\operatorname{Var}_{\Theta}[\calL_{\rm MMD}]
	\le
	\frac{C_+}{T},
\ee
where $C_-, C_+>0$ are independent of $T$ and target distribution $P_T$.
\end{theorem}
The formal theorem and proofs of the upper and lower bound are provided in Appendices~\ref{app:hqcs-upper-organized} and~\ref{app:hqcs-curvature-lower}, respectively. For the upper bound, it suffices to draw each unitary block from an ensemble at least 4-design. For the lower bound, a 4-design for $U_1$ and 2-design for remaining unitary blocks are sufficient. The condition for the readout map in Theorem~\ref{thm:hqcs-main-concentration} can be naturally satisfied by a wide family of non-constant choices. This theorem establish an exponential separation in concentration of MMD loss with $T$ from conventional scheme proved in Theorem~\ref{thm:haar_mmd}.

We begin with the numerical verification on a special product distribution example, where the target and model distributions are generated by concatenating independent sequences of length $\delta=6$ from an AR(1) process with correlation length $\xi=16$ and a randomly-initialized HQCS, respectively. Therefore, $T/\delta$ counts the number of temporal sectors. In Fig.~\ref{fig:var_loss}(a), the variance decays as $\delta/T$ (dashed line) for various choices of $N_A$. In Fig.~\ref{fig:var_loss}(b), we further extend to general HQCS and target distributions with a particular focus on the following four examples. Diamonds show an AR(1) process with a correlation length of $\xi=20$. Circles show global Gaussian copula $x_t = \Phi(\sqrt{\rho} y + \sqrt{1-\rho} \epsilon_t)$ where $\Phi$ is the cumulative distribution function (CDF) for $\calN(0,1)$. Here, the white noise $\epsilon_t \sim \calN(0,1)$ is independently drawn over time, while the latent variable $y\sim \calN(0,1)$ is fixed across the sequence, thus providing an infinite-range-correlated counterpart to AR(1).
Squares show two-point delta distribution $P_{\rm 2pt}(\bm x) = (\delta_{\bm x,\bm 0} + \delta_{\bm x, \bm 1})/2$ where sequences can only be either all zero or all one with maximal global correlation. Finally, we consider a target distribution generated by another random HQCS (triangles) of $N_A^{\rm tar}=9$ and $N_B^{\rm tar}=3$ with each unitary block being an independent Haar-random unitary. Such a circuit can approximate statistical properties from conventional random circuit sampling~\cite{zhang2025anticoncentrated}. Overall, the decay of the variance of MMD loss is bounded by the scaling from $1/T^2$ to $1/T$ in Theorem~\ref{thm:hqcs-main-concentration}.

\begin{figure}[tbp]
    \centering
    \includegraphics[width=\columnwidth]{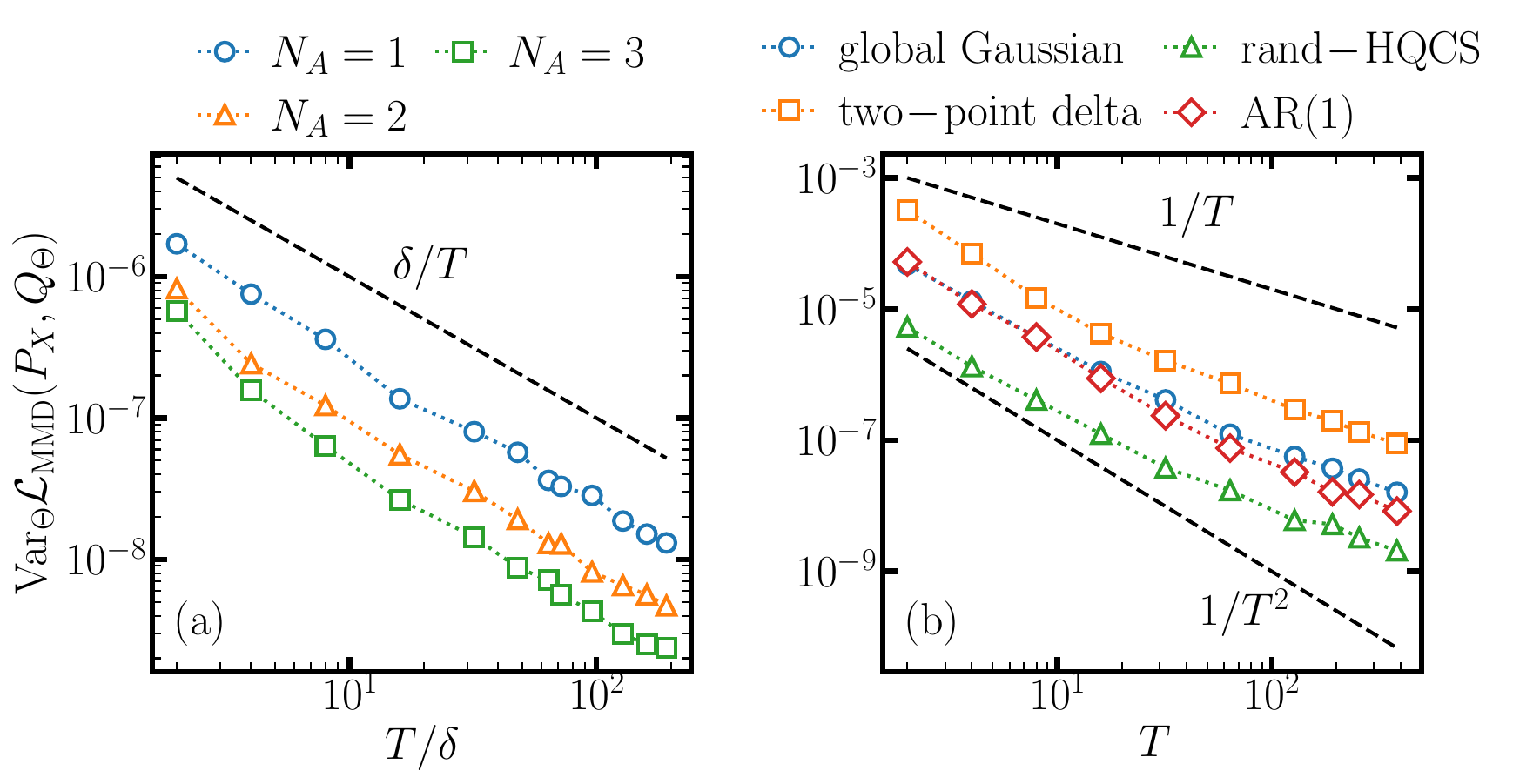}
    \caption{Variance of MMD for HQCS versus number of steps $T$. (a) Target and model distributions are generated by concatenating $T/\delta$ i.i.d.-sampled AR(1) sequences with correlation length $\xi=16$ and independently initialized HQCS sequences of length $\delta=6$, respectively. Different markers show circuits of $N_A=1, 2, 3$ and $N_B = 5$. (b) Target distribution includes global Gaussian copula (circles) with $\rho=0.95$, two-point delta (squares), AR(1) with $\xi=20$ (diamonds) and random HQCS of $N_A^{\rm tar}=9, N_B^{\rm tar} = 3$ (triangles). Model HQCS consists of $N_A=4, N_B=3$. Black dashed lines represent scalings of $1/T^2$ and $1/T$.
    In both subplots, the parameterized model circuit has depth $L=4(N_A + N_B)$ for each block.}
    \label{fig:var_loss}
\end{figure}

As a supplement, we next analyze the trainability of HQCS in learning increment-space
Eq.~\eqref{eq:hqcs-increment-window-loss}, using the
normalized Gaussian kernel with fixed bandwidth $\sigma$. We again take $r = T-1$ to focus on the temporal scaling.
Define the mean and variance of $f(z)$ with uniformly sampled $z$ as $\mu = d_B^{-1}\sum_z f(z)$ and $v = d_B^{-1}\sum_z (f(z) - \mu)^2$,
and write $J_T(x)=(T-1)^{-1}\sum_{t=2}^T(\Delta x_t)^2$
for the average squared increment. We have the following result.
\begin{theorem}
\label{cor:hqcs-increment-main}
{\rm (Informal version of increment MMD variance for HQCS)}
For randomly-initialized, fixed-size HQCS with independent Haar-random unitary at each step, the variance of increment-space MMD satisfies
$\operatorname{Var}_{\Theta}[\mathcal L_{\rm MMD}^{\Delta}]=O(T^{-1})$. 
In particular, for a Gaussian kernel and a non-constant readout map, if the target distribution satisfies $|\mathbb E_{x\sim P_T}[e^{-J_T(x)/(2\sigma^2)}] - e^{-v/\sigma^2}|\ge \varepsilon$ for constant $\varepsilon>0$ independent of $T$
and all sufficiently large $T$, then
\begin{equation}
    \frac{C_-^\Delta}{T}
    \le \operatorname{Var}_{\Theta}[\mathcal L_{\rm MMD}^{\Delta}]
    \le \frac{C_+^\Delta}{T},
\end{equation}
where $C_\pm^\Delta > 0$ are constants independent of $T$.
\end{theorem}
The formal theorem for general kernel in Eq.~\eqref{KT_general} and proofs are provided in Appendix~\ref{app:hqcs-increment}. This theorem also shows inverse polynomial concentration of $\Theta(1/T)$ for the landscape of increment-space MMD under stated condition for target distribution.
Intuitively, $e^{-J_T/(2\sigma^2)}$ assigns greater weight to trajectories
with smaller step-to-step changes. The condition in Theorem~\ref{cor:hqcs-increment-main} requires that the increment-space fluctuation of target distribution should be discriminated from the uniform reference readout. The target data distributions in Fig.~\ref{fig:var_loss}(a) and (b) all satisfy this condition except for the random HQCS which asymptotically equivalent to randomly-initialized HQCS.

\section{Polynomial resource separation from temporal parameterization}
Beyond its favorable loss-variance scaling, HQCS allows independently parameterized circuit blocks at successive time steps. We compare the circuit resources required for exact sequence generation by HQCS and shared-parameter quantum recurrent generators, by comparing HQCS with shared-parameter quantum recurrent generators motivated by Refs.~\cite{takaki2021learning,li2023quantum}. In the shared model, the same parameterized unitary block $U(\theta)$ is executed at every temporal step, and all target-dependent information is encoded in its continuous gate parameters while state initialization, readout, and classical processing rules are target-independent. We count elementary gates and circuit depth per generated trajectory in a system with $N$ physical qubits in total. Under these assumptions, temporal parameterization yields the following polynomial resource separation and the formal statement and proof are given in Appendix~\ref{proof:resource separation}.
\begin{theorem}
\label{thm:temporal_resource_separation_main}
(Informal version for polynomial resource separation). There exist length-$T$ nonstationary binary Markov processes with uniform single-time marginals that HQCS can generate exactly using one physical qubit and $O(T)$ elementary gates. In contrast, any shared-block generator requires $\Omega(T)$ independent parameters in its repeated block, and hence $\Omega(T^2)$ gates and $\Omega(T^2/N)$ depth to generate the same process exactly .
\end{theorem}
The physical intuition is illustrated by a binary random process with bit-transition probability specified for each step. Starting from a uniformly distributed bit, the marginals remain uniform, while the temporal correlations depend on the transition schedule. Generating the next outcome requires only the current bit, but specifying the dynamics requires the full sequence of transition probabilities unless the transition probabilities are fixed over different steps.
Both architectures must encode the independent transition probabilities. For almost every schedule in this family, exact representation requires a number of independent continuous parameters proportional to $T$. HQCS distributes these parameters across temporal blocks, each executed once. However, the shared generator encodes the schedule in a single block containing $\Omega(T)$ elementary gates and hence gives an $\Omega(T^2)$ total gate count and depth of $\Omega(T^2/N)$ for the whole sequence. An internal clock can identify the current time step, but the corresponding transition probabilities must still be encoded in the shared block. A clock-based construction attains $O(T^2)$ total gates using $O(\log T)$ qubits, matching the gate-count lower bound under complete-block repetition. A detailed analysis is given in Appendix~\ref{proof:resource separation}.
This separation concerns exact generation of generic processes in the specified family under the stated execution model, and structured transition schedules may admit more economical shared representations.

\section{Learning performance in examples}
\label{sec:numerical_simulations}
In this work, we study the performance of HQCS in learning time series and benchmark it against other classical models. In practice, we use two standard empirical estimators for MMD of Eq.~\eqref{eq:mmd_loss_main}. During training, we adopt the biased estimator as a direct finite-sample version of Eq.~\eqref{eq:mmd_loss_main}. For performance benchmarking against other models, we instead use the unbiased one, which removes the finite-sample contribution from self-similarity and makes it suitable for comparisons across different distributions. Note that the biased estimator is always nonnegative while the unbiased one can sometimes be negative.
To benchmark the generated distribution, we evaluate the windowed sequence-space MMD in Eq.~\eqref{eq:window_mmd_def} for $\ell=1,\ldots,T$, using the unbiased estimator. Varying $\ell$ probes agreement from single-time marginals to the full trajectory distribution.

\begin{figure}[tbp]
    \centering
    \includegraphics[width=\columnwidth]{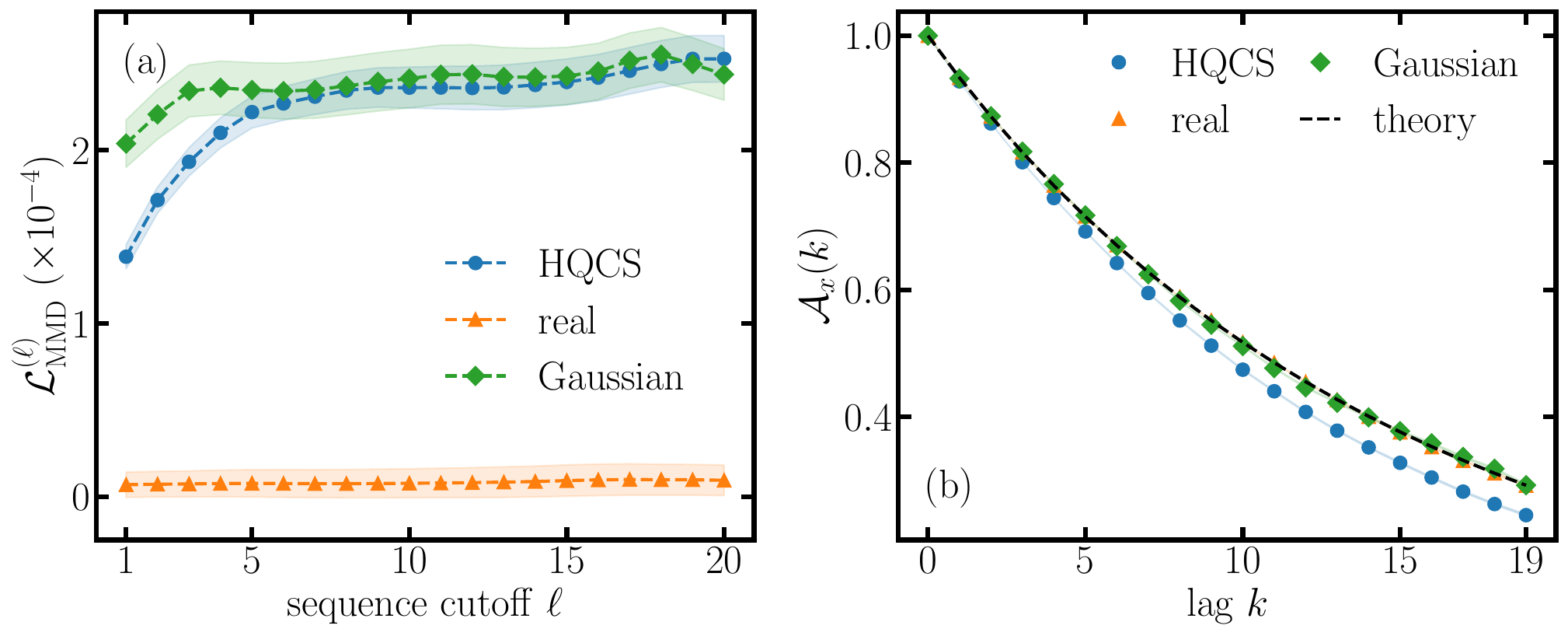}
    \caption{
    Generative learning of Gaussian AR(1) process. (a) MMD loss versus sequence cutoff length $\ell$ with $\xi = 16$ for HQCS (circles), real data baseline (triangles) and Gaussian benchmark (diamonds).
    Shaded areas show standard error of 20 iid sample sets.
    (b) ACF versus time lag $k$ for AR(1) with $\xi=16$ using HQCS (circles), real data (triangles) and Gaussian benchmark (diamonds). The black dashed line represents the theory of $\calA_x(k) = 6\arcsin(\phi^k/2)/\pi$. 
    Here we choose $N_A=3, N_B=5$ for HQCS with a total length of $T=20$ steps and the bandwidth in Gaussian kernel is $\sigma = 0.3$. }
    \label{fig:ar1_main}
\end{figure}

\subsection{Autoregressive model}
We first consider the autoregressive (AR) model as a toy example to test HQCS. For the AR(1) process of interest, its dynamics is formally defined as
\be
    y_{t+1} = \phi y_t + \sqrt{1-\phi^2} \epsilon_t,
    \label{eq:ar1_def}
\ee
where the only parameter $\phi$ controls the correlation between temporal steps and $\epsilon_t \sim \calN(0,1)$ is Gaussian white noise. The prefactor $\sqrt{1-\phi^2}$ is introduced to keep the marginal variance unity, ${\rm Var}[y_t] = 1$. We can define the normalized autocorrelation function (ACF) of the stationary AR(1) process at time lag $k\ge 0$ as
\be
    \calA_y(k)
    \coloneqq \frac{{\rm Cov}(y_t, y_{t+k})}{{\rm Var}(y_t)}.
    \label{eq:acf_def}
\ee
Note that $\calA_y(k)$ is equivalent to the covariance of distribution because of unity variance. For a Gaussian initial condition $y_0 \sim \calN(0,1)$, the marginal distribution at any step is also Gaussian $y_t \sim \calN(0,1)$ and the ACF is therefore $\calA_y(k) =  \phi^{k}$ with temporal correlation length defined as $\xi \coloneqq -1/\ln \phi$. A CDF normalization is further applied to map the raw variable $y_t$ to its cumulative probability, $x_t=\Phi(y_t) \sim {\rm U}[0,1]$, which is uniformly distributed on $[0,1]$. For the Gaussian AR(1) process, this monotone marginal transformation modifies the ACF to $\calA_x(k)=6\arcsin(\phi^k/2)/\pi$, while preserving the same exponential decay length at large time separation.

For AR(1), we train the HQCS solely with full-sequence MMD, namely $\lambda = 0$ and $\ell=T$ in Eq.~\eqref{eq:hqcs-training-loss}.
In Fig.~\ref{fig:ar1_main}(a), we choose the target AR(1) with correlation length $\xi=16$ as an intuitive example. Using only $1000$ training samples, the generated distribution from HQCS remains close to the true AR(1) process in terms of MMD as $\calL_{\rm MMD}^{(\ell)}\le 2.6\times 10^{-4}$ for all $\ell \le T$ evaluated with $2\times 10^4$ samples. The slight increase in MMD with $\ell$ may be due to the finite memory register $N_A=3$. Moreover, we also consider a classical Gaussian benchmark by generating data from a multivariate Gaussian distribution with matching mean and covariance (see details in Appendix~\ref{app:details}), which is expected to reproduce the Gaussian AR(1) process.
Our HQCS result closely aligns with the Gaussian benchmark (diamonds), showing a comparable learning performance. Note that the MMD gap between Gaussian and real data is mainly due to the finite available training set.
As a complementary two-point diagnostic, we examine whether the HQCS also recovers the correlation parameter of the target AR(1) process. In Fig.~\ref{fig:ar1_main}(b), we again take the example of $\xi=16$. Fitting the empirical ACF gives $\hat{\phi} = 0.932$ (HQCS) and $\hat{\phi} = 0.939$ (Gaussian), both of which approach the theoretical value of $\phi=0.939$, indicating the capability to capture the correlation parameter.

\begin{figure}[tbp]
    \centering
    \includegraphics[width=\columnwidth]{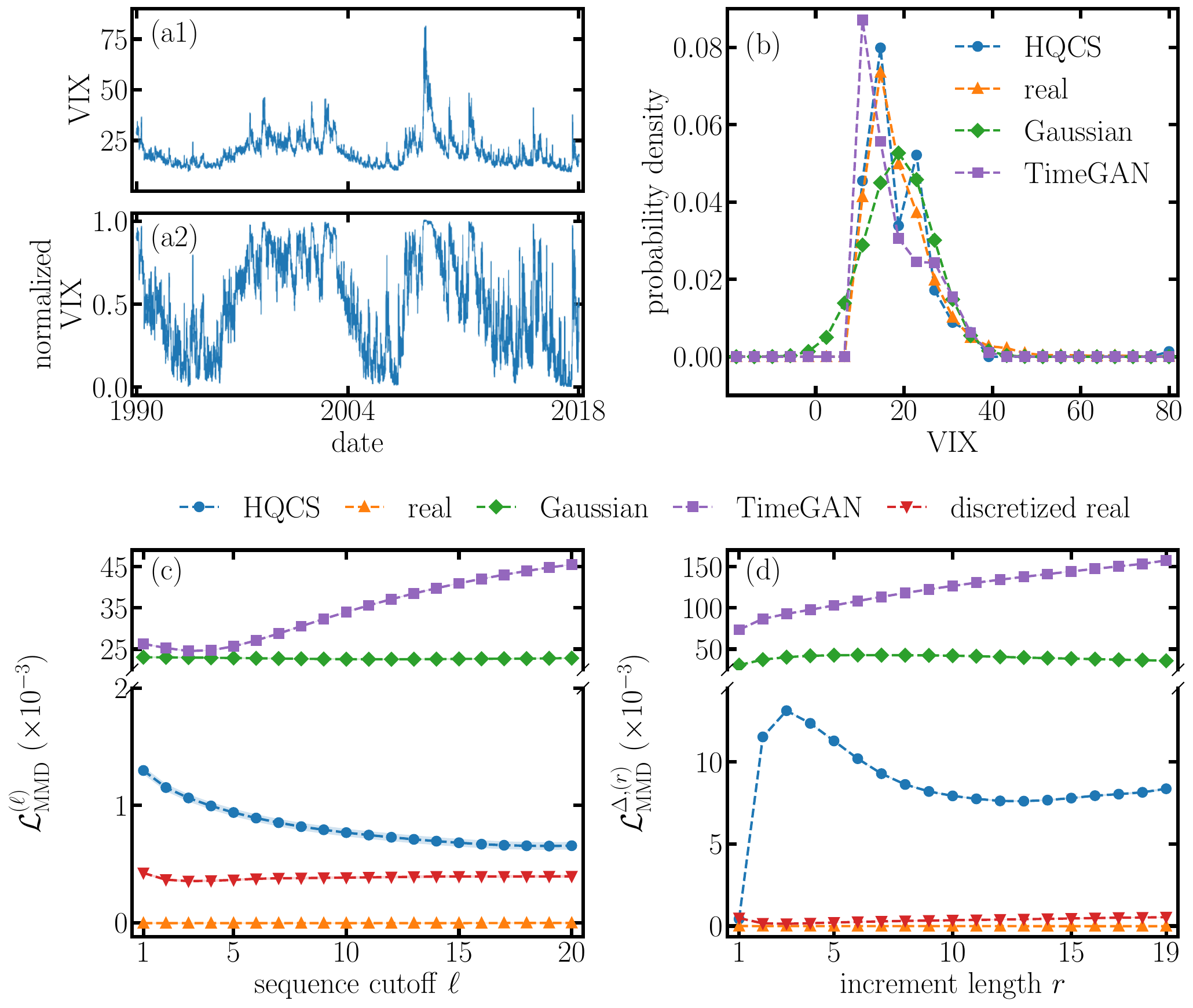}
    \caption{Generative learning of VIX index. (a1) Raw and (a2) normalized VIX index through time. We take the time period from 1990/09/18 to 2018/06/28. (b) Marginal distribution of HQCS (circles), real VIX (triangles), Gaussian model (diamonds) and TimeGAN (squares). 
    (c) MMD loss versus sequence cutoff length $\ell$ and (d) increment MMD versus increment length $r$ for HQCS (circles), real VIX (up triangles), discretized real VIX with precision of $1/31$ (down triangles), Gaussian model (diamonds) and TimeGAN (squares).
    Shaded areas show standard error of 20 iid sample sets.
    Here we choose $N_A=3, N_B=5$ with a total length of $T=20$ steps. In each step the circuit depth is $8$. The bandwidth of Gaussian kernel is $\sigma=6.28$ in (c) and $\sigma = 0.96$ in (d).
    }
     \label{fig:vix_main}
\end{figure}

\subsection{VIX index}
As a real-world benchmark, we consider the daily closing values of the CBOE Volatility Index (VIX), which summarizes the market-implied volatility of the S\&P 500 over an approximately 30-day horizon. Unlike the stationary Gaussian AR(1) process, the VIX is a positive and strongly non-Gaussian financial time series, as demonstrated by the marginal distribution in Fig.~\ref{fig:vix_main}b. Its dynamics combine persistent low-volatility regimes with intermittent large spikes associated with periods of market stress and temporal clustering of high VIX levels (see Fig.~\ref{fig:vix_main}a1). These features make the VIX a useful test for whether HQCS can learn local temporal patterns from realistic data beyond a toy model.

Unlike the AR(1) process, we have only a single observed VIX trajectory $\bm y$. 
To allow flexible distribution scope in learning, we instead sample length-$T$ slices from the real trajectory to form a data distribution. In this approach, the single VIX trajectory $\bm x$ can be split into data samples $\bm y_{s,T}=(y_s ,\ldots,y_{s+T-1})$ with the starting index $s$ ranging over all possible positions along the real trajectory. We utilize the empirical CDF to normalize the VIX index into the range of $\bm x \in [0, 1]^T$ (Fig.~\ref{fig:vix_main}(a2)). In the main text, we show results for training HQCS using a combined loss in Eq.~\eqref{eq:hqcs-training-loss} with $\ell=T, r=1$ and a properly chosen $\lambda \simeq 2.38$. At finite sample size, we find that adding the increment-space MMD term improves agreement with the target increment distribution while preserving comparable sequence-space MMD performance (see details in Appendix~\ref{app:numerical_detail}).

For performance benchmarking, besides the Gaussian model, we also consider a classical machine learning counterpart, TimeGAN~\cite{yoon2019time} (see details in Appendix~\ref{app:details}). For a fair comparison, we take approximately the same number of training parameters in HQCS and TimeGAN.
We begin with a qualitative comparison of the marginal data distribution in raw space, shown in Fig.~\ref{fig:vix_main}(b). The marginal distribution generated by HQCS (circles) approximately reproduces the right-skewed shape of the empirical VIX distribution (triangles), including its concentration at low values and extended upper tail. At the same time,
the Gaussian model (diamonds) generates data with symmetric tails around the mean value, while the data from TimeGAN (squares) exhibits a much sharper peak in the probability density function.

As the primary quantitative
diagnostic, Fig.~\ref{fig:vix_main}(c) shows the MMD $\calL_{\rm MMD}^{(\ell)}$ for various sequence cutoff lengths $\ell=1,\ldots,T$. The HQCS sequence-space MMD (circle) decreases with window length and gradually levels off at larger $\ell$, and furthermore, approaches the discretized VIX quantized on a $1/(2^{N_B}-1)=1/31$ grid (down triangles) for $\ell\ge 8$, which is the same precision as that of HQCS with $N_B = 5$. This indicates HQCS is capturing the real data almost up to the limitation of the discretization error from the quantum model being simulated. In contrast, the Gaussian model (diamonds) is approximately 18-35 times larger in MMD across all cutoff lengths $\ell$, and this parameter-matched TimeGAN (squares) provides even worse performance in MMD for long-sequence distribution matching at $\ell \ge 8$. In Appendix~\ref{app:correlation}, we use a permutation trick on generated trajectories to test whether HQCS learns the temporal dependence. Independently shuffling the generated values across trajectories at each time step increases the MMD by roughly two orders of magnitude, indicating that the learned temporal dependence substantially improves distribution matching beyond the marginals. In addition, we also benchmark the distribution of generated data from various models in increment-space in Fig.~\ref{fig:vix_main}(d). The performance advantage of HQCS (circles) over Gaussian (diamonds) and TimeGAN (squares) still preserves for the increment-space MMD, meanwhile the gap between HQCS and real data (triangles) call for improvements as a future direction.

While we focus on the training comparison in the main text, in Appendix~\ref{app:numerical_detail} we also present the out-of-distribution comparison with the data between 2018 and 2026, where the advantage of HQCS over other models survives. 

\begin{figure}[tbp]
    \centering
    \includegraphics[width=\columnwidth]{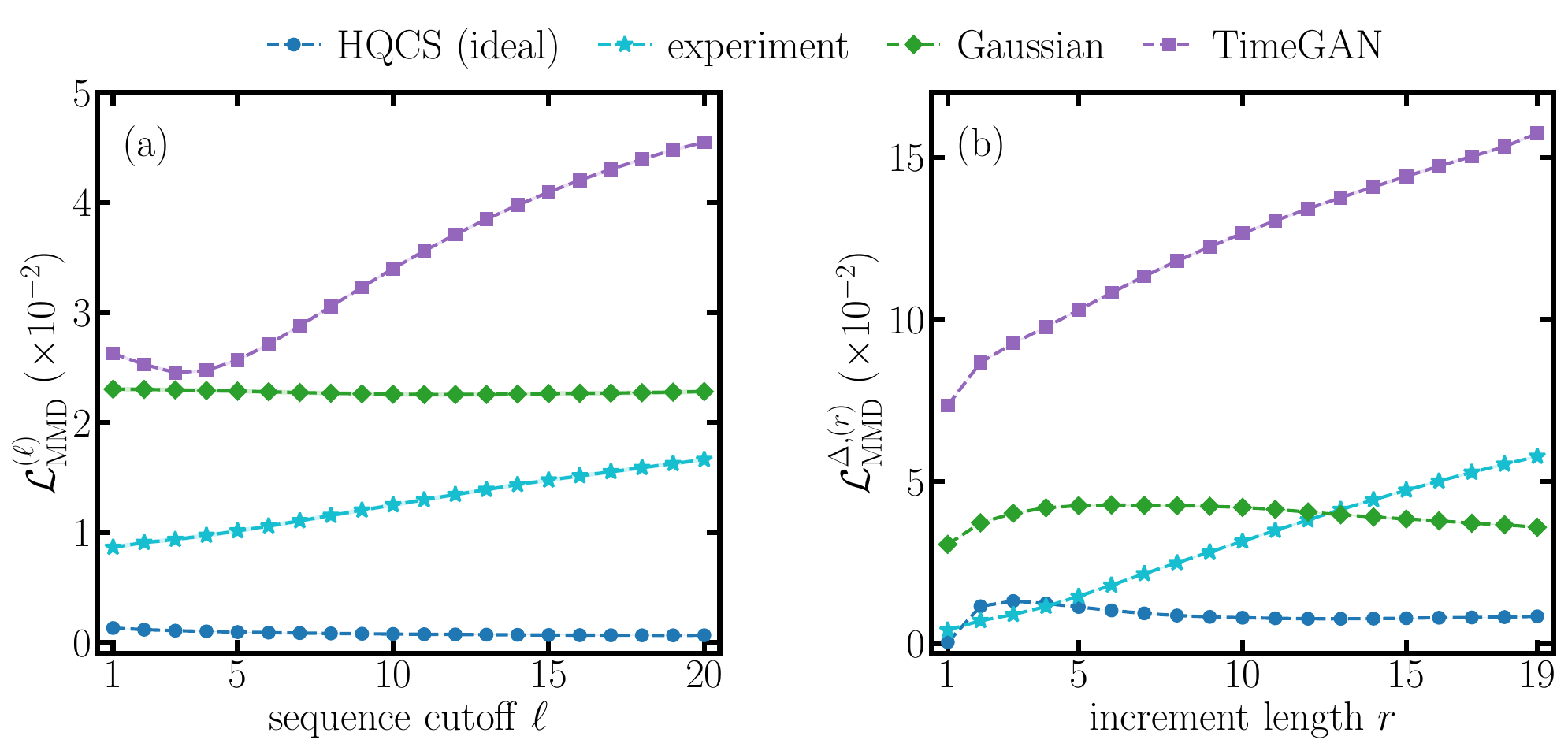}
    \caption{Experimental validation of HQCS for VIX index. (a) MMD loss versus sequence cutoff length $\ell$ and (b) increment MMD versus increment length $r$ for HQCS (ideal, circles), experiment (stars), Gaussian model (diamonds) and TimeGAN (squares).
    The experiment allocates the trained model (circles) onto the IBM Phoenix.
    Shaded areas show standard error of 20 disjoint sample sets, each containing $20{,}000$ trajectories. Note that errors are small and the region is not visible (see Appendix~\ref{app:numerical_detail}).
    Here we choose $N_A=3, N_B=5$ with a total length of $T=20$ steps. In each step the circuit depth is $8$. The bandwidth of Gaussian kernel is $\sigma=6.28$ in (a) and $\sigma=0.96$ in (b).
    }
    \label{fig:experimental_validation}
\end{figure}

\section{Validation on IBM Quantum Devices}
We next test the trained HQCS for VIX index on IBM Phoenix. We use the same trained model as in the ideal simulation, and allocate it on IBM Quantum devices to validate its performance. The experiment uses transpiler optimization level $3$, dynamical decoupling and gate twirling. 

In Fig.~\ref{fig:experimental_validation}(a), the experimental MMD loss increases with sequence cutoff length $\ell$, while remaining below the Gaussian model and TimeGAN. At $\ell=T=20$, the experimental loss (stars) is $(1.662\pm0.019)\times10^{-2}$, compared with $(2.280\pm0.021)\times10^{-2}$ for the Gaussian model (diamonds) and $(4.548\pm0.016)\times10^{-2}$ for TimeGAN (squares). Thus, hardware execution retains a reduction of approximately $27\%$ in full-sequence MMD loss relative to the Gaussian model. However, the experimental loss remains substantially larger than the ideal HQCS value (circles), indicating a substantial degradation from ideal sampling.
Fig.~\ref{fig:experimental_validation}(b) shows that the experimental increment MMD also increases with increment length $r$. The experimental results can maintain a better performance than both classical benchmarks for $r\le 12$, while it is only better than TimeGAN for longer $r\ge 13$. 
We expect that the degradation of performance can be due to both hardware noise and readout error, and an on-device training may improve the performance in practice. These device noises have a larger impact on the MMD of the increment, as increment concerns daily VIX changes, and more susceptible to perturbations.

\section{Discussion}
In this work, we introduce HQCS as a quantum generative framework
for learning time-series distributions. By combining mid-circuit
measurements with qubit reuse, HQCS generates trajectories of
arbitrary length using only a fixed number of physical qubits. Compared with the conventional approach based on a single deep parameterized circuit, which faces training difficulties, HQCS exhibits a variance of the MMD loss that decays only polynomially with sequence length $T$ at fixed register size. We establish this scaling analytically and observe the same behavior numerically across learning tasks with different target distributions. We extend this result on polynomial concentration to increment MMD. This absence of exponential cost concentration suggests favorable scaling of the gradient landscape, although a direct analysis of gradient fluctuations remains open~\cite{arrasmith2022equivalence}. 
We also establish a resource separation from recurrent models in completing the same generative task, with HQCS requiring fewer gates.
Our numerical simulations demonstrate HQCS on AR(1) processes with long temporal correlations. For the VIX dataset, HQCS achieves an MMD loss an order of magnitude lower than those obtained with the Gaussian model and TimeGAN. The sequence lengths considered here already exceed the simulation capabilities of the conventional quantum generative approach. Our implementations on IBM Phoenix show that the sequence MMD advantage over the tested classical benchmarks persists despite device noises, although the increment MMD advantage over the Gaussian model is limited to shorter increment lengths.

Several directions remain open. An important direction is to understand the interplay among the memory of the target process, the number of memory qubits, and the training-window length required for accurate learning. The information-scaling results in Ref.~\cite{zhang2026scaling} and recent studies of tradeoffs in state
ensemble resampling~\cite{mo2026measurement} may help address these questions. 
Although HQCS with sufficiently many qubits may eventually reach regimes beyond efficient classical simulation, its finite-memory structure motivates tensor-network descriptions that could guide architecture design and the development of quantum-inspired classical generative models.
Beyond the examples studied in this work, HQCS can also be applied to other natural stochastic processes and to generative tasks beyond time series. Extending the noisy-circuit framework of Ref.~\cite{zhang2025anticoncentrated} to trained HQCS could help quantify how noise affects the MMD loss. Reducing the experimental results' gap to ideal sampling and improving increment-space agreement remain important directions for practical applications.

\section*{Acknowledgements}
We acknowledge the use of IBM Quantum Credits via the IBM Quantum Startups Program for this work. The views expressed are those of the authors and do not reflect the official policy or position of IBM or the IBM Quantum Platform team.
We acknowledge GPT and Claude Code for assistance. The authors remain solely responsible for all results. B.Z. and Q.Z. acknowledge support from Gordon and Betty Moore Foundation and an unrestricted gift from Google during the manuscript preparation.

\section*{Author Contributions}
Q.Z. conceived the initial idea and designed the study. B.Z. and Q.Z. performed the analyses and proofs. B.Z. performed the simulations and generated the figures, with input from all authors. T.F. implemented HQCS on IBM Quantum devices. M.W.L contributed to the finance aspects of the work. B.Z. and Q.Z. wrote the manuscript, with inputs from all authors.

\section*{Competing Interests}
B.Z. and Q.Z. are the applicants and named inventors on a
U.S. provisional patent application
(No.~64/163,742, filed 09/28/2026; currently pending)
concerning the quantum algorithms described in this work;
M.L., T.F. and Q.Z. are co-founders of Applied Quantum Software Inc., which is developing commercial software incorporating the methods described in this work. B.Z. serves as an advisor to Applied Quantum Software Inc. B.Z., M.L., T.F. and Q.Z. hold equity in Applied Quantum Software Inc. M.L. is also the founder and managing partner of Triaxiom Capital, a quantitative investment management firm.

\section*{Code and Data Availability} Data for figures in this manuscript is available at \url{https://fabriq.appqsoft.com/hqcs} with free registration.
The HQCS implementation used in this study is owned by Applied Quantum Software Inc. and is not publicly distributed because of commercial and intellectual-property restrictions. Deployment of the HQCS implementation through the IBM Qiskit Functions Catalog is planned.

%

\onecolumngrid
\clearpage
\raggedbottom
\appendixmaterialtrue
\appendix

\tableofcontents

\section{Details on model implementations}
\label{app:details}

For HQCS, we adopt the hardware-efficient ansatz for each unitary block~\cite{kandala2017hardware}, where each layer applies independent single-qubit rotations in Pauli-Y and Z directions on every qubit, followed by nearest-neighbor CNOT gates in brickwork style. In Figs.~\ref{fig:ar1_main} and~\ref{fig:vix_main} in the main text, we choose circuit depth $L=N_A+N_B=8$, and the total number of parameters scales as $2(N_A+N_B)TL = 2560$.
Both memory and data registers were initialized in the trivial zero state at the beginning. At each time step, register $B$ was measured to obtain an integer outcome \(z_t\in\{0,\ldots,2^{N_B}-1\}\), which was mapped to the normalized output \(\tilde{x}_t=z_t/(2^{N_B}-1)\). 

For the Gaussian benchmark, we fit a multivariate Gaussian distribution to the raw trajectories from the training dataset. This benchmark matches the empirical mean and covariance in the space where it is fitted, providing a reference for assessing structure beyond these moments.

For a more advanced classical-model benchmark, we consider the TimeGAN~\cite{yoon2019time}. TimeGAN combines adversarial sequence generation with supervised learning of temporal transitions in a learned latent space. An embedder maps real sequences to latent trajectories, while a recovery network reconstructs the observations. A generator produces synthetic latent trajectories, a supervisor encourages their stepwise dynamics to match those of real trajectories, and a discriminator distinguishes real from generated latent sequences. The model is trained in three stages: reconstruction pretraining, supervisor pretraining, and joint adversarial training with reconstruction and temporal supervision.
For the VIX benchmark, the implementation comprises five three-layer GRU modules—embedder, recovery, generator, supervisor, and discriminator—with hidden dimension 5 and $2682$ trainable parameters, comparable to the $2560$ parameters of HQCS. Training uses alternating updates of the generator, discriminator, and autoencoding networks under adversarial, reconstruction, and temporal-supervision losses. Raw VIX values are scaled to $[0,1]$ internally using training-window minima and maxima; generated sequences are mapped back to VIX units before evaluation.

In Figs.~\ref{fig:ar1_main} and~\ref{fig:vix_main} in the main text, for each model, we independently generate $2\times 10^4$ samples as a single generated set and evaluate the corresponding MMD loss with an equal number of real trajectories. We repeat the sampling 20 times to obtain a precise estimate. 
We use bandwidth prefactors appropriate to the scale of each data representation. 
For CDF-normalized data on \([0,1]\), \(\sigma=0.3\) is comparable to the marginal standard deviation \(1/\sqrt{12}\simeq 0.289\) (Fig.~\ref{fig:ar1_main} in the main text). For raw VIX, \(\sigma=6.28\) is comparable to typical trajectory separations: the median pairwise Euclidean distance divided by \(\sqrt{T}\) is approximately \(6.12\), and for raw VIX in increment space, we choose $\sigma = 0.96$.

\section{Learning temporal dependence}
\label{app:correlation}

\begin{figure}[t]
    \centering
    \includegraphics[width=0.35\textwidth]{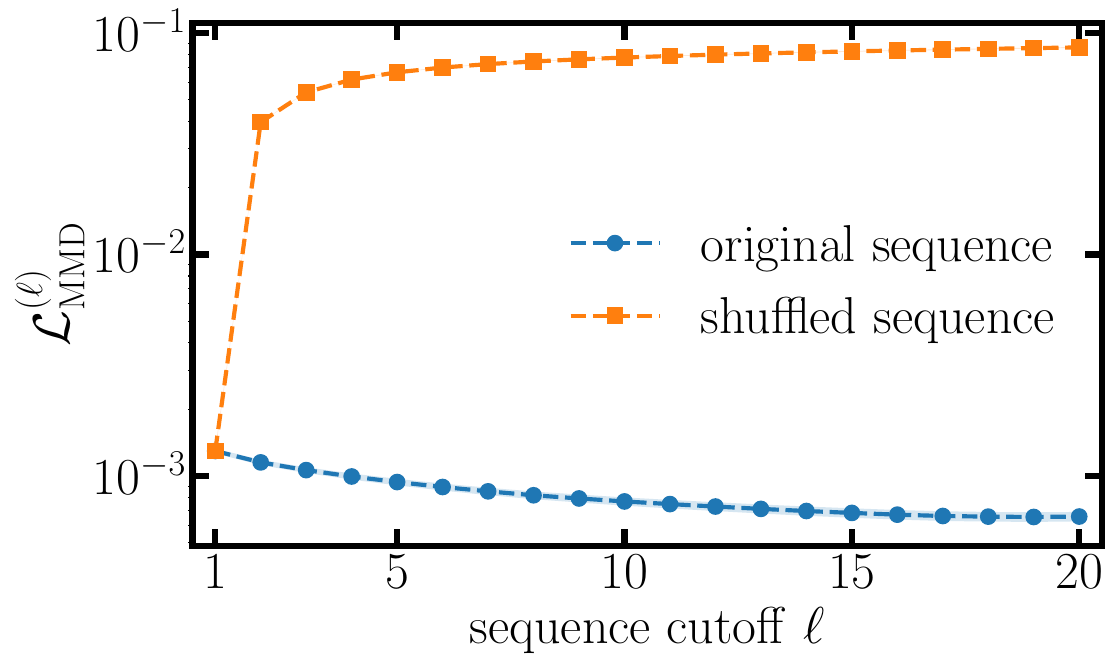}
    \caption{MMD loss between generated data from HQCS and real VIX data. Circles show original generated sequences while squares show shuffled sequences. The HQCS circuit is the same as in Fig.~\ref{fig:vix_main} in the main text.}
    \label{fig:shuffle}
\end{figure}

Here we adopt classical postprocessing to diagnose the learning of temporal correlation in HQCS beyond a simple matching of marginal distribution of VIX. Specifically, at each time step, we apply an independent permutation to the generated values across trajectories in the generated data set. This procedure preserves the empirical marginal distribution at every time step while disrupting the temporal dependence within trajectories. We then compare the MMD loss between the original or shuffled generated sequences and the real VIX data, shown in Fig.~\ref{fig:shuffle}. Unlike the original sequence (circles), the shuffled data yield MMD values roughly two orders of magnitude larger for any nontrivial correlation $\ell\ge 2$ beyond marginal distribution, which shows that the relatively low MMD in Fig.~\ref{fig:vix_main}(c) in the main text benefits substantially from the temporal dependence present in the true generated sequences, beyond their agreement on marginals.

\section{On related works}

Our setting differs from Ref.~\cite{rudolph2024trainability} in that
HQCS keeps the physical quantum register fixed while increasing the
generated sequence length through repeated measurement and qubit
reuse. The relevant scaling is therefore with temporal extent rather
than spatial circuit size. For independently Haar-random unitary
blocks, we establish the two-side inverse polynomial bounds for MMD variance in HQCS and rigorously prove the MMD variance for Haar-random unitary in conventional scheme to complete their analysis.
Thus, the
polynomial landscape does not rely on the temporal locality of target data, but instead originates from the sequential HQCS
structure together with the normalized kernel choice. The
dependence on the physical register size remains a separate
question.

HQCS is also distinct from several related recurrent quantum
architectures. Quantum hidden Markov models
(QHMMs)~\cite{markov2022implementation} repeatedly apply the same
unitary and are trained by matching sequence probabilities, whereas
HQCS allows non-identical unitaries across time steps and directly
learns the data distribution through a sample-based loss without a significant trainability cost. In the
hybrid quantum--classical recurrent model of
Ref.~\cite{xu2025hybrid}, expectation-value features obtained by
measuring the quantum system are classically fed forward to subsequent
steps. Other quantum recurrent neural networks
~\cite{takaki2021learning,li2023quantum} share circuit parameters
across time, thereby keeping the number of trainable parameters fixed
as the sequence length increases. In contrast, HQCS permits
independent parameters at different time steps, allowing the model to
represent nonstationary dynamics with provable less physical resources. Our analysis shows that the
resulting growth of parameter dimension with $T$ does not lead to
exponential concentration of the MMD landscape with temporal extent.

\section{Proof of Theorem~\ref{thm:haar_mmd} in the main text}
\label{app:mmd-haar-mean-variance}

We first state the formal theorem for the mean and variance of the MMD loss under Haar-random state sampling.
\begin{theorem}
[Formal version of Theorem~\ref{thm:haar_mmd} in the main text]
Let $\bm q_{\Theta}$ denote the measurement probabilities in a fixed complete orthonormal basis of an $N$-qubit Haar-random state $\ket{\psi_\Theta}$ satisfying at least 4-design, and $\bm p$ denote the target probability vector on the corresponding encoded trajectories, fixed independently of $\Theta$. The mean and variance of the MMD loss between $\bm q_{\Theta}$ and $\bm p$ are
\begin{align}
\mathbb E_\Theta[\mathcal L_{\rm MMD}]
&=
\bm p^T\bm K_{T,c}\bm p
+
\frac{t_c}{d(d+1)},\\
{\rm Var}_\Theta[\mathcal L_{\rm MMD}]
&=
 \frac{t_c^2+2f_c+6z_c}{(d)_4}
    -
    \frac{t_c^2}{(d)_2^2}
    +
    \frac{4\chi_P}{(d)_2}
   -
    \frac{8\gamma_P}{(d)_3},
\end{align}
where $(d)_r = d (d+1)\cdots (d+r-1)$ is the rising factorial. Here $\bm K_{T,c} = H \bm K_T H$ is the centered-kernel matrix with $H = \bI_d - \bm 1 \bm 1^T/d$. $t_c, f_c, z_c, \chi_P, \gamma_P$ depend only on $\bm p$ and the kernel function with detailed definitions in Eq.~\eqref{eq:mmd-haar-kernel-invariants}.
\end{theorem}

\begin{proof}
We assume that the random initialization ensemble for the model induces a
complex Haar-random state on this output space $\ket{\psi_\Theta}
    \sim
    {\rm Haar}(\mathbb C^d)$,
whose computational-basis measurement probabilities are
\begin{align}
    q_{\Theta,i}
    =
    Q_\Theta(\bm x^{(i)})
    =
    |\langle i|\psi_\Theta\rangle|^2.
    \label{eq:mmd-haar-born-probabilities}
\end{align}
Here $i=1,\ldots,d$ labels a computational-basis outcome with $d=2^N$, while its encoded trajectory $\bm x^{(i)}$ need not be binary, provided that the encoding is one-to-one.

Let $\bm K_T$ be the kernel matrix with entries $(\bm K_T)_{ij}
    =
    K_T(\bm x^{(i)},\bm x^{(j)})$.
The matrix is real, symmetric, and positive semidefinite. The MMD loss in Eq.~\eqref{eq:mmd_loss_main} in the main text can then be written as
\begin{align}
    \mathcal L_{\rm MMD}(\Theta)
    &=
    \bm p^T\bm K_T\bm p
    +
    \bm q_\Theta^T\bm K_T\bm q_\Theta
    -
    2\bm p^T\bm K_T\bm q_\Theta
    \nonumber\\
    &\coloneqq
    A_T+B_T(\Theta)-2C_T(\Theta),
    \label{eq:mmd-haar-ABC}
\end{align}
where $\bm p=(p_1,\ldots,p_d)^T$ and $p_i=P(\{\bm x^{(i)}\})$ is the target probability mass assigned to the $i$th encoded trajectory.
$A_T$ is independent of
the random model parameters, while $B_T$ and $C_T$ fluctuate across the
ensemble. 

Using Eq.~\eqref{eq:mmd-haar-born-probabilities}, we can write the product of probabilities as
\begin{align}
    \prod_{a=1}^r q_{\Theta, i_a} = \tr(\ketbra{\psi_\Theta}{\psi_\Theta}^{\otimes r} \otimes_{a=1}^r\ketbra{i_a}{i_a}).
\end{align}
Note that the ensemble-averaged moment of Haar-random states can be fully described by permutation operators on the multi-copy Hilbert space (see Refs.~\cite{zhang2025holographic, zhang2025anticoncentrated}):
\be
    \E_{\Theta} \ketbra{\psi_\Theta}{\psi_\Theta}^{\otimes r} = \frac{1}{(d)_r} \sum_{\pi \in S_r} \hat{\pi},
\ee
where $S_r$ is the symmetric group of $r$ elements and $(d)_r = d (d+1) \cdots (d+r-1)$ is the rising factorial. The Haar-averaged product of sampling probabilities is then
\begin{align}
    \E_\Theta \left[\prod_{a=1}^r q_{\Theta, i_a}\right] &= \frac{1}{(d)_r} \sum_{\pi \in S_r} \braket{i_1, \dots, i_r|\hat{\pi}|i_1,\dots, i_r} \nonumber\\
    &= \frac{1}{(d)_r} \sum_{\pi \in S_r} \prod_{a=1}^r \delta_{i_a, i_{\pi(a)}}.
    \label{eq:mmd-haar-permutation-moments}
\end{align}
In particular,
\begin{align}
    \mathbb E_\Theta[q_{\Theta,i}]
    &=
    \frac{1}{d},
    \label{eq:mmd-haar-first-moment}
    \\
    \mathbb E_\Theta[q_{\Theta,i}q_{\Theta,j}]
    &=
    \frac{1+\delta_{ij}}{d(d+1)}.
    \label{eq:mmd-haar-second-probability-moment}
\end{align}

\subsection{Mean of MMD loss}

Using Eq.~\eqref{eq:mmd-haar-second-probability-moment}, the model--model
term has ensemble mean
\begin{align}
    \mathbb E_\Theta[B_T]
    &=
    \sum_{i,j=1}^d
    (\bm K_T)_{ij}
    \mathbb E_\Theta
    [q_{\Theta,i}q_{\Theta,j}]
    \nonumber\\
    &=
    \frac{
    \bm 1^T\bm K_T\bm 1
    +
    {\rm Tr}(\bm K_T)
    }{d(d+1)},
    \label{eq:mmd-haar-mean-B}
\end{align}
where $\bm 1 = (1, \dots, 1)^T$.
The first moment in Eq.~\eqref{eq:mmd-haar-first-moment} similarly gives
\begin{align}
    \mathbb E_\Theta[C_T]
    =
    \frac{1}{d}\bm p^T\bm K_T\bm 1.
    \label{eq:mmd-haar-mean-C}
\end{align}
Substitution into Eq.~\eqref{eq:mmd-haar-ABC} yields the exact uncentered
expression
\begin{align}
    \mathbb E_\Theta[\mathcal L_{\rm MMD}]
    &=
    \bm p^T\bm K_T\bm p
    -
    \frac{2}{d}\bm p^T\bm K_T\bm 1
    +\frac{
    \bm 1^T\bm K_T\bm 1
    +
    {\rm Tr}(\bm K_T)
    }{d(d+1)}.
    \label{eq:mmd-haar-mean-uncentered}
\end{align}

The subsequent expressions become more transparent after removing the
uniform direction of the kernel. Define
\begin{align}
    H
    =
    \bI_d-\frac{1}{d}\bm 1\bm 1^T,
    \qquad
    \bm K_{T,c}
    =
    H\bm K_TH.
    \label{eq:mmd-haar-centered-kernel}
\end{align}
For any two probability vectors, $\bm 1^T(\bm p-\bm q_\Theta)=0$ due to normalization and hence
$H(\bm p-\bm q_\Theta)=\bm p-\bm q_\Theta$.  It follows that
\begin{align}
    (\bm p-\bm q_\Theta)^T
    \bm K_T
    (\bm p-\bm q_\Theta)
    =
    (\bm p-\bm q_\Theta)^T
    \bm K_{T,c}
    (\bm p-\bm q_\Theta).
    \label{eq:mmd-haar-centering-invariance}
\end{align}
Thus kernel centering does not change the MMD loss.  Since
$\bm K_{T,c}\bm 1=0$, Eq.~\eqref{eq:mmd-haar-mean-uncentered} can be reduced to
\begin{align}
    \mathbb E_\Theta[\mathcal L_{\rm MMD}]
    &=  \bm p^T\bm K_{T,c}\bm p
    - \frac{\bm 1^T \bm K_T \bm 1}{d^2} + \frac{\bm 1^T \bm K_T \bm 1}{d(d+1)} 
    + \frac{\Tr(\bm K_{T,c})}{d(d+1)} \nonumber\\
    &\quad + \frac{\bm 1^T \bm K_T \bm 1}{d^2(d+1)} \nonumber\\
    &=  \bm p^T\bm K_{T,c}\bm p + \frac{\Tr(\bm K_{T,c})}{d(d+1)}.
    \label{eq:mmd-haar-exact-mean}
\end{align}
Writing $\bm u=\bm 1/d$, so that $H\bm p = \bm p - \bm u$, the first term is
\begin{align}
    \bm p^T\bm K_{T,c}\bm p
    =
    (\bm p-\bm u)^T
    \bm K_T
    (\bm p-\bm u)
    =
    \calL_{\rm MMD}(P, \mathbb{U}),
\end{align}
where $\mathbb{U}$ denotes the uniform distribution on the corresponding finite-dimensional space.  The exact mean
therefore separates into the target--uniform discrepancy and the remaining fluctuation.
For the normalized kernels used in this work, $(\bm K_T)_{ii}=1$, so $\Tr(\bm K) = d$, and
Eq.~\eqref{eq:mmd-haar-exact-mean} can be written as
\begin{align}
    \mathbb E_\Theta[\mathcal L_{\rm MMD}]
    =
     \calL_{\rm MMD}(P, \mathbb{U})
    +
    \frac{1-\bm u^T\bm K_T\bm u}{d+1}.
    \label{eq:mmd-haar-normalized-mean}
\end{align}
The Haar fluctuation correction to the mean is therefore $O(d^{-1})$, while
the target--uniform discrepancy need not decrease with $d$.

\subsection{Variance of MMD loss}

For the remainder of the derivation, we use the centered kernel
$\bm K_{T,c}$ in the three terms of Eq.~\eqref{eq:mmd-haar-ABC}. By
Eq.~\eqref{eq:mmd-haar-centering-invariance}, their combination is the same
MMD loss, thus
\begin{align}
    A_T
    &=
    \bm p^T\bm K_{T,c}\bm p,
    \nonumber \\
    B_T(\Theta)
    &=
    \bm q_\Theta^T\bm K_{T,c}\bm q_\Theta,
    \nonumber \\
    C_T(\Theta)
    &=
    \bm p^T\bm K_{T,c}\bm q_\Theta.
    \label{eq:mmd-haar-centered-ABC}
\end{align}
For compactness, define the centered-kernel invariants
\begin{align}
    t_c
    &=
    {\rm Tr}(\bm K_{T,c}),
    &
    f_c
    &=
    {\rm Tr}(\bm K_{T,c}^2),
    \nonumber\\
    z_c
    &=
    \|\operatorname{diag}(\bm K_{T,c})\|_2^2,
    &
    \chi_P
    &=
    \bm p^T\bm K_{T,c}^2\bm p,
    \nonumber\\
    \gamma_P
    &=
    \operatorname{diag}(\bm K_{T,c})^T
    \bm K_{T,c}\bm p,
    \label{eq:mmd-haar-kernel-invariants}
\end{align}
where $\operatorname{diag}(\bm K_{T,c}) = ((\bm K_{T,c})_{11}, \dots, (\bm K_{T,c})_{dd})$ is a vector consisting of diagonal entries of the centered-kernel matrix.
The centering condition and Eqs.~\eqref{eq:mmd-haar-first-moment}--\eqref{eq:mmd-haar-second-probability-moment}
give
\begin{align}
    \mathbb E_\Theta[B_T]
    =
    \frac{t_c}{d(d+1)},
    \quad
    \mathbb E_\Theta[C_T]
    =0.
    \label{eq:mmd-haar-centered-first-moments}
\end{align}

We next calculate the three moments needed for the variance.  Let
$\bm b=\bm K_{T,c}\bm p$. Since $\bm 1^T\bm b=0$, contraction with the
second probability moment gives
\begin{align}
    \mathbb E_\Theta[C_T^2]
    &=
    \sum_{i,j=1}^d
    b_i b_j
    \mathbb E_\Theta[q_{\Theta,i}q_{\Theta,j}]
    \nonumber\\
    &=
    \frac{\bm b^T\bm b}{d(d+1)}
    =
    \frac{\chi_P}{d(d+1)}.
    \label{eq:mmd-haar-second-C}
\end{align}
For the mixed moment, the third-order case of
Eq.~\eqref{eq:mmd-haar-permutation-moments} gives
\begin{align}
    &\mathbb E_\Theta[B_TC_T]
    =
    \sum_{i,j,k=1}^d
    (\bm K_{T,c})_{ij}b_k
    \mathbb E_\Theta
    [q_{\Theta,i}q_{\Theta,j}q_{\Theta,k}]
    \nonumber\\ 
    &= \frac{1}{(d)_3} \sum_{i,j,k} (\bm K_{T,c})_{ij}b_k \left(
1+\delta_{ij}+\delta_{ik}+\delta_{jk}
+2\delta_{ij}\delta_{jk}\delta_{ki}
\right) \nonumber\\
&= \frac{2}{(d)_3} \sum_i (\bm K_{T,c})_{ii}b_i = \frac{2\gamma_P}{(d)_3}.
    \label{eq:mmd-haar-mixed-BC}
\end{align}
All terms containing a free contraction with $\bm 1$ vanish because
$\bm K_{T,c}\bm 1=0$ and the surviving third-moment contraction sets all three
indices equal.

\begin{table*}[t]
\centering
\renewcommand{\arraystretch}{1.35}
\begin{tabular}{c c l c c}
\hline\hline
Cycle type
& Number
& Permutations
& Individual contractions $C_\pi$
& Total contribution
\\
\hline

Identity
& $1$
& $e$
& $0$
& $0$
\\

Single transpositions
& $6$
&
$\begin{array}{c}
(12),\ (13),\ (14),\\
(23),\ (24),\ (34)
\end{array}$
& $0$
& $0$
\\

Double transpositions
& $3$
&
$\begin{array}{c}
(12)(34)\\
(13)(24)\\
(14)(23)
\end{array}$
&
$\begin{array}{c}
t_c^2\\
f_c\\
f_c
\end{array}$
& $t_c^2+2f_c$
\\

Three-cycles
& $8$
&
$\begin{array}{c}
(123),\ (132),\ (124),\ (142),\\
(134),\ (143),\ (234),\ (243)
\end{array}$
& $0$
& $0$
\\

Four-cycles
& $6$
&
$\begin{array}{c}
(1234),\ (1243),\ (1324),\\
(1342),\ (1423),\ (1432)
\end{array}$
& $z_c$ for each permutation
& $6z_c$
\\

\hline\hline
\end{tabular}
\caption{
Contributions of the $24$ permutations in $S_4$ to the fourth-order
contraction. Here
$C_\pi=\sum_{i,j,k,\ell}(\bm K_{T,c})_{ij}(\bm K_{T,c})_{k\ell} \prod_{a=1}^4 \delta_{i_a, i_{\pi(a)}}$.
The centered-kernel conditions
$K_{T,c}\bm 1=\bm 0$ and
$\bm 1^T K_{T,c}=\bm 0^T$ have already been applied.
\label{tab:haar-fourth-order-permutations}
}
\end{table*}

Finally, the fourth-order case of
Eq.~\eqref{eq:mmd-haar-permutation-moments} yields
\begin{align}
    \mathbb E_\Theta[B_T^2]
    &=
    \frac{t_c^2+2f_c+6z_c}{(d)_4},
    \label{eq:mmd-haar-second-B}
\end{align}
where the detailed counting for each of the 24 permutations can be found in Table~\ref{tab:haar-fourth-order-permutations}.

Expanding
\begin{align}
    \mathcal L_{\rm MMD}^2
    &=
    A_T^2+B_T^2+4C_T^2
    +2A_TB_T
    \nonumber\\
    &\quad-
    4A_TC_T-4B_TC_T
\end{align}
and substituting Eqs.~\eqref{eq:mmd-haar-centered-first-moments}--\eqref{eq:mmd-haar-second-B}
gives the exact ensemble-averaged second moment,
\begin{align}
    \mathbb E_\Theta[\mathcal L_{\rm MMD}^2]
    &=
    A_T^2
    +
    \frac{2A_Tt_c}{(d)_2}
    +
    \frac{t_c^2+2f_c+6z_c}{(d)_4}
    +
    \frac{4\chi_P}{(d)_2}
    -
    \frac{8\gamma_P}{(d)_3},
    \label{eq:mmd-haar-exact-second-moment}
\end{align}
As Eq.~\eqref{eq:mmd-haar-exact-mean} gives the mean as $\E_\Theta[\calL_{\rm MMD}] = A_T + t_c/(d)_2$, the variance of loss is then
\begin{align}
    {\rm Var}_\Theta[\mathcal L_{\rm MMD}]
    &= \mathbb E_\Theta[\mathcal L_{\rm MMD}^2] - \left(\E_\Theta[\calL_{\rm MMD}]\right)^2 \nonumber\\
    &= 
    \frac{t_c^2+2f_c+6z_c}{(d)_4}
    -
    \frac{t_c^2}{(d)_2^2}
    +
    \frac{4\chi_P}{(d)_2}
   -
    \frac{8\gamma_P}{(d)_3}.
    \label{eq:mmd-haar-exact-variance}
\end{align}

We finally perform a leading-order analysis of the variance in the limit of $d \to \infty$. For the kernels in Eq.~\eqref{KT_general} in the main text,
$0\leq K_T(\bm x,\widetilde{\bm x})\leq1$, which indicates $(\bm K_{T,c})_{ij} = O(1)$ and implies
\begin{align}
    t_c&=O(d),
    &
    f_c&=O(d^2),
    \\
    z_c&=O(d),
    &
    \chi_P&=O(d),
    \\
    \gamma_P&=O(d).
    \label{eq:mmd-haar-bounded-invariants}
\end{align}
The two terms proportional to $t_c^2$ must be combined
before power counting.  Their denominator difference satisfies
\begin{align}
    \frac{1}{d(d+1)(d+2)(d+3)}
    -
    \frac{1}{d^2(d+1)^2}
    =
    O(d^{-5}).
\end{align}
Consequently, Eq.~\eqref{eq:mmd-haar-exact-variance} gives the general bound
\begin{align}
    {\rm Var}_\Theta[\mathcal L_{\rm MMD}]
    =
    O(d^{-1}) = O(2^{-n}).
    \label{eq:mmd-haar-bounded-scaling}
\end{align}
\end{proof}

\section{Upper bound on variance of MMD of HQCS}
\label{app:hqcs-upper-organized}

This section proves a finite-length upper bound on the circuit-ensemble variance of the population squared-MMD loss.

\begin{definition}[Normalized exponential kernel]
\label{def:hqcs-exponential-kernel}
The local discrepancy $\varphi:[0,1]\to[0,\infty)$ is bounded,
with $\varphi(0)=0$ and
$\varphi_{\max}:=\sup_{0\le r\le1}\varphi(r)<\infty$.
The normalized exponential kernel is defined as
\begin{equation}
    K_T(x,y)
    :=\exp\!\left[-\frac1T\sum_{t=1}^T
                    \varphi(|x_t- y_t|)\right].
    \label{eq:hqcs-general-kernel}
\end{equation}
This expression is assumed positive semidefinite for every $T$.
Gaussian and Laplace kernels correspond to
$\varphi(r)=r^2/(2\sigma^2)$ and $\varphi(r)=r/\sigma$,
respectively, for fixed $\sigma>0$.
\end{definition}

The register dimensions are fixed, with $d_A\ge 2$, $d_B\ge2$, and
$d=d_Ad_B$. The initial state $\rho_0$ and the target distribution
$P_T$ on $[0,1]^T$ are independent of the circuit. The blocks
$\mathbf U=(U_1,\ldots,U_T)$ are independent Haar-random unitaries on
$AB$. The measurement projectors are $\Pi_z=I_A\otimes\ketbra{z}$,
and the conditional post-measurement states for both memory $A$ and data $B$ are retained at the next step.
The readout and the normalized exponential kernel are those of
Definition~\ref{def:hqcs-exponential-kernel}.

The population squared-MMD loss can be written as
\begin{equation}
    \mathcal L_{\rm MMD}(\mathbf U)
      =A_T+B_T(\mathbf U)-2C_T(\mathbf U),
    \label{eq:hqcs-mmd-decomposition}
\end{equation}
where
\begin{align}
    A_T&:=\mathbb E_{X,X'\sim P_T}[K_T(X,X')],
       \nonumber\\
    B_T(\mathbf U)&:=\mathbb E_{Y,Y' \sim Q_{\mathbf U}}[K_T(Y,Y')],
       \nonumber\\
    C_T(\mathbf U)&:=\mathbb E_{X\sim P_T}
                       \mathbb E_{Y \sim Q_{\mathbf U}}[K_T(X,Y)].
\end{align}
Throughout this section,
$\mathbb E_{\mathbf U}$ averages the entire circuit ensemble,
$\mathbb E_{U\sim\mathrm{Haar}}$ averages one independent unitary, and
$\mathbb E_{Y \sim Q_{\mathbf U}}$ averages measurement outcomes for a fixed
circuit. Two trajectories in $\mathbb E_{Y,Y'\sim Q_{\mathbf U}}$ are
independent conditional on the same circuit. The target expectation
$\mathbb E_{X\sim P_T}$ is separate from these averages. The norms
$\|R\|_1$ and $\|R\|_2$ denote the trace and Hilbert--Schmidt norms.

The transfer maps below encode the trajectory probabilities and kernel weights and are also used to define the constants in the upper bound.

\begin{definition}[Measurement instruments and weighted transfer maps]
\label{def:hqcs-organized-transfer}
For an operator $R$ on $AB$, the measurement instrument with outcome $z$ and its
nonselective channel are
\begin{align}
    \mathcal J_{U,z}(R)
    :=\Pi_zURU^\dagger\Pi_z,
    \qquad
    \mathcal T_U(R)
    :=\sum_z\mathcal J_{U,z}(R).
    \label{eq:hqcs-organized-instrument}
\end{align}
For a fixed trajectory $x\in[0,1]^T$, the weighted one-copy map at
position $t$ is
\begin{equation}
    \mathcal M_{U,t}(R;x)
    :=\sum_z e^{-\varphi(|x_t-f(z)|)/T}\mathcal J_{U,z}(R).
    \label{eq:hqcs-organized-map-one}
\end{equation}
For an operator $R$ on $(AB)^{\otimes2}$, the weighted two-copy map is
\begin{equation}
    \mathcal N_{U,t}(R)
    :=\sum_{z,z'}e^{-\varphi(|f(z)-f(z')|)/T}
        (\mathcal J_{U,z}\otimes\mathcal J_{U,z'})(R).
    \label{eq:hqcs-organized-map-two}
\end{equation}
\end{definition}

With these definitions, we have the main result as follows.
\begin{theorem}[Formal upper bound of MMD variance of HQCS]
\label{thm:hqcs-organized-upper}
Under the fixed-register HQCS setup above, with independent Haar-random
blocks and the normalized exponential kernel in
Definition~\ref{def:hqcs-exponential-kernel}, every initial density matrix
$\rho_0$ and every circuit-independent target distribution $P_T$
satisfy, for every integer $T\ge1$,
\begin{equation}
    \operatorname{Var}_{\mathbf U}[\mathcal L_{\rm MMD}]
       \le\frac{\varphi_{\max}^2(C_2+2C_1)^2}{2T}.
    \label{eq:hqcs-organized-upper}
\end{equation}
The constants are explicitly
\begin{align}
    C_1&=1+\frac{2\sqrt d}{1-\sqrt{\kappa_1}},
    \qquad
    C_2=1+\frac{2d}{1-\sqrt{\kappa_2}},
    \qquad
    \kappa_1=\frac{d_A^2d_B-1}{d^2-1},
    \nonumber\\
    \kappa_2&=
    \max_{\substack{R=R^\dagger,\ \operatorname{Tr}R=0\\\|R\|_2=1}}
       \mathbb E_{U\sim{\rm Haar}}
       \|(\mathcal T_U\otimes\mathcal T_U)(R)\|_2^2<1.
\end{align}
The strict inequality for $\kappa_2$ is established in
Lemma~\ref{lem:hqcs-organized-mixing}. All constants are independent of
$T$, $\rho_0$, and $P_T$.
\end{theorem}

The proof uses the Efron--Stein inequality to reduce the variance
bound to the mean-square effect of replacing one circuit block.
Lemma~\ref{lem:hqcs-organized-weighting} shows that kernel normalization
makes each local weighting deficit at most $\varphi_{\max}/T$.
Lemma~\ref{lem:hqcs-organized-mixing} supplies mean-square contraction
for the subsequent unweighted evolution, preventing these deficits
from accumulating with the sequence length.
Lemma~\ref{lem:hqcs-organized-influence} combines the two estimates into
a uniform replacement bound, which yields the stated $1/T$ variance
bound after summation over the blocks.

The transfer representation makes these replacement effects accessible.
The maps in Definition~\ref{def:hqcs-organized-transfer} retain the Born
probabilities and kernel weights in their output traces; their outputs
are not normalized.
For a fixed circuit, the probability for obtaining outcome $z_1, \dots, z_T$ becomes
\begin{equation}
    \Pr(z_1,\ldots,z_T\mid\mathbf U)
    =\operatorname{Tr}\!\left[
       \mathcal J_{U_T,z_T}\circ\cdots\circ
       \mathcal J_{U_1,z_1}(\rho_0)\right].
    \label{eq:hqcs-organized-branch}
\end{equation}
Consequently, the factorization of the kernel in
Eq.~\eqref{eq:hqcs-general-kernel} gives
\begin{align}
    \mathbb E_{Y\sim Q_{\mathbf U}}K_T(x,Y)
    &=\sum_{z_1,\ldots,z_T}
       \left[\prod_{t=1}^T e^{-\varphi(|x_t-f(z_t)|)/T}\right]
       \Pr(z_1,\ldots,z_T\mid\mathbf U)\nonumber\\
    &=\operatorname{Tr}\!\left[
       \bigl(\mathcal M_{U_T,T}(\cdot;x)\circ\cdots\circ
       \mathcal M_{U_1,1}(\cdot;x)\bigr)(\rho_0)\right].
    \label{eq:hqcs-organized-fixed-cross}
\end{align}
Further averaging the target trajectory in this identity yields the cross term in MMD loss
\begin{equation}
    C_T(\mathbf U)
    =\mathbb E_{X\sim P_T}\operatorname{Tr}\!\left[
       \bigl(\mathcal M_{U_T,T}(\cdot;X)\circ\cdots\circ
       \mathcal M_{U_1,1}(\cdot;X)\bigr)(\rho_0)\right].
    \label{eq:hqcs-organized-cross}
\end{equation}
For the model--model term, the conditional independence of the two
measurement histories makes their joint branch probability the product
of the two probabilities in Eq.~\eqref{eq:hqcs-organized-branch}.
Expanding both histories and collecting the local weights in
Eq.~\eqref{eq:hqcs-organized-map-two} therefore gives
\begin{equation}
    B_T(\mathbf U)
    =\operatorname{Tr}\!\left[
       \mathcal N_{U_T,T}\circ\cdots\circ
       \mathcal N_{U_1,1}(\rho_0^{\otimes2})\right].
    \label{eq:hqcs-organized-model}
\end{equation}

The replacement estimate first requires control of the local difference
between weighted and unweighted evolution. The next lemma bounds this
difference on arbitrary Hermitian inputs, including the signed
perturbations generated by replacing a block.

\begin{lemma}[Weak kernel weighting]
\label{lem:hqcs-organized-weighting}
The maps $\mathcal M_{U,t}(\cdot;x)$ and $\mathcal N_{U,t}$ are completely
positive and trace nonincreasing. In completely positive order,
\begin{align}
    0&\preceq\mathcal T_U-\mathcal M_{U,t}(\cdot;x)
       \preceq\epsilon_T\mathcal T_U,\nonumber\\
    0&\preceq\mathcal T_U\otimes\mathcal T_U-\mathcal N_{U,t}
       \preceq\epsilon_T(\mathcal T_U\otimes\mathcal T_U),
    \label{eq:hqcs-organized-cp-order}
\end{align}
where $\epsilon_T = \varphi_{\rm max}/T$ with $\varphi_{\max} = \sup_{0\le r\le1} \varphi(r)$.
Consequently, for every Hermitian input $R$ on the corresponding space,
\begin{align}
    \|(\mathcal T_U-\mathcal M_{U,t}(\cdot;x))(R)\|_1
       &\le\epsilon_T\|R\|_1,\nonumber\\
    \|(\mathcal T_U\otimes\mathcal T_U-\mathcal N_{U,t})(R)\|_1
       &\le\epsilon_T\|R\|_1.
    \label{eq:hqcs-organized-weight-norm}
\end{align}
The weighted maps themselves are trace-norm contractive on Hermitian
inputs.
\end{lemma}

\begin{proof}
Every local weight obeys
\begin{equation}
    0\le1-e^{-\varphi(r)/T}
      \le\frac{\varphi(r)}T
      \le\epsilon_T,\qquad 0\le r\le1.
    \label{eq:hqcs-organized-local-weight}
\end{equation}
The instruments in Eq.~\eqref{eq:hqcs-organized-instrument} are
completely positive, and $\mathcal T_U$ is trace preserving.
The difference in the one-copy case is
\begin{equation}
    \mathcal T_U-\mathcal M_{U,t}(\cdot;x)
    =\sum_z[1-e^{-\varphi(|x_t-f(z)|)/T}]\mathcal J_{U,z}.
\end{equation}
Both this map and its difference from $\epsilon_T\mathcal T_U$ have
nonnegative instrument coefficients by
Eq.~\eqref{eq:hqcs-organized-local-weight}. This proves the first
completely positive order relation in
Eq.~\eqref{eq:hqcs-organized-cp-order}. The same coefficient argument
with $\mathcal J_{U,z}\otimes\mathcal J_{U,z'}$ proves the second.
The weighted maps are completely positive and dominated by their
trace-preserving unweighted counterparts, so they are trace
nonincreasing.

The trace-norm consequences also hold for inputs that are not positive.
For a positive map $\Phi$ satisfying
$\operatorname{Tr}\Phi(S)\le a\operatorname{Tr}S$ for all $S\succeq0$,
the positive and negative parts of a Hermitian $R=R_+-R_-$ give
\begin{align}
    \|\Phi(R)\|_1
    &\le\|\Phi(R_+)\|_1+\|\Phi(R_-)\|_1\nonumber\\
    &=\operatorname{Tr}\Phi(R_+)+\operatorname{Tr}\Phi(R_-)\nonumber\\
    &\le a(\operatorname{Tr}R_++\operatorname{Tr}R_-)
     =a\|R\|_1.
    \label{eq:hqcs-organized-positive-map-norm}
\end{align}
Application to the difference maps in
Eq.~\eqref{eq:hqcs-organized-cp-order}, with $a=\epsilon_T$, proves
Eq.~\eqref{eq:hqcs-organized-weight-norm}. Application to the weighted
maps with $a=1$ proves their contractivity. 
\end{proof}

The local deficit bound alone would allow the effects of later positions
to accumulate with the sequence length. The next lemma provides
mean-square contraction for one and two replicas, making their
propagated contributions summable.

\begin{lemma}[Haar mixing for one and two replicas]
\label{lem:hqcs-organized-mixing}
For every deterministic traceless Hermitian $R$ on $AB$,
\begin{equation}
    \mathbb E_{U\sim\mathrm{Haar}}\|\mathcal T_U(R)\|_2^2
       =\kappa_1\|R\|_2^2,\qquad
    \kappa_1=\frac{d_A^2d_B-1}{d^2-1}<1.
    \label{eq:hqcs-organized-mix-one}
\end{equation}
For every deterministic traceless Hermitian $R$ on $(AB)^{\otimes2}$,
\begin{equation}
    \mathbb E_{U\sim\mathrm{Haar}}
    \|(\mathcal T_U\otimes\mathcal T_U)(R)\|_2^2
       \le\kappa_2\|R\|_2^2,
    \label{eq:hqcs-organized-mix-two}
\end{equation}
where the dimension-dependent constant is
\begin{equation}
    \kappa_2:=
    \max_{\substack{R=R^\dagger,\ \operatorname{Tr}R=0\\ \|R\|_2=1}}
       \mathbb E_{U\sim\mathrm{Haar}}
       \|(\mathcal T_U\otimes\mathcal T_U)(R)\|_2^2<1.
    \label{eq:hqcs-organized-kappa-two}
\end{equation}
\end{lemma}

\begin{proof}
The swap operator $F$ on two copies satisfies
$\operatorname{Tr}[F(A\otimes B)]=\operatorname{Tr}(AB)$.
For deterministic operators $A,B$ on the two-copy space, the
two-copy Haar integral is
\begin{equation}
    \mathbb E_{U\sim\mathrm{Haar}}
      \operatorname{Tr}\!\left[A U^{\otimes2}B(U^\dagger)^{\otimes2}\right]
    =
      \frac{\operatorname{Tr}A\,\operatorname{Tr}B
            +\operatorname{Tr}(AF)\operatorname{Tr}(BF)}{d^2-1}
      -\frac{\operatorname{Tr}A\,\operatorname{Tr}(BF)
            +\operatorname{Tr}(AF)\operatorname{Tr}B}{d(d^2-1)}.
    \label{eq:hqcs-organized-haar-two}
\end{equation}
Indeed, Haar invariance places the averaged operator
$\mathbb E_{U\sim\mathrm{Haar}}
U^{\otimes2}B(U^\dagger)^{\otimes2}$ in the span of $I$ and $F$.
Preservation of its trace and swap trace determines their coefficients
and gives Eq.~\eqref{eq:hqcs-organized-haar-two}.

For the one-copy channel, orthogonality of the measurement projectors
and the swap identity yield
\begin{equation*}
\begin{aligned}
    \|\mathcal T_U(R)\|_2^2
    &=\sum_z\operatorname{Tr}
       [\Pi_zURU^\dagger\Pi_zURU^\dagger]\\
    &=\sum_z\operatorname{Tr}
       [F\Pi_z^{\otimes2}U^{\otimes2}R^{\otimes2}(U^\dagger)^{\otimes2}].
\end{aligned}
\end{equation*}
Here $\operatorname{Tr}\Pi_z=d_A$,
$\operatorname{Tr}\Pi_z^{\otimes2}=d_A^2$, and
$\operatorname{Tr}(F\Pi_z^{\otimes2})=d_A$.
Since $\operatorname{Tr}R=0$, the remaining input traces are
$\operatorname{Tr}R^{\otimes2}=0$ and
$\operatorname{Tr}(FR^{\otimes2})=\|R\|_2^2$.
Eq.~\eqref{eq:hqcs-organized-haar-two} therefore gives
\begin{align}
    \mathbb E_{U\sim\mathrm{Haar}}\|\mathcal T_U(R)\|_2^2
    &=\sum_z\left[
          \frac{d_A^2}{d^2-1}-\frac{d_A}{d(d^2-1)}
       \right]\|R\|_2^2\nonumber\\
    &=\frac{d_A^2d_B-1}{d^2-1}\|R\|_2^2.
\end{align}
This proves Eq.~\eqref{eq:hqcs-organized-mix-one}.

For two replicas, $\mathcal D_B\otimes\mathcal D_B$ is a
Hilbert--Schmidt orthogonal projection. Consequently,
\begin{align}
    \|(\mathcal T_U\otimes\mathcal T_U)(R)\|_2
    =\|(\mathcal D_B\otimes\mathcal D_B)
          [(U\otimes U)R(U^\dagger\otimes U^\dagger)]\|_2 \le\|R\|_2.
    \label{eq:hqcs-organized-two-nonexpansion}
\end{align}
It remains to show that Haar averaging makes the inequality strict
for every nonzero traceless Hermitian input, and that the strict
contraction has a uniform constant. The following
argument excludes such an input.

Suppose $\E_{U \sim \rm Haar}\|(\mathcal T_U\otimes\mathcal T_U)(R)\|_2^2 = \|R\|_2^2$ for a nonzero traceless Hermitian $R$, then Eq.~\eqref{eq:hqcs-organized-two-nonexpansion} directly implies that $\|(\mathcal T_U\otimes\mathcal T_U)(R)\|_2^2 - \|R\|_2^2 = 0$ for every $U$ on Haar measure, 
Equality under an orthogonal projection means that the input already
lies in its range. Hence the conjugated matrix
$X_U:=(U\otimes U)R(U^\dagger\otimes U^\dagger)$ is block diagonal
in both measured registers as $(\mathcal D_B\otimes\mathcal D_B)(X_U)=X_U$, which implies
\begin{align}
    [X_U,\Pi_z\otimes I]=0,
    \quad [X_U,I\otimes\Pi_z]=0
       \qquad\text{for every }z.
\end{align}
Conjugating these relations back gives
\begin{align}
    0=(U^\dagger\otimes U^\dagger)
       [X_U,\Pi_z\otimes I](U\otimes U)
       &=[R,(U^\dagger\Pi_zU)\otimes I], \nonumber \\
    0=(U^\dagger\otimes U^\dagger)
       [X_U,I\otimes\Pi_z](U\otimes U)
       &=[R,I\otimes(U^\dagger\Pi_zU)], \qquad \text{for every }U,z.
       \label{eq:hqcs-organized-two-commutators}
\end{align}
The requirement \emph{for every $U,z$} is essential here.
The rotated projectors $U^\dagger\Pi_zU$ range over all rank-$d_A$
projectors on $AB$, and their linear span contains every Hermitian
operator. Indeed, since $1\le d_A<d$, any two orthogonal rank-one
projectors can be extended by the same $(d_A-1)$-dimensional
subspace; subtracting the resulting rank-$d_A$ projectors gives
their rank-one difference. These differences span the traceless
Hermitian operators, and $\sum_z\Pi_z=I$ supplies the identity.
The first relation in
Eq.~\eqref{eq:hqcs-organized-two-commutators} therefore says that $R$
commutes with every operator on the first replica, so
$R=I\otimes R_2$. The second says that $R_2$ commutes with every
operator on the second replica, so $R_2$ is a scalar identity.
Consequently, $R \propto I_2$ and since $\Tr(R) = 0$, we have $R=0$.
This contradicts the chosen nonzero input. Thus every nonzero
traceless Hermitian $R$ satisfies
\begin{equation}
    \mathbb E_{U\sim\mathrm{Haar}}
       \|(\mathcal T_U\otimes\mathcal T_U)(R)\|_2^2
       <\|R\|_2^2.
\end{equation}

Finally, a uniform constant follows from compactness.
On the traceless Hermitian unit sphere $\|R\|_2=1$, the averaged
squared norm is a continuous function and is strictly less than $1$
at every point. This sphere is compact in finite dimension, so the
maximum in Eq.~\eqref{eq:hqcs-organized-kappa-two} is attained at
some input and is itself strictly less than $1$.
Thus $\kappa_2<1$, and rescaling an arbitrary nonzero $R$ to unit
norm gives Eq.~\eqref{eq:hqcs-organized-mix-two}.
The zero input satisfies that bound trivially.
\end{proof}

The remaining task is to translate local weighting and mixing into a
bound on the actual loss terms. The next lemma controls the mean-square
change in both circuit-dependent terms when a single block is replaced.

\begin{lemma}[Single-block influence]
\label{lem:hqcs-organized-influence}
The circuit $\mathbf U^{(s)}$ is obtained from $\mathbf U$ by replacing
$U_s$ with an independent Haar-random copy $U_s'$ and retaining all
other blocks. The local weighting bound is
$\epsilon_T=\varphi_{\max}/T$, as defined in
Lemma~\ref{lem:hqcs-organized-weighting}.
For every $T\ge1$, $1\le s\le T$, every initial density
matrix $\rho_0$, and every target $P_T$ independent of the circuit,
\begin{align}
    \sqrt{\mathbb E_{\mathbf U,U_s'}
       \left[(B_T(\mathbf U)-B_T(\mathbf U^{(s)}))^2\right]}
       &\le C_2\epsilon_T,\nonumber\\
    \sqrt{\mathbb E_{\mathbf U,U_s'}
       \left[(C_T(\mathbf U)-C_T(\mathbf U^{(s)}))^2\right]}
       &\le C_1\epsilon_T,
    \label{eq:hqcs-organized-influence}
\end{align}
where
\begin{equation}
    C_1=1+\frac{2\sqrt d}{1-\sqrt{\kappa_1}},
    \qquad
    C_2=1+\frac{2d}{1-\sqrt{\kappa_2}}.
    \label{eq:hqcs-organized-influence-constants}
\end{equation}
The constants $\kappa_1,\kappa_2$ are those of
Lemma~\ref{lem:hqcs-organized-mixing}, and the expectation includes
both the original circuit and the independent replacement.
\end{lemma}

\begin{proof}
The fixed-target expectation in
Eq.~\eqref{eq:hqcs-organized-fixed-cross} and the model--model
expectation in Eq.~\eqref{eq:hqcs-organized-model} have the same
transfer structure. Their common proof uses the following assignments
\begin{equation}
\begin{array}{c|ccccc}
 &\mathfrak T_t&\mathfrak M_t&\rho&n&\kappa\\ \hline
 \text{one replica}
 &\mathcal T_{U_t}&\mathcal M_{U_t,t}(\cdot;x)&\rho_0&d&\kappa_1\\
 \text{two replicas}
 &\mathcal T_{U_t}\otimes\mathcal T_{U_t}
 &\mathcal N_{U_t,t}&\rho_0^{\otimes2}&d^2&\kappa_2
\end{array}
\label{eq:hqcs-organized-assignments}
\end{equation}
These symbols are confined to this proof. In the one-replica case,
$x$ remains deterministic until the final target average.
The maps with a prime at position $s$ correspond to $U_s'$.
Empty compositions below are identity maps.

The weighted prefix and the perturbation created at position $s$ are
\begin{align}
    R_{s-1}&:=\mathfrak M_{s-1}\circ\cdots\circ\mathfrak M_1(\rho),
       \qquad R_{s-1}\succeq0,\quad\operatorname{Tr}R_{s-1}\le1,\nonumber\\
    \Delta_s&:=(\mathfrak M_s-\mathfrak M_s')(R_{s-1}).
    \label{eq:hqcs-organized-prefix-perturbation}
\end{align}
The positivity and trace bound follow from
Lemma~\ref{lem:hqcs-organized-weighting}. The terminal difference is
$\operatorname{Tr}[
\mathfrak M_T\circ\cdots\circ\mathfrak M_{s+1}(\Delta_s)]$.
The weighting deficits
$\mathfrak E_t:=\mathfrak T_t-\mathfrak M_t$ and
$\mathfrak E_s':=\mathfrak T_s'-\mathfrak M_s'$ satisfy
Eq.~\eqref{eq:hqcs-organized-cp-order}. Trace preservation of the
unweighted maps consequently gives
\begin{align}
    \operatorname{Tr}\Delta_s
    &=\operatorname{Tr}\mathfrak E_s'(R_{s-1})
       -\operatorname{Tr}\mathfrak E_s(R_{s-1}),\nonumber\\
    |\operatorname{Tr}\Delta_s|
    &\le\epsilon_T\operatorname{Tr}R_{s-1}\le\epsilon_T,\nonumber\\
    \|\Delta_s\|_1
    &\le\operatorname{Tr}\mathfrak M_s(R_{s-1})
       +\operatorname{Tr}\mathfrak M_s'(R_{s-1})\le2.
    \label{eq:hqcs-organized-delta-bounds}
\end{align}
In the trace estimate, each of the two nonnegative deficit traces
belongs to the same interval
$[0,\epsilon_T\operatorname{Tr}R_{s-1}]$ and their difference is thus bounded
by the length of that interval.

The difference in replacement admits a decomposition into identity and traceless component as
\begin{equation}
    \Delta_s=\frac{\operatorname{Tr}\Delta_s}{n}I+\Delta_s^\circ,
\end{equation}
where $\Delta_s^\circ$ represents the traceless component as $\Tr(\Delta_s^\circ) = 0$ and $n$ is the dimension specified in Eq.~\eqref{eq:hqcs-organized-assignments}.
The future weighted composition is positive and trace nonincreasing.
Its scalar contribution therefore satisfies
\begin{align}
    \left|\operatorname{Tr}\!\left[
       \mathfrak M_T\circ\cdots\circ\mathfrak M_{s+1}
       \left(\frac{\operatorname{Tr}\Delta_s}{n}I\right)\right]\right|
    &\le\frac{|\operatorname{Tr}\Delta_s|}{n}\operatorname{Tr}I
     \le\epsilon_T.
    \label{eq:hqcs-organized-scalar}
\end{align}
Hilbert--Schmidt orthogonality of the decomposition also gives a universal bound on 2-norm of the traceless component as
\begin{align}
    \|\Delta_s^\circ\|_2^2
    &=\|\Delta_s\|_2^2-\frac{|\operatorname{Tr}\Delta_s|^2}{n}
      \le\|\Delta_s\|_2^2
      \le\|\Delta_s\|_1^2\le4.
    \label{eq:hqcs-organized-traceless-size}
\end{align}

The traceless component has zero terminal trace under the unweighted
future composition, because every $\mathfrak T_t$ preserves trace.
Its contribution under the weighted evolution can thus be isolated
by the telescoping identity
\begin{align}
    \mathfrak M_T\circ\cdots\circ\mathfrak M_{s+1}
     -\mathfrak T_T\circ\cdots\circ\mathfrak T_{s+1} 
    =-\sum_{t=s+1}^T
       \mathfrak M_T\circ\cdots\circ\mathfrak M_{t+1}
       \circ\mathfrak E_t
       \circ\mathfrak T_{t-1}\circ\cdots\circ\mathfrak T_{s+1}.
    \label{eq:hqcs-organized-telescope}
\end{align}
With
$Q_t:=\mathfrak T_{t-1}\circ\cdots\circ
\mathfrak T_{s+1}(\Delta_s^\circ)$ for $s+1\le t\le T$,
Eq.~\eqref{eq:hqcs-organized-telescope} becomes
\begin{align}
    \operatorname{Tr}[
       \mathfrak M_T\circ\cdots\circ
       \mathfrak M_{s+1}(\Delta_s^\circ)]
    &=-\sum_{t=s+1}^T\operatorname{Tr}[
       \mathfrak M_T\circ\cdots\circ
       \mathfrak M_{t+1}(\mathfrak E_t(Q_t))].
    \label{eq:hqcs-organized-trace-telescope}
\end{align}
Every $Q_t$ is Hermitian and traceless. The contractivity and deficit
bound in Lemma~\ref{lem:hqcs-organized-weighting}, followed by
$\|R\|_1\le\sqrt n\|R\|_2$, imply the pointwise estimate
\begin{align}
    \left|\operatorname{Tr}[
       \mathfrak M_T\circ\cdots\circ
       \mathfrak M_{t+1}(\mathfrak E_t(Q_t))]\right|
    &\le\|\mathfrak E_t(Q_t)\|_1
     \le\epsilon_T\|Q_t\|_1
     \le\epsilon_T\sqrt n\|Q_t\|_2.
    \label{eq:hqcs-organized-summand}
\end{align}
The perturbation $\Delta_s^\circ$ depends only on
$U_1,\ldots,U_s,U_s'$. Conditional application of
Lemma~\ref{lem:hqcs-organized-mixing} at each unweighted step gives
\begin{align}
    &\mathbb E_{U_{t-1},\dots,U_{s+1}}
       \left[\|Q_t\|_2^2\mid U_1,\ldots,U_s,U_s'\right]
    \le\kappa^{t-s-1}\|\Delta_s^\circ\|_2^2,\nonumber\\
    &\sqrt{\mathbb E_{\mathbf U,U_s'}\|Q_t\|_2^2}
    \le\kappa^{(t-s-1)/2}
       \sqrt{\mathbb E_{\mathbf U,U_s'}\|\Delta_s^\circ\|_2^2}
     \le2\kappa^{(t-s-1)/2}.
    \label{eq:hqcs-organized-propagation}
\end{align}
where the second line averages the
prefix and replacement as well and uses
Eq.~\eqref{eq:hqcs-organized-traceless-size}.

The triangle inequality for square-integrable scalar random variables,
applied to Eq.~\eqref{eq:hqcs-organized-trace-telescope}, now yields
\begin{align}
    \sqrt{\mathbb E_{\mathbf U,U_s'}
       \left|\operatorname{Tr}[
         \mathfrak M_T\circ\cdots\circ
         \mathfrak M_{s+1}(\Delta_s^\circ)]\right|^2}
    &\quad\le\sum_{t=s+1}^T
       \sqrt{\mathbb E_{\mathbf U,U_s'}
       \left|\operatorname{Tr}[
         \mathfrak M_T\circ\cdots\circ
         \mathfrak M_{t+1}(\mathfrak E_t(Q_t))]\right|^2}
       \nonumber\\
    &\quad\le\epsilon_T\sqrt n\sum_{t=s+1}^T
       \sqrt{\mathbb E_{\mathbf U,U_s'}\|Q_t\|_2^2}
       \nonumber\\
    &\quad\le2\epsilon_T\sqrt n
       \sum_{r=0}^{T-s-1}\kappa^{r/2}
     =2\epsilon_T\sqrt n
       \frac{1-\kappa^{(T-s)/2}}{1-\sqrt\kappa}.
    \label{eq:hqcs-organized-finite-sum}
\end{align}
The first inequality follows by expanding the squared sum and bounding
each cross term by Cauchy--Schwarz. The second and third inequalities use
Eqs.~\eqref{eq:hqcs-organized-summand}
and~\eqref{eq:hqcs-organized-propagation}, respectively.
Combining Eq.~\eqref{eq:hqcs-organized-finite-sum} with the scalar
bound in Eq.~\eqref{eq:hqcs-organized-scalar} gives
\begin{align}
    \sqrt{\mathbb E_{\mathbf U,U_s'}
       \left|\operatorname{Tr}[
         \mathfrak M_T\circ\cdots\circ
         \mathfrak M_{s+1}(\Delta_s)]\right|^2}
    \le\epsilon_T\left[
          1+\frac{2\sqrt n(1-\kappa^{(T-s)/2})}{1-\sqrt\kappa}
       \right]
       \le\epsilon_T\left[1+\frac{2\sqrt n}{1-\sqrt\kappa}\right].
    \label{eq:hqcs-organized-terminal}
\end{align}
The two-replica assignment in
Eq.~\eqref{eq:hqcs-organized-assignments}, together with
Eq.~\eqref{eq:hqcs-organized-model}, proves the $B_T$ bound in
Eq.~\eqref{eq:hqcs-organized-influence}.
The one-replica assignment and
Eq.~\eqref{eq:hqcs-organized-fixed-cross} give, for every fixed $x$,
\begin{equation}
    \mathbb E_{\mathbf U,U_s'}
      \left|
        \mathbb E_{Y\sim Q_{\mathbf U}}K_T(x,Y)
        -\mathbb E_{Y\sim Q_{\mathbf U^{(s)}}}K_T(x,Y)
      \right|^2
      \le C_1^2\epsilon_T^2.
    \label{eq:hqcs-organized-fixed-influence}
\end{equation}
The two inner expectations average the trajectories of their
respective fixed circuits. Target averaging in
Eq.~\eqref{eq:hqcs-organized-cross} then gives
\begin{align}
    \mathbb E_{\mathbf U,U_s'}
       |C_T(\mathbf U)-C_T(\mathbf U^{(s)})|^2 
    & =
      \mathbb E_{\mathbf U,U_s'}
      \left|\mathbb E_{X\sim P_T}
       \left[
        \mathbb E_{Y\sim Q_{\mathbf U}}K_T(X,Y)
        -\mathbb E_{Y\sim Q_{\mathbf U^{(s)}}}K_T(X,Y)
       \right]\right|^2\nonumber\\
    &\quad\le
      \mathbb E_{X\sim P_T}\mathbb E_{\mathbf U,U_s'}
       \left|
        \mathbb E_{Y\sim Q_{\mathbf U}}K_T(X,Y)
        -\mathbb E_{Y\sim Q_{\mathbf U^{(s)}}}K_T(X,Y)
       \right|^2
      \le C_1^2\epsilon_T^2.
    \label{eq:hqcs-organized-target-jensen}
\end{align}
This proves the $C_T$ bound and completes the lemma.
\end{proof}

The replacement estimates now complete the proof of
Theorem~\ref{thm:hqcs-organized-upper} by summing over the independent blocks.

\begin{proof}[Proof of Theorem~\ref{thm:hqcs-organized-upper}]
Independence of $U_1,\ldots,U_T$ gives the Efron--Stein inequality
\begin{equation}
    \operatorname{Var}_{\mathbf U}\mathcal L_{\rm MMD}
    \le\frac12\sum_{s=1}^T
       \mathbb E_{\mathbf U,U_s'}
       \left[(\mathcal L_{\rm MMD}(\mathbf U)
        -\mathcal L_{\rm MMD}(\mathbf U^{(s)}))^2\right].
    \label{eq:hqcs-organized-efron-stein}
\end{equation}
The replacement circuits and their coupling are those of
Lemma~\ref{lem:hqcs-organized-influence}.
Since $A_T$ in Eq.~\eqref{eq:hqcs-mmd-decomposition} is independent of
the circuit, the loss difference contains only the differences of
$B_T$ and $-2C_T$. The scalar mean-square triangle inequality and
Eq.~\eqref{eq:hqcs-organized-influence} imply
\begin{align}
    \sqrt{\mathbb E_{\mathbf U,U_s'}
       \left[(\mathcal L_{\rm MMD}(\mathbf U)
        -\mathcal L_{\rm MMD}(\mathbf U^{(s)}))^2\right]}
   &\le
      \sqrt{\mathbb E_{\mathbf U,U_s'}
       \left[(B_T(\mathbf U)-B_T(\mathbf U^{(s)}))^2\right]}
      +2\sqrt{\mathbb E_{\mathbf U,U_s'}
       \left[(C_T(\mathbf U)-C_T(\mathbf U^{(s)}))^2\right]} \nonumber\\
   &\le(C_2+2C_1)\epsilon_T.
    \label{eq:hqcs-organized-loss-influence}
\end{align}
Substitution into Eq.~\eqref{eq:hqcs-organized-efron-stein}, with
$\epsilon_T=\varphi_{\max}/T$ as defined in
Lemma~\ref{lem:hqcs-organized-weighting}, yields
\begin{align}
    \operatorname{Var}_{\mathbf U}\mathcal L_{\rm MMD}
    \le\frac12\sum_{s=1}^T
       (C_2+2C_1)^2\epsilon_T^2 
    =\frac{\varphi_{\max}^2(C_2+2C_1)^2}{2T}.
\end{align}
All constants depend only on the fixed dimensions and the local
discrepancy bound. Neither target averaging nor boundary
nondegeneracy is required.
\end{proof}

\section{Lower bound on the variance of MMD of HQCS}
\label{app:hqcs-curvature-lower}

This section proves an explicit lower bound for the variance of MMD loss in randomly-initialized HQCS.
The initial state $\rho_0$ is a fixed pure state, including the all-zero
state. The circuit blocks $\mathbf U=(U_1,\ldots,U_T)$ are independent
Haar-random unitaries on $AB$. The measurement projectors are still
$\Pi_z=I_A\otimes\ketbra{z}$, and the conditional post-measurement state
is retained. The target distribution $P_T$ on $[0,1]^T$ is independent
of the circuit and otherwise arbitrary.
The readout $f$ and the kernel $K_T$ satisfy
Definition~\ref{def:hqcs-exponential-kernel}. The only additional requirement for the readout map is that there exists two distinct measurement outcome $i \neq j$, such that
\begin{equation}
    c:=\varphi(|f(i)-f(j)|)>0.
    \label{eq:hqcs-curvature-separation}
\end{equation}
For discrepancies satisfying $\varphi(r)>0$ for every $r>0$,
including the Gaussian and Laplace choices, this condition
holds for every nonconstant readout map.

The first-step probabilities and the future-averaged loss are
\begin{align}
    q_z(U_1)&:=\operatorname{Tr}(\Pi_zU_1\rho_0U_1^\dagger),
    \nonumber\\
    \Phi_T(U_1)&:=
    \mathbb E_{U_2,\ldots,U_T}
    [\mathcal L_{\rm MMD}(U_1,\ldots,U_T)\mid U_1].
    \label{eq:hqcs-curvature-first-prob}
\end{align}
For $T=1$, the future expectation is the identity operation.
The future kernel retains the normalization by the full length:
\begin{equation}
    K_{{\rm future},T}(x, y)
    :=\exp\!\left[-\frac1T\sum_{t=2}^T
                    \varphi(|x_t-  y_t|)\right].
    \label{eq:hqcs-curvature-future-kernel}
\end{equation}
For fixed orthogonal pure states $\psi,\chi$ as input for second temporal step on $AB$,
$Q^\psi_{\mathbf U_{2:T}}$ and $Q^\chi_{\mathbf U_{2:T}}$ denote their
future trajectory distributions under the same subsequent circuit.
We also introduce two coefficients used throughout the proof, defined by
\begin{align}
    s_T&:=
    \mathbb E_{U_2,\ldots,U_T}
    \mathbb E_{Y,Y'\sim Q^\psi_{\mathbf U_{2:T}}}
    [K_{{\rm future},T}(Y,Y')],
    \nonumber\\
    r_T&:=
    \mathbb E_{U_2,\ldots,U_T}
    \mathbb E_{Y\sim Q^\psi_{\mathbf U_{2:T}},
                         \ Y'\sim Q^\chi_{\mathbf U_{2:T}}}
    [K_{{\rm future},T}(Y,Y')].
    \label{eq:hqcs-curvature-future-coefficients}
\end{align}
Their values are independent of the chosen orthogonal pair, as proved
below. At $T=1$, the empty future kernel is $1$, so $s_1=r_1=1$.

With these definitions, we have the main result as follows.
\begin{theorem}[Formal lower bound of the MMD variance of HQCS]
\label{thm:hqcs-curvature-lower}
For fixed-register HQCS model with every unitary independently sampled from Haar-random unitary and normalized exponential kernel in Definition~\ref{def:hqcs-exponential-kernel}, assuming the readout map satisfies Eq.~\eqref{eq:hqcs-curvature-separation}, the variance of MMD loss satisfies
\begin{align}
    \operatorname{Var}_{\mathbf U}[\mathcal L_{\rm MMD}]
    &\ge
    \frac{4d_A^2(d_A+1)}
         {(2d_A+1)d(d+1)(d+2)(d+3)}
    r_T^2(1-e^{-c/T})^2,
    \label{eq:hqcs-curvature-finite-lower}\\
    &\ge \frac{C_-}{T^2},
    \label{eq:hqcs-curvature-polynomial-lower}
\end{align}
where 
\begin{align}
    C_-:=
    \frac{4d_A^2(d_A+1)c^2e^{-2\varphi_{\max}}}
         {(2d_A+1)d(d+1)(d+2)(d+3)}>0.
\end{align}
The coefficient $r_T$ is defined in
Eq.~\eqref{eq:hqcs-curvature-future-coefficients}.
The constant $C_-$ is independent of $T$ and $P_T$.
\end{theorem}

The proof uses the total variance formula to retain the variance of the
loss after averaging all blocks except $U_1$.
Lemma~\ref{lem:hqcs-curvature-representation} writes this conditional
mean as a quadratic model term plus a linear target term, with a
strictly positive quadratic coefficient.
Lemma~\ref{lem:hqcs-curvature-dirichlet} supplies the Haar-induced
probability distribution and the required moments, while
Lemma~\ref{lem:hqcs-curvature-variance} shows that the linear target
term cannot cancel the variance of the centered quadratic component.
Combining these estimates gives the finite-length lower bound;
the bounds on $r_T$ and $1-e^{-c/T}$ then yield the stated $T^{-2}$
scaling.

We first identify the dependence of the future-averaged loss on the
first-step probabilities. The next lemma isolates the quadratic
coefficient that will retain variance independently of the target.

\begin{lemma}[Exact future average with a positive quadratic coefficient]
\label{lem:hqcs-curvature-representation}
Under the conditions above, $\Phi_T(U_1)$ has the exact representation
\begin{align}
    \Phi_T(q)
    &=A_T
      +r_T\sum_{z,z'}q_zq_{z'}
            e^{-\varphi(|f(z)-f(z')|)/T}
      +(s_T-r_T)\sum_zq_z^2
      -2\sum_zq_zc_{z,T},
    \label{eq:hqcs-curvature-polynomial}
\end{align}
where $q_z$ is the measurement probability of outcome $z$ at first step, and the deterministic target coefficients are
\begin{equation}
    c_{z,T}:=
    \mathbb E_{X\sim P_T}\!
    \left[
    e^{-\varphi(|X_1-f(z)|)/T}
    \prod_{t=2}^T
    \left(\frac1{d_B}\sum_b
                 e^{-\varphi(|X_t-f(b)|)/T}\right)
    \right].
    \label{eq:hqcs-curvature-target-coefficients}
\end{equation}
The future coefficients satisfy
\begin{equation}
    s_T\ge r_T\ge e^{-\varphi_{\max}(T-1)/T}.
    \label{eq:hqcs-curvature-future-bounds}
\end{equation}
With all other probabilities and $q_i+q_j$ fixed, substitution of
$q_i=(q_i+q_j)w$, $q_j=(q_i+q_j)(1-w)$ into
Eq.~\eqref{eq:hqcs-curvature-polynomial} gives the exact quadratic polynomial
\begin{equation}
    \Phi_T(q)=a_T(q_{\rm rest})+b_T(q_{\rm rest})w
       +D_{ij,T}(q_i+q_j)^2w^2,
    \label{eq:hqcs-curvature-uncentered-polynomial}
\end{equation}
where $q_{\rm rest}=(q_z)_{z\ne i,j}$ and
$q_i+q_j=1-\sum_{z\ne i,j}q_z$. Here $a_T,b_T$ are the constant
and linear coefficients, derived explicitly below. The quadratic
coefficient is $(q_i+q_j)^2D_{ij,T}$, where
\begin{align}
    D_{ij,T}
    &:=2s_T-2r_Te^{-c/T}
      =2(s_T-r_T)+2r_T(1-e^{-c/T})
      \nonumber\\
    &\ge2r_T(1-e^{-c/T})>0,
    \label{eq:hqcs-curvature-quadratic-coefficient}
\end{align}
where $c = \varphi(|f(i) - f(j)|)$.
\end{lemma}

\begin{proof}
Throughout this proof, $U_1$ is fixed and only the future blocks
are averaged. For one trajectory and a fixed target sequence $x$,
the first outcome $z$ has probability $q_z$. The subsequent
weighted sum over outcomes can be evaluated one step at a time.
The single-copy weighted map from
Definition~\ref{def:hqcs-organized-transfer} and the first Haar moment
give, for every deterministic unnormalized input $R$ and fixed $x$,
\begin{align}
    \mathbb E_{U_t\sim{\rm Haar}}[\mathcal M_{U_t,t}(R;x)]
    &=\sum_b e^{-\varphi(|x_t-f(b)|)/T}
         \Pi_b\,\mathbb E_{U_t\sim{\rm Haar}}[U_tRU_t^\dagger]\Pi_b
      \nonumber\\
    &=\frac{\operatorname{Tr}R}{d}
       \sum_b e^{-\varphi(|x_t-f(b)|)/T}\Pi_b,
      \nonumber\\
    \operatorname{Tr}
    \mathbb E_{U_t\sim{\rm Haar}}[\mathcal M_{U_t,t}(R;x)]
    &=\left(\frac1{d_B}\sum_b
         e^{-\varphi(|x_t-f(b)|)/T}\right)\operatorname{Tr}R.
    \label{eq:hqcs-curvature-weighted-haar}
\end{align}
The last equality uses $\operatorname{Tr}\Pi_b=d_A$.
Independence of the future blocks permits this identity to be used
successively, conditional on the preceding blocks.
For each possible first result and every fixed target sequence $x$,
Eq.~\eqref{eq:hqcs-organized-fixed-cross} therefore yields
\begin{align}
    &\mathbb E_{U_2,\ldots,U_T}
      \mathbb E_{Y\sim Q_{\mathbf U}(\cdot\mid z_1=z)}[K_T(x,Y)]
    =
      e^{-\varphi(|x_1-f(z)|)/T}
      \prod_{t=2}^T\left(\frac1{d_B}\sum_b
                    e^{-\varphi(|x_t-f(b)|)/T}\right).
    \label{eq:hqcs-curvature-fixed-target}
\end{align}
The initial conditional state has unit trace.
Summation over the first result and then target averaging give
\begin{align}
    \mathbb E_{U_2,\ldots,U_T}\mathbb E_{X\sim P_T}
       \mathbb E_{Y\sim Q_{\mathbf U}}[K_T(X,Y)]
    &=\mathbb E_{X\sim P_T}
      \sum_zq_z e^{-\varphi(|X_1-f(z)|)/T}
      \prod_{t=2}^T\left(\frac1{d_B}\sum_b
                    e^{-\varphi(|X_t-f(b)|)/T}\right)
      \nonumber\\
    &=\sum_zq_zc_{z,T}.
    \label{eq:hqcs-curvature-cross-average}
\end{align}

For two independently-sampled trajectories with the same circuit,
the first outcome pair $(z,z')$ has probability $q_zq_{z'}$.
The pure initial state is written as $\rho_0=\ketbra{\psi_0}$.
Conditional on a first result $z$ with $q_z>0$, the retained state is
the projector onto
$\Pi_zU_1\ket{\psi_0}/\sqrt{q_z}$.
Two equal results therefore produce identical pure states in the two
replicas whereas distinct results produce orthogonal pure states.

For $T\ge2$, any two ordered orthogonal pairs of pure states are
related by a common unitary $V$. Replacing the pair by its image under
$V$ replaces the next gate by $U_2V$ in both replicas.
Right invariance of Haar measure leaves the future average unchanged.
The same argument applies to a pair of identical pure states.
Consequently, the two cases have precisely the coefficients
$s_T,r_T$ in Eq.~\eqref{eq:hqcs-curvature-future-coefficients},
regardless of the directions of the conditional states.
The statement is immediate for an empty future.
Summation over equal and distinct first outcomes therefore gives
\begin{align}
    \mathbb E_{U_2,\ldots,U_T}
       \mathbb E_{Y,Y'\sim Q_{\mathbf U}}[K_T(Y,Y')]
    &=s_T\sum_zq_z^2
      +r_T\sum_{z\ne z'}q_zq_{z'}
                   e^{-\varphi(|f(z)-f(z')|)/T}
      \nonumber\\
    &=r_T\sum_{z,z'}q_zq_{z'}
                   e^{-\varphi(|f(z)-f(z')|)/T}
      +(s_T-r_T)\sum_zq_z^2.
    \label{eq:hqcs-curvature-model-average}
\end{align}
The diagonal kernel entries equal $1$ because $\varphi(0)=0$.
Eqs.~\eqref{eq:hqcs-curvature-model-average}
and~\eqref{eq:hqcs-curvature-cross-average}, together with
Eq.~\eqref{eq:hqcs-mmd-decomposition}, prove
Eq.~\eqref{eq:hqcs-curvature-polynomial}.

The kernel in Eq.~\eqref{eq:hqcs-curvature-future-kernel} is a
positive-semidefinite kernel: it is the restriction of $K_T$ to
sequences with the same fixed first coordinate.
For each future circuit, its squared-MMD between the two future
distributions is nonnegative. Their averaged self-kernel terms both
equal $s_T$ by Haar invariance. Thus
\begin{align}
    2(s_T-r_T)
    &=\mathbb E_{U_2,\ldots,U_T}
      \left[
      \operatorname{MMD}_{K_{{\rm future},T}}^2
      (Q^\psi_{\mathbf U_{2:T}},Q^\chi_{\mathbf U_{2:T}})
      \right]\ge0.
    \label{eq:hqcs-curvature-psd}
\end{align}
Every future kernel value is at least
$e^{-\varphi_{\max}(T-1)/T}$, which proves the remaining bound in
Eq.~\eqref{eq:hqcs-curvature-future-bounds}.

To obtain Eq.~\eqref{eq:hqcs-curvature-uncentered-polynomial}, split
Eq.~\eqref{eq:hqcs-curvature-polynomial} according to whether neither,
one, or both indices belong to $\{i,j\}$:
\begin{align}
    \Phi_T(q)
    &=A_T+r_T\sum_{\substack{z\ne i,j\\z'\ne i,j}}
        q_zq_{z'}e^{-\varphi(|f(z)-f(z')|)/T}
      +(s_T-r_T)\sum_{z\ne i,j}q_z^2
      -2\sum_{z\ne i,j}q_zc_{z,T} \nonumber\\
    &\quad+2r_T\sum_{z\ne i,j}q_z
       \left[q_ie^{-\varphi(|f(i)-f(z)|)/T}
             +q_je^{-\varphi(|f(j)-f(z)|)/T}\right]
    +s_T(q_i^2+q_j^2)+2r_Te^{-c/T}q_iq_j
      -2(q_ic_{i,T}+q_jc_{j,T}).
\end{align}
The factor $2$ in the mixed terms accounts for the two orders of
each outcome pair. The first line is independent of $w$.
Substitution of $q_i=(q_i+q_j)w$ and $q_j=(q_i+q_j)(1-w)$
into the remaining terms gives
\begin{align}
    &s_T(q_i^2+q_j^2)+2r_Te^{-c/T}q_iq_j
    =(q_i+q_j)^2
       \left[s_T\left(w^2+(1-w)^2\right)+2r_Te^{-c/T}w(1-w)\right]\nonumber\\
    &\quad=(q_i+q_j)^2
       \left[s_T-D_{ij,T}w+D_{ij,T}w^2\right],\\
    &2r_T\sum_{z\ne i,j}q_z
       \left[q_ie^{-\varphi(|f(i)-f(z)|)/T}
             +q_je^{-\varphi(|f(j)-f(z)|)/T}\right] \nonumber\\
    &=2r_T(q_i+q_j)\sum_{z\ne i,j}q_z
       \left[
       e^{-\varphi(|f(j)-f(z)|)/T}
       +w\left(e^{-\varphi(|f(i)-f(z)|)/T}
                 -e^{-\varphi(|f(j)-f(z)|)/T}\right)\right],\\
    &2(q_ic_{i,T}+q_jc_{j,T})=2(q_i+q_j)
       \left[c_{j,T}+w(c_{i,T}-c_{j,T})\right].
\end{align}
Collecting powers of $w$ proves
Eq.~\eqref{eq:hqcs-curvature-uncentered-polynomial}, with
\begin{align}
    a_T(q_{\rm rest})
    &=A_T+r_T\sum_{\substack{z\ne i,j\\z'\ne i,j}}
          q_zq_{z'}e^{-\varphi(|f(z)-f(z')|)/T}
      +(s_T-r_T)\sum_{z\ne i,j}q_z^2
      -2\sum_{z\ne i,j}q_zc_{z,T} +s_T(q_i+q_j)^2
      \nonumber\\
    &\quad
      +2r_T(q_i+q_j)\sum_{z\ne i,j}
          q_ze^{-\varphi(|f(j)-f(z)|)/T}
      -2(q_i+q_j)c_{j,T},
      \nonumber\\
    b_T(q_{\rm rest})
    &=-D_{ij,T}(q_i+q_j)^2
      +2r_T(q_i+q_j)\sum_{z\ne i,j}q_z
       \left[e^{-\varphi(|f(i)-f(z)|)/T}
             -e^{-\varphi(|f(j)-f(z)|)/T}\right]
      -2(q_i+q_j)(c_{i,T}-c_{j,T}).
    \label{eq:hqcs-curvature-uncentered-coefficients}
\end{align}
Since $q_i+q_j$ is fixed by $q_{\rm rest}$, both coefficients
are independent of $w$. Only the terms with both indices in
$\{i,j\}$ contribute to the quadratic coefficient.
Eq.~\eqref{eq:hqcs-curvature-psd} and $c>0$ now prove
Eq.~\eqref{eq:hqcs-curvature-quadratic-coefficient}.
\end{proof}

The conditional mean now has an explicit positive quadratic
coefficient. To quantify the variance it retains, the next lemma
determines the distribution of the first-step probabilities and the
moments of the relative share of the selected outcome pair.

\begin{lemma}[First-step Haar probabilities and their moments]
\label{lem:hqcs-curvature-dirichlet}
For the fixed pure initial state, the random vector in
Eq.~\eqref{eq:hqcs-curvature-first-prob} has distribution
$q\sim\operatorname{Dir}(d_A,\ldots,d_A)$.
The variable $W=q_i/(q_i+q_j)$ has distribution
$\operatorname{Beta}(d_A,d_A)$ and is independent of $q_i+q_j$
and the remaining probability vector
$q_{\rm rest}$ excluding $q_i, q_j$.
The moments needed below are
\begin{align}
    \operatorname{Var}_{W\sim\operatorname{Beta}(d_A,d_A)}
       [(W-\tfrac12)^2]
    &=\frac{d_A}{4(2d_A+1)^2(2d_A+3)},
    \label{eq:hqcs-curvature-beta-variance}\\
    \mathbb E_{q\sim\operatorname{Dir}(d_A,\ldots,d_A)}
       [(q_i+q_j)^4]
    &=\frac{(2d_A)(2d_A+1)(2d_A+2)(2d_A+3)}
            {d(d+1)(d+2)(d+3)}.
    \label{eq:hqcs-curvature-aggregate-moment}
\end{align}
\end{lemma}

\begin{proof}
The state $U_1\ket{\psi_0}$ can be described by a standard complex Gaussian
vector divided by its Euclidean norm.
The squared moduli of amplitude norms prior normalization independently follow exponential distribution $\operatorname{Exp}(1)$.
Their sums over the $d_A$ coordinates in each measurement block for $B$ are
independent variables $G_z\sim\operatorname{Gamma}(d_A,1)$, with
\begin{equation}
    q_z=\frac{G_z}{\sum_bG_b}.
    \label{eq:hqcs-curvature-gamma}
\end{equation}
Normalization gives the stated Dirichlet distribution.
For the selected pair, the transformation from $(G_i,G_j)$ to their
sum and relative share factors its density into a
$\operatorname{Gamma}(2d_A,1)$ density and a
$\operatorname{Beta}(d_A,d_A)$ density.
Thus $W=G_i/(G_i+G_j)$ is independent of $G_i+G_j$ and every other
$G_z$. Since $q_i+q_j$ and $q_{\rm rest}$ are functions of these
remaining variables, the claimed independence follows.

The Beta integral gives the raw moments
\begin{equation}
    \mathbb E_{W\sim\operatorname{Beta}(d_A,d_A)}[W^k]
    =\prod_{\ell=0}^{k-1}\frac{d_A+\ell}{2d_A+\ell},
    \qquad k=1,2,3,4.
    \label{eq:hqcs-curvature-beta-raw}
\end{equation}
Expansion of the centered powers gives
\begin{align}
    \mathbb E_{W\sim\operatorname{Beta}(d_A,d_A)}
       [(W-\tfrac12)^2]
    &=\frac{d_A+1}{2(2d_A+1)}-\frac12+\frac14
      =\frac1{4(2d_A+1)},
      \nonumber\\
    \mathbb E_{W\sim\operatorname{Beta}(d_A,d_A)}
       [(W-\tfrac12)^4]
    &=\frac{(d_A+2)(d_A+3)}{4(2d_A+1)(2d_A+3)}
      -\frac{d_A+2}{2(2d_A+1)}
      +\frac{3(d_A+1)}{4(2d_A+1)}
      -\frac14+\frac1{16}
      \nonumber\\
    &=\frac3{16(2d_A+1)(2d_A+3)}.
    \label{eq:hqcs-curvature-beta-centered}
\end{align}
Subtracting the square of the second moment yields
\begin{align}
    \operatorname{Var}_{W\sim\operatorname{Beta}(d_A,d_A)}
       [(W-\tfrac12)^2]
    &=\frac3{16(2d_A+1)(2d_A+3)}
      -\frac1{16(2d_A+1)^2}
      \nonumber\\
    &=\frac{d_A}{4(2d_A+1)^2(2d_A+3)}.
\end{align}
For $d_B>2$, Gamma aggregation in
Eq.~\eqref{eq:hqcs-curvature-gamma} gives
$q_i+q_j\sim\operatorname{Beta}(2d_A,d-2d_A)$.
Its fourth moment is
\begin{align}
    \mathbb E_{q\sim\operatorname{Dir}(d_A,\ldots,d_A)}
       [(q_i+q_j)^4]
    &=\prod_{\ell=0}^{3}\frac{2d_A+\ell}{d+\ell},
\end{align}
which is Eq.~\eqref{eq:hqcs-curvature-aggregate-moment}.
\end{proof}

The symmetry of the Beta distribution makes the centered linear and quadratic
components uncorrelated. Using the moments above, the next lemma
turns the positive quadratic coefficient into a variance bound that
holds for every choice of the target coefficients.

\begin{lemma}[Variance retained by the quadratic component]
\label{lem:hqcs-curvature-variance}
For the polynomial in Eq.~\eqref{eq:hqcs-curvature-polynomial}
under $q\sim\operatorname{Dir}(d_A,\ldots,d_A)$,
\begin{equation}
    \operatorname{Var}_{q\sim\operatorname{Dir}(d_A,\ldots,d_A)}
    [\Phi_T(q)]
    \ge
    \frac{d_A^2(d_A+1)}
         {(2d_A+1)d(d+1)(d+2)(d+3)}
    D_{ij,T}^2.
    \label{eq:hqcs-curvature-dirichlet-bound}
\end{equation}
This bound holds for every choice of the linear target coefficients
$c_{z,T}$.
\end{lemma}

\begin{proof}
Conditional on $q_{\rm rest}$, the sum $q_i+q_j$ is fixed.
Lemma~\ref{lem:hqcs-curvature-dirichlet} shows that the remaining
relative share $W$ still has distribution $\operatorname{Beta}(d_A,d_A)$.
Writing $W=(W-\tfrac12)+\tfrac12$ in
Eq.~\eqref{eq:hqcs-curvature-uncentered-polynomial} gives the exact
centered polynomial
\begin{equation}
    \Phi_T(q)=\alpha_T(q_{\rm rest})
       +\beta_T(q_{\rm rest})(W-\tfrac12)
       +D_{ij,T}(q_i+q_j)^2(W-\tfrac12)^2.
    \label{eq:hqcs-curvature-centered-polynomial}
\end{equation}
Here $\alpha_T,\beta_T$ are simply the constant and linear
coefficients of this polynomial conditional on $q_{\rm rest}$.
Explicitly, Eq.~\eqref{eq:hqcs-curvature-uncentered-coefficients} gives
\begin{align}
    \alpha_T(q_{\rm rest})
    &=a_T(q_{\rm rest})+\tfrac12b_T(q_{\rm rest})
      +\tfrac14D_{ij,T}(q_i+q_j)^2,\\
    \beta_T(q_{\rm rest})
    &=b_T(q_{\rm rest})+D_{ij,T}(q_i+q_j)^2.
\end{align}
The Beta distribution is symmetric about $1/2$, so
\begin{align}
    \mathbb E_{W\sim\operatorname{Beta}(d_A,d_A)}
       [W-\tfrac12]&=0,
       \nonumber\\
    \mathbb E_{W\sim\operatorname{Beta}(d_A,d_A)}
       [(W-\tfrac12)^3]&=0,
       \nonumber\\
    \operatorname{Cov}_{W\sim\operatorname{Beta}(d_A,d_A)}
       (W-\tfrac12,(W-\tfrac12)^2)&=0.
    \label{eq:hqcs-curvature-orthogonality}
\end{align}
In particular, this covariance vanishes irrespective of the value of
$\beta_T(q_{\rm rest})$.
Expansion of the conditional variance in
Eq.~\eqref{eq:hqcs-curvature-centered-polynomial} now gives
\begin{align}
    \operatorname{Var}_{W\sim\operatorname{Beta}(d_A,d_A)}
       [\Phi_T(q)\mid q_{\rm rest}]
    &=\beta_T(q_{\rm rest})^2
         \operatorname{Var}_{W\sim\operatorname{Beta}(d_A,d_A)}
         [W-\tfrac12]
       +D_{ij,T}^2(q_i+q_j)^4
         \operatorname{Var}_{W\sim\operatorname{Beta}(d_A,d_A)}
         [(W-\tfrac12)^2]
       \nonumber\\
    &\ge D_{ij,T}^2(q_i+q_j)^4
         \frac{d_A}{4(2d_A+1)^2(2d_A+3)}.
    \label{eq:hqcs-curvature-conditional-bound}
\end{align}
The last line uses Eq.~\eqref{eq:hqcs-curvature-beta-variance}.
The total variance formula, followed by
Eq.~\eqref{eq:hqcs-curvature-aggregate-moment}, gives
\begin{align}
    \operatorname{Var}_{q\sim\operatorname{Dir}(d_A,\ldots,d_A)}
       [\Phi_T(q)]
    &= \mathbb E_{q_{\rm rest}}
       \left[
       \operatorname{Var}_{W\sim\operatorname{Beta}(d_A,d_A)}
       [\Phi_T(q)\mid q_{\rm rest}]
       \right]
       +
       \operatorname{Var}_{q_{\rm rest}}
       \left[
       \mathbb E_{W\sim\operatorname{Beta}(d_A,d_A)}
       [\Phi_T(q)\mid q_{\rm rest}]
       \right]
       \nonumber\\
    &\ge \mathbb E_{q_{\rm rest}}
       \left[
       \operatorname{Var}_{W\sim\operatorname{Beta}(d_A,d_A)}
       [\Phi_T(q)\mid q_{\rm rest}]
       \right]
       \nonumber\\
    &\ge D_{ij,T}^2
      \frac{d_A}{4(2d_A+1)^2(2d_A+3)}
      \mathbb E_{q\sim\operatorname{Dir}(d_A,\ldots,d_A)}
      [(q_i+q_j)^4]
       \nonumber\\
    &=D_{ij,T}^2
      \frac{d_A^2(d_A+1)}
           {(2d_A+1)d(d+1)(d+2)(d+3)}.
\end{align}
The expectation and variance over $q_{\rm rest}$ use its marginal distribution
induced by the Dirichlet vector. 
\end{proof}

The three lemmas now provide the conditional variance bound needed
to complete the proof of Theorem~\ref{thm:hqcs-curvature-lower}.
The total variance formula transfers this bound to the full circuit
ensemble, and the explicit future-kernel estimate gives the uniform
$T^{-2}$ lower bound.

\begin{proof}[Proof of Theorem~\ref{thm:hqcs-curvature-lower}]
Conditioning on the first circuit block in the total variance formula
gives
\begin{align}
    \operatorname{Var}_{\mathbf U}[\mathcal L_{\rm MMD}]
    &=\mathbb E_{U_1}
      \left[
      \operatorname{Var}_{U_2,\ldots,U_T}
      [\mathcal L_{\rm MMD}\mid U_1]
      \right]
      +\operatorname{Var}_{U_1}
      \left[
      \mathbb E_{U_2,\ldots,U_T}
      [\mathcal L_{\rm MMD}\mid U_1]
      \right]
      \nonumber\\
    &\ge\operatorname{Var}_{U_1}[\Phi_T(U_1)]
      =\operatorname{Var}_{q\sim\operatorname{Dir}(d_A,\ldots,d_A)}
       [\Phi_T(q)].
    \label{eq:hqcs-curvature-first-conditioning}
\end{align}
Lemma~\ref{lem:hqcs-curvature-variance} and
Eq.~\eqref{eq:hqcs-curvature-quadratic-coefficient} consequently give
\begin{align}
    \operatorname{Var}_{\mathbf U}[\mathcal L_{\rm MMD}]
    &\ge
    \frac{d_A^2(d_A+1)}
         {(2d_A+1)d(d+1)(d+2)(d+3)}D_{ij,T}^2
      \nonumber\\
    &\ge
    \frac{4d_A^2(d_A+1)}
         {(2d_A+1)d(d+1)(d+2)(d+3)}
      r_T^2(1-e^{-c/T})^2,
\end{align}
which proves Eq.~\eqref{eq:hqcs-curvature-finite-lower}.
For $x\ge0$, integration gives
$1-e^{-x}=\int_0^xe^{-u}\,du\ge xe^{-x}$.
Together with Eq.~\eqref{eq:hqcs-curvature-future-bounds}, this yields
\begin{align}
    r_T^2(1-e^{-c/T})^2
    &\ge e^{-2\varphi_{\max}(T-1)/T}
          \frac{c^2}{T^2}e^{-2c/T}
      \nonumber\\
    &=\frac{c^2}{T^2}
       e^{-2[\varphi_{\max}(T-1)+c]/T}
      \nonumber\\
    &\ge\frac{c^2e^{-2\varphi_{\max}}}{T^2}.
    \label{eq:hqcs-curvature-explicit-factor}
\end{align}
The last inequality follows from $c\le\varphi_{\max}$.
Substitution proves Eq.~\eqref{eq:hqcs-curvature-polynomial-lower},
including $T=1$.
\end{proof}

Notice that to obtain this lower bound, we require $U_1$ to satisfy 4-design, while $U_2, \dots, U_T$ only need to satisfy 2-design.

\section{Variance of MMD of HQCS for increment space}
\label{app:hqcs-increment}

This section studies the normalized kernel on the complete increment
vector of an HQCS trajectory. The circuit and measurement model are
those of Definition~\ref{def:hqcs-organized-transfer}.
The fixed initial density matrix $\rho_0$ is arbitrary and independent
of the circuit; it need not be pure or maximally mixed.
For $T\ge2$, the increments are defined by
$\Delta x_t=x_t-x_{t-1}$ for $t=2,\ldots,T$.
The local discrepancy $\varphi:[0,2]\to[0,\infty)$ is measurable, with
$\varphi(0)=0$ and
$\varphi_{\max}:=\sup_{0\le r\le2}\varphi(r)<\infty$.
The domain extends that of Definition~\ref{def:hqcs-exponential-kernel},
since the difference between two increments can have magnitude $2$.
The normalized exponential kernel is thus
\begin{equation}
    K_T^\Delta(x,\tilde x)
    :=\exp\!\left[-\frac1{T-1}\sum_{t=2}^T
                     \varphi(|\Delta x_t-\Delta\tilde x_t|)\right],
    \label{eq:hqcs-inc-kernel}
\end{equation}
which is assumed positive semidefinite for each $T$.
It acts on the entire increment vector, rather than on a single
increment marginal.
The MMD loss and average local costs are
\begin{align}
    \mathcal L_{\rm MMD}^\Delta(\mathbf U)
    &=A_T^\Delta+B_T^\Delta(\mathbf U)-2C_T^\Delta(\mathbf U),
    \nonumber\\
    A_T^\Delta
    &:=\mathbb E_{X,X'\sim P_T}
        [K_T^\Delta(X,X')], \quad B_T^\Delta(\mathbf U):=\mathbb E_{Y,Y'\sim Q_{\mathbf U}}[e^{-S_{QQ}}], \quad
    C_T^\Delta(\mathbf U):=
    \mathbb E_{X\sim P_T}\mathbb E_{Y\sim Q_{\mathbf U}}[e^{-S_{XQ}}],
    \label{eq:hqcs-inc-loss}
\end{align}
where the accumulated kernels are
\be
    S_{QQ}:=\frac1{T-1}\sum_{t=2}^T
        \varphi(|\Delta Y_t-\Delta Y'_t|),
    \quad
    S_{xQ}:=\frac1{T-1}\sum_{t=2}^T
        \varphi(|\Delta x_t-\Delta Y_t|).
    \label{eq:accumulate_kernel}
\ee
Rather than the setup of pure initial state, we first consider a reference initial state to be maximally mixed. The maximally mixed reference input gives the local joint probability
and increment value
\begin{equation}
    p_U(a,b):=\frac1d\operatorname{Tr}(\Pi_bU\Pi_aU^\dagger),
    \qquad \delta_{ab}:=f(b)-f(a).
    \label{eq:hqcs-inc-reference}
\end{equation}
All outcome indices below range over $0,\ldots,d_B-1$.
The local reference costs and their Haar means are
\begin{align}
    h_\xi(U)&:=\sum_{a,b}p_U(a,b)\varphi(|\xi-\delta_{ab}|),
         \qquad \xi\in[-1,1],
         \nonumber\\
    h_{QQ}(U)&:=\sum_{a,b,c,e}p_U(a,b)p_U(c,e)
                  \varphi(|\delta_{ab}-\delta_{ce}|),
         \nonumber\\
    \bar h_\xi&:=\mathbb E_{U\sim{\rm Haar}}[h_\xi(U)]
            =\frac1{d_B^2}\sum_{a,b}\varphi(|\xi-\delta_{ab}|),
    \qquad
    \overline{h}_{QQ}:=\mathbb E_{U\sim{\rm Haar}}[h_{QQ}(U)].
    \label{eq:hqcs-inc-local-costs}
\end{align}
The centering quantities used in the final expansion are
\begin{align}
    m_T(x)&:=\frac1{T-1}\sum_{t=2}^T\bar h_{\Delta x_t},
    \qquad
    \ell_T:=A_T^\Delta+e^{-\overline{h}_{QQ}}
              -2\mathbb E_{X\sim P_T}[e^{-m_T(X)}],
    \nonumber\\
    \mathcal{I}_{t,T}(U)
    &:=-e^{-\overline{h}_{QQ}}(h_{QQ}(U)-\overline{h}_{QQ})
       +2\mathbb E_{X\sim P_T}
        [e^{-m_T(X)}(h_{\Delta X_t}(U)-\bar h_{\Delta X_t})],
    \nonumber\\
    \Gamma_T&:=\frac1{T-1}\sum_{t=2}^T
         \operatorname{Var}_{U\sim{\rm Haar}}[\mathcal{I}_{t,T}(U)].
    \label{eq:hqcs-inc-leading}
\end{align}
The target expectation in $\mathcal{I}_{t,T}$ retains the full trajectory
weight $e^{-m_T(X)}$.

The following constants recur in the moment and remainder bounds:
\begin{align}
    \kappa_1&=\frac{d_A^2d_B-1}{d^2-1}\in(0,1),
    \qquad
    C_{\rm init}:=
       \frac{\sqrt d\,\|\rho_0-I/d\|_2}{1-\sqrt{\kappa_1}},
       \nonumber\\
    C_G&:=\varphi_{\max}\left(1+\frac1{1-\kappa_1}\right),
    \qquad K_D:=\varphi_{\max}+2C_G,
       \nonumber\\
    C_{\rm rem}&:=
       4\varphi_{\max}C_{\rm init}
       +\frac12\sqrt{56K_D^4+128C_G^4}
       +\sqrt{288}\,\varphi_{\max}^2.
    \label{eq:hqcs-inc-constants}
\end{align}
$\kappa_1$ is also used in
Eq.~\eqref{eq:hqcs-organized-mix-one}.
The notation $K_D$ denotes a bound on a martingale difference, not a
kernel.
The estimate
$C_{\rm init}\le\sqrt{d-1}/(1-\sqrt{\kappa_1})$
follows from
$\|\rho_0-I/d\|_2^2=\operatorname{Tr}\rho_0^2-1/d\le1-1/d$.
Pure initial states, including the all-zero state, attain this
purity bound; the maximally mixed state has $C_{\rm init}=0$.

With these definitions, we have the main result as follows.
\begin{theorem}[Formal variance of MMD of HQCS in increment space]
\label{thm:hqcs-inc-variance}
For fixed-register HQCS model with every unitary independently sampled from Haar-random unitary and normalized exponential kernel in Eq.~\eqref{eq:hqcs-inc-kernel}, for every
$T\ge2$ the increment-space MMD loss has the exact decomposition
\begin{equation}
    \mathcal L_{\rm MMD}^\Delta(\mathbf U)
       =\ell_T+\frac1{T-1}\sum_{t=2}^T\mathcal{I}_{t,T}(U_t)+R_T(\mathbf U),
    \label{eq:hqcs-inc-loss-expansion}
\end{equation}
with $\sqrt{\mathbb E_{\mathbf U}[|R_T|^2]}
\le C_{\rm rem}/(T-1)$. 
The quantities $\ell_T,\mathcal{I}_{t,T},\Gamma_T$ are defined in
Eq.~\eqref{eq:hqcs-inc-leading}, and $C_{\rm rem}$ in
Eq.~\eqref{eq:hqcs-inc-constants}. The variance of MMD can thus be bounded from two sides as
\begin{align}
    \left|\operatorname{Var}_{\mathbf U}
       [\mathcal L_{\rm MMD}^\Delta]-\frac{\Gamma_T}{T-1}\right|
    &\le\frac{2C_{\rm rem}\sqrt{\Gamma_T}}{(T-1)^{3/2}}
                   +\frac{C_{\rm rem}^2}{(T-1)^2},
       \label{eq:hqcs-inc-variance-error}
\end{align}
which yields an upper bound of
\be
	\operatorname{Var}_{\mathbf U}[\mathcal L_{\rm MMD}^\Delta]
	\le\frac{9\varphi_{\max}^2}{T-1}
	+\frac{6\varphi_{\max}C_{\rm rem}}{(T-1)^{3/2}}
	+\frac{C_{\rm rem}^2}{(T-1)^2}.
	\label{eq:hqcs-inc-variance-upper-bound}
\ee
At $\Gamma_T\ge\gamma>0$ is lower bounded by a positive constant $\gamma$ and 
$T-1\ge25C_{\rm rem}^2/\gamma$, the explicit two-sided bound is
\begin{equation}
    \frac{\Gamma_T}{2(T-1)}
    \le\operatorname{Var}_{\mathbf U}[\mathcal L_{\rm MMD}^\Delta]
    \le\frac{3\Gamma_T}{2(T-1)}.
    \label{eq:hqcs-inc-nondegenerate}
\end{equation}
\end{theorem}

The proof isolates independent local circuit contributions to the loss.
Lemma~\ref{lem:hqcs-inc-dynamics} controls the initial transient,
and Lemma~\ref{lem:hqcs-inc-means} expresses the conditional mean
costs as sums of local reference costs with small errors.
Lemma~\ref{lem:hqcs-inc-fourth} bounds their fourth moments, allowing
Lemma~\ref{lem:hqcs-inc-linearization} to expand the exponential
kernel averages with a remainder of order $1/(T-1)$ in mean square
norm. Combining the model and target terms gives the stated loss
decomposition. Independence and centering of the local contributions
yield the leading variance $\Gamma_T/(T-1)$, while Cauchy--Schwarz
controls the remainder and covariance errors.

We first control the departure from the maximally mixed reference
state. The next lemma bounds the initial transient and describes the
shared-unitary two-copy average dynamics used in the moment estimates.

\begin{lemma}[Average-channel dynamics and the initial transient]
\label{lem:hqcs-inc-dynamics}
The nonselective fixed-circuit states
$\rho_r=\mathcal T_{U_r}\cdots\mathcal T_{U_1}(\rho_0)$ satisfy
\begin{align}
    \mathbb E_{U_1,\ldots,U_r}[\|\rho_r-I/d\|_2^2]
       &=\kappa_1^r\|\rho_0-I/d\|_2^2,
       \label{eq:hqcs-inc-state-contraction}\\
    \sqrt{\mathbb E_{\mathbf U}
       \left[\left(\sum_{r=1}^{T-1}\|\rho_r-I/d\|_1\right)^2\right]}
       &\le C_{\rm init}\sqrt{\kappa_1}(1-\kappa_1^{(T-1)/2})
        \le C_{\rm init}.
       \label{eq:hqcs-inc-transient}
\end{align}
The shared-unitary two-copy average channel
$\overline{\mathcal T}^{(2)}
 :=\mathbb E_{U\sim{\rm Haar}}[\mathcal T_U\otimes\mathcal T_U]$
satisfies, for any operator $R$ and every integer $m\ge1$,
\begin{equation}
    \left(\overline{\mathcal T}^{(2)}\right)^m(R)
    =\frac{\operatorname{Tr}R}{d^2}I_2
       +\kappa_1^{m-1}
        \left[\overline{\mathcal T}^{(2)}(R)
                    -\frac{\operatorname{Tr}R}{d^2}I_2\right],
    \label{eq:hqcs-inc-average-iteration}
\end{equation}
where $I_2$ is the identity on the two complete $AB$ replicas.
\end{lemma}

\begin{proof}
Each $\mathcal T_U$ is unital and trace preserving.
Conditional on $U_1,\ldots,U_{r-1}$, the matrix
$\rho_{r-1}-I/d$ is deterministic, Hermitian, and traceless.
The single-copy identity in Lemma~\ref{lem:hqcs-organized-mixing} gives one-step result as
\begin{align}
    \mathbb E_{U_r\sim{\rm Haar}}
    [\|\rho_r-I/d\|_2^2\mid U_1,\ldots,U_{r-1}]
    &=\kappa_1\|\rho_{r-1}-I/d\|_2^2.
\end{align}
Iteration proves Eq.~\eqref{eq:hqcs-inc-state-contraction}.
The matrix inequality $\|R\|_1\le\sqrt d\,\|R\|_2$ and the scalar
mean-square triangle inequality give
\begin{align}
    \sqrt{\mathbb E_{\mathbf U}
       \left[\left(\sum_{r=1}^{T-1}\|\rho_r-I/d\|_1\right)^2\right]}
    &\le\sum_{r=1}^{T-1}
          \sqrt{\mathbb E_{U_1,\ldots,U_r}
                       [\|\rho_r-I/d\|_1^2]}
       \nonumber\\
    &\le\sqrt d\,\|\rho_0-I/d\|_2
                  \sum_{r=1}^{T-1}\kappa_1^{r/2}
       =C_{\rm init}\sqrt{\kappa_1}(1-\kappa_1^{(T-1)/2}).
\end{align}
The first inequality follows by expanding the square and bounding
each cross expectation by Cauchy--Schwarz inequality.

For the second assertion, $F$ denotes the swap between the complete
replicas and $\mathcal D(R)=\sum_z\Pi_zR\Pi_z$ denotes measurement
dephasing. The two-copy Haar identity in
Eq.~\eqref{eq:hqcs-organized-haar-two} equivalently gives
\begin{align}
    \mathbb E_{U\sim{\rm Haar}}
      [U^{\otimes2}R(U^\dagger)^{\otimes2}]
    &=\frac{\operatorname{Tr}R-\operatorname{Tr}(FR)/d}{d^2-1}I_2
      +\frac{\operatorname{Tr}(FR)-\operatorname{Tr}R/d}{d^2-1}F.
    \label{eq:hqcs-inc-twirl}
\end{align}
In register order $AA'BB'$, measurement dephasing replaces $F$ by
$S:=(\mathcal D\otimes\mathcal D)(F)$ gives
\begin{equation}
	S=F_A\otimes\sum_z|zz\rangle\langle zz|_{BB'},
\end{equation}
with $\Tr S = d$ and $\Tr(FS) = d_A^2 d_B$. Then we have
\begin{align}
    \overline{\mathcal T}^{(2)}(R)
    &=\frac{\operatorname{Tr}R}{d^2}I_2
       +\frac{\operatorname{Tr}(FR)-\operatorname{Tr}R/d}{d^2-1}
           \left(S-\frac{I_2}{d}\right).
    \label{eq:hqcs-inc-average-channel}
\end{align}
The direction $S-I_2/d$ has trace $d-d^2/d=0$, while
$I_2/d^2$ is the stationary density matrix.
Substitution of the traceless direction into
Eq.~\eqref{eq:hqcs-inc-average-channel} gives
\begin{align}
    \overline{\mathcal T}^{(2)}(S-I_2/d)
    &=\frac{\operatorname{Tr}[F(S-I_2/d)]}{d^2-1}(S-I_2/d)
       \nonumber\\
    &=\frac{d_A^2d_B-1}{d^2-1}(S-I_2/d)
      =\kappa_1(S-I_2/d).
\end{align}
After one average step, only the stationary direction and this
traceless direction remain. Unitality and iteration give
Eq.~\eqref{eq:hqcs-inc-average-iteration}.
\end{proof}

The transient bound makes the cumulative effect of the initial state
uniformly controlled. The next lemma uses this estimate to replace
the conditional mean costs by sums of local reference costs, with
explicit errors of order $1/(T-1)$ in mean square norm.

\begin{lemma}[Localization of the conditional mean costs]
\label{lem:hqcs-inc-means}
For every fixed initial density matrix, the accumulated kernels in Eq.~\eqref{eq:accumulate_kernel} have the representation of
\begin{align}
    \mathbb E_{Y,Y'\sim Q_{\mathbf U}}[S_{QQ}]
       &=\frac1{T-1}\sum_{t=2}^T h_{QQ}(U_t)+\varepsilon_B(\mathbf U),
       \nonumber\\
    \mathbb E_{Y\sim Q_{\mathbf U}}[S_{xQ}]
       &=\frac1{T-1}\sum_{t=2}^T h_{\Delta x_t}(U_t)
                       +\varepsilon_x(\mathbf U),
       \label{eq:hqcs-inc-mean-representations}
\end{align}
where $h_QQ$ and $h_{\Delta x_t}$ are defined in Eq.~\eqref{eq:hqcs-inc-leading} and
the residual errors are bounded by
\begin{align}
	\sqrt{\mathbb E_{\mathbf U}[|\varepsilon_B|^2]}
	&\le\frac{2\varphi_{\max}C_{\rm init}}{T-1}
	\sqrt{\kappa_1}(1-\kappa_1^{(T-1)/2}),
	\nonumber\\
	\sup_{x\in[0,1]^T}
	\sqrt{\mathbb E_{\mathbf U}[|\varepsilon_x|^2]}
	&\le\frac{\varphi_{\max}C_{\rm init}}{T-1}
	\sqrt{\kappa_1}(1-\kappa_1^{(T-1)/2}).
	\label{eq:hqcs-inc-mean-errors}
\end{align}
Both errors vanish identically when $\rho_0=I/d$.
\end{lemma}

\begin{proof}
For the reference input $I/d$, the first outcome in an adjacent pair
has probability $\operatorname{Tr}(\Pi_aI/d)=1/d_B$ and conditional
state $\Pi_a/d_A$. The next measurement under a fixed $U$ gives
\begin{align}
    \Pr(a,b\mid U,I/d)
    &=\frac1{d_B}\operatorname{Tr}(\Pi_bU(\Pi_a/d_A)U^\dagger)
      =p_U(a,b),
    \label{eq:hqcs-inc-uniform-marginals}
\end{align}
which satisfies $\sum_b p_U(a,b)=\frac{\operatorname{Tr}\Pi_a}{d}
=\frac1{d_B}$ and $
\sum_a p_U(a,b)=\frac{\operatorname{Tr}\Pi_b}{d}
=\frac1{d_B}$.
The first Haar moment used in Eq.~\eqref{eq:hqcs-curvature-weighted-haar}
also gives
\begin{align}
    \mathbb E_{U\sim{\rm Haar}}[p_U(a,b)]
    &=\frac1d\operatorname{Tr}
       \left(\Pi_b\,\frac{\operatorname{Tr}\Pi_a}{d}I\right)
      =\frac{d_A^2}{d^2}=\frac1{d_B^2}.
    \label{eq:hqcs-inc-average-probability}
\end{align}
This proves the third line for $\bar h_\xi$ in
Eq.~\eqref{eq:hqcs-inc-local-costs}.
All the local costs and their means lie in $[0,\varphi_{\max}]$.
When the initial state is $I/d$, unitality makes every nonselective
state equal to $I/d$, for every fixed circuit.
The two trajectories are independently conditional on
the circuit and gives $h_{QQ}(U_t)$. Linearity of the expectation of the sum proves
Eq.~\eqref{eq:hqcs-inc-mean-representations} with zero errors.

For an arbitrary initial state, $\rho_{t-1}$ is block diagonal after
the previous measurement on $B$. The actual adjacent probability is
\begin{align}
    \Pr(z_{t-1}=a,z_t=b\mid\mathbf U)
    =\operatorname{Tr}
      [\Pi_bU_t\Pi_a\rho_{t-1}\Pi_aU_t^\dagger]
      =\operatorname{Tr}[E_{ab}(U_t)\rho_{t-1}],
    \label{eq:hqcs-inc-actual-pair}
\end{align}
where $E_{ab}(U):=\Pi_aU^\dagger\Pi_bU\Pi_a\ge0$ also define POVM as $\sum_{a,b}E_{ab}(U)=I$.
The cost effect
$O_{\xi,U}:=\sum_{a,b}\varphi(|\xi-\delta_{ab}|)E_{ab}(U)$
satisfies $0\le O_{\xi,U}\le\varphi_{\max}I$.
Consequently,
\begin{align}
    \left|\mathbb E_{Y\sim Q_{\mathbf U}}
       [\varphi(|\Delta x_t-\Delta Y_t|)]
                 -h_{\Delta x_t}(U_t)\right|
    &=\left|\operatorname{Tr}
          [O_{\Delta x_t,U_t}(\rho_{t-1}-I/d)]\right|
       \nonumber\\
    &\le\varphi_{\max}\|\rho_{t-1}-I/d\|_1.
    \label{eq:hqcs-inc-single-input-error}
\end{align}
Similarly, the two-copy effect uses
$E_{ab}(U)\otimes E_{ce}(U)$ and weights
$\varphi(|\delta_{ab}-\delta_{ce}|)$, and the corresponding cost effect lies between
$0$ and $\varphi_{\max}I_2$. We therefore have
\begin{align}
	\left|\mathbb E_{Y,Y'\sim Q_{\mathbf U}}
	[\varphi(|\Delta Y_t-\Delta Y'_t|)]-h_{QQ}(U_t)\right|
	&= \left|\Tr(O_{\Delta Y_t, \Delta Y'_t}(\rho_{t-1}^{\otimes 2} - I_2/d^2))\right| \nonumber\\
	&\le \varphi_{\max}\|\rho_{t-1}^{\otimes 2}-I_2/d^2\|_1 \le 2\varphi_{\max}\|\rho_{t-1}-I/d\|_1,
	\label{eq:hqcs-inc-double-input-error}
\end{align}
where the last inequality follows from $\rho^{\otimes2}-I_2/d^2
=(\rho-I/d)\otimes\rho+(I/d)\otimes(\rho-I/d)$.
The errors in Eq.~\eqref{eq:hqcs-inc-mean-representations} are defined as
\begin{align}
	\varepsilon_x(\mathbf U) &= \frac{1}{T-1}\sum_{t=2}^T \left[\E_{Y}\varphi(|\Delta x_t - \Delta Y_t|) - h_{\Delta x_t}(U_t)\right], \nonumber\\
	\varepsilon_B(\mathbf U) &= \frac{1}{T-1}\sum_{t=2}^T \left[\E_{Y, Y'}\varphi(|\Delta Y_t - \Delta Y_t'|) - h_{QQ}(U_t)\right].
\end{align}
Using Eqs.~\eqref{eq:hqcs-inc-single-input-error}
and~\eqref{eq:hqcs-inc-double-input-error} yields
\begin{align}
    |\varepsilon_x(\mathbf U)|
      &\le\frac{\varphi_{\max}}{T-1}
                \sum_{r=1}^{T-1}\|\rho_r-I/d\|_1,
    \qquad
    |\varepsilon_B(\mathbf U)|
      \le\frac{2\varphi_{\max}}{T-1}
                \sum_{r=1}^{T-1}\|\rho_r-I/d\|_1.
\end{align}
Lemma~\ref{lem:hqcs-inc-dynamics} proves
Eq.~\eqref{eq:hqcs-inc-mean-errors}, uniformly in $x$.
\end{proof}

The conditional mean representations identify the local circuit
contributions. To transfer them to the exponential kernel averages,
we also need control of the fluctuations of the sampled costs, and the
next lemma supplies the required fourth-moment bounds.

\begin{lemma}[Explicit fourth moments of the average costs]
\label{lem:hqcs-inc-fourth}
For any fixed initial density matrix and every $T\ge2$,
\begin{align}
    \sup_{x\in[0,1]^T}\mathbb E_{\mathbf U}
       \mathbb E_{Y\sim Q_{\mathbf U}}\!\left[\left(S_{xQ}-\frac1{T-1}\sum_{t=2}^T\bar h_{\Delta x_t}\right)^4\right]
    &\le\varphi_{\max}^4
                \left(\frac{96}{(T-1)^2}+\frac{192}{(T-1)^3}\right)
       \le\frac{288\varphi_{\max}^4}{(T-1)^2},
       \label{eq:hqcs-inc-fourth-single}\\
    \mathbb E_{\mathbf U}
       \mathbb E_{Y,Y'\sim Q_{\mathbf U}}[(S_{QQ}-\overline{h}_{QQ})^4]
    &\le\frac{8K_D^4}{(T-1)^4}
       \left[3(T-1)(T-2)+4\sum_{j=0}^{T-2}\sqrt j+(T-1)\right]
              +\frac{128C_G^4}{(T-1)^4}
       \nonumber\\
    &\le\frac{56K_D^4}{(T-1)^2}+\frac{128C_G^4}{(T-1)^4}.
       \label{eq:hqcs-inc-fourth-double}
\end{align}
The constants are defined in Eq.~\eqref{eq:hqcs-inc-constants}.
\end{lemma}

\begin{proof}
For the single trajectory, the averaged trajectory probability is
$\bar p_T(\mathbf z):=\mathbb E_{\mathbf U}[\Pr(\mathbf z\mid\mathbf U)]$.
For a fixed outcome prefix, the unnormalized branch states from
Eq.~\eqref{eq:hqcs-organized-branch} satisfy
$\tau_0=\rho_0$ and
$\tau_j=\Pi_{z_j}U_j\tau_{j-1}U_j^\dagger\Pi_{z_j}$.
The first-order Haar moment gives, conditional on the preceding blocks,
\begin{align}
    \mathbb E_{U_j\sim{\rm Haar}}
      [\operatorname{Tr}\tau_j\mid U_1,\ldots,U_{j-1}]
    &=\operatorname{Tr}
       \left[\Pi_{z_j}\frac{\operatorname{Tr}\tau_{j-1}}d
                         I\Pi_{z_j}\right]
      =\frac1{d_B}\operatorname{Tr}\tau_{j-1},
       \nonumber\\
    \bar p_T(z_1,\ldots,z_T)
    &=d_B^{-T}\operatorname{Tr}\rho_0=d_B^{-T}.
    \label{eq:hqcs-inc-averaged-output-law}
\end{align}
The second line follows by successive averaging from the last block
back to the first. 

For fixed $x$, the centered pair costs
$V_t:=\varphi(|\Delta x_t-\Delta Y_t|)-\bar h_{\Delta x_t}$
have mean zero under $\bar p_T$, with $|V_t|\le\varphi_{\max}$.
Each depends only on $(z_{t-1},z_t)$.
For $|t-s|>1$ and measurable sets $A,B$, the joint distribution factors as
\begin{align}
    \Pr_{\mathbf z\sim\bar p_T}(V_t\in A,V_s\in B)
    &=\frac1{d_B^4}\sum_{a,b,c,e}
       \mathbf1_{\{\varphi(|\Delta x_t-\delta_{ab}|)
                                      -\bar h_{\Delta x_t}\in A\}}
       \mathbf1_{\{\varphi(|\Delta x_s-\delta_{ce}|)
                                      -\bar h_{\Delta x_s}\in B\}}
       \nonumber\\
    &=\Pr_{\mathbf z\sim\bar p_T}(V_t\in A)
             \Pr_{\mathbf z\sim\bar p_T}(V_s\in B).
    \label{eq:hqcs-inc-disjoint-edge-law}
\end{align}
The same factorization holds between any groups of pair costs whose
underlying output coordinate sets are disjoint.
In the fourth-power expansion
\begin{equation}
    \mathbb E_{\mathbf z\sim\bar p_T}
       \left[\left(\sum_{t=2}^TV_t\right)^4\right]
    =\sum_{i_1,i_2,i_3,i_4=2}^T
      \mathbb E_{\mathbf z\sim\bar p_T}
       [V_{i_1}V_{i_2}V_{i_3}V_{i_4}],
    \label{eq:hqcs-inc-fourth-count-expansion}
\end{equation}
the four indices can first be sorted.
Adjacent sorted indices with gap $0$ or $1$ belong to the same
cluster, and a gap at least $2$ separates clusters.
Cluster sizes count factors, including repeated indices.
A single-factor cluster is independent of all other factors and has
mean zero, so its term vanishes.
The only remaining cluster-size patterns are $4$ and $2+2$.
A four-factor cluster has the index form as
$(a,a+\epsilon_1,a+\epsilon_1+\epsilon_2,
a+\epsilon_1+\epsilon_2+\epsilon_3)$, with
$\epsilon_1,\epsilon_2,\epsilon_3\in\{0,1\}$.
There are at most $(T-1)\cdot2^3=8(T-1)$ sorted choices.
Meanwhile, two two-factor clusters have the index form
$(a,a+\epsilon,b,b+\eta)$, with
$\epsilon,\eta\in\{0,1\}$ and $b\ge a+\epsilon+2$.
Ignoring boundary and separation restrictions gives at most
$(T-1)^2\cdot2^2=4(T-1)^2$ sorted choices.
Each sorted tuple corresponds to at most $4!=24$ ordered tuples;
repetitions only reduce that number. Each summand has absolute value
at most $\varphi_{\max}^4$, hence
\begin{align}
    \mathbb E_{\mathbf z\sim\bar p_T}
       \left[\left(\sum_{t=2}^TV_t\right)^4\right]
    &\le24(8(T-1)+4(T-1)^2)\varphi_{\max}^4,
       \nonumber\\
    \mathbb E_{\mathbf U}\mathbb E_{Y\sim Q_{\mathbf U}}
       \!\left[\left(S_{xQ}-\frac1{T-1}\sum_{t=2}^T\bar h_{\Delta x_t}\right)^4\right]
    &\le\varphi_{\max}^4
       \left(\frac{96}{(T-1)^2}+\frac{192}{(T-1)^3}\right),
\end{align}
which proves
Eq.~\eqref{eq:hqcs-inc-fourth-single}.

We now prove Eq.~\eqref{eq:hqcs-inc-fourth-double}.
The two measurement trajectories are independent conditional on the
same circuit. We retain the separate circuit expectation
$\mathbb E_{\mathbf U}$ and trajectory expectation
$\mathbb E_{\mathbf z,\mathbf z'\sim Q^{\rm out}_{\mathbf U}}$,
where $Q^{\rm out}_{\mathbf U}$ is the outcome distribution of a fixed
circuit. Define
\begin{equation}
    H_t:=\varphi(|\Delta Y_t-\Delta Y'_t|)-\overline{h}_{QQ},
    \qquad t=2,\ldots,T.
    \label{eq:hqcs-inc-centered-pair-cost}
\end{equation}
The accumulated-cost definition in Eq.~\eqref{eq:accumulate_kernel}
gives
\begin{equation}
    S_{QQ}-\overline{h}_{QQ}=\frac1{T-1}\sum_{t=2}^T H_t.
    \label{eq:hqcs-inc-fourth-double-reduction}
\end{equation}
Thus it suffices to bound the fourth moment of $\sum_t H_t$ by a
quantity of order $(T-1)^2$.
The costs in Eq.~\eqref{eq:hqcs-inc-local-costs} and their means lie
in $[0,\varphi_{\max}]$, so $|H_t|\le\varphi_{\max}$.
Unlike the single-trajectory argument based on
Eq.~\eqref{eq:hqcs-inc-averaged-output-law}, the joint process of
two trajectories sharing the same blocks need not have independent
outputs. We will instead remove the conditional drift of $H_t$ and
estimate the resulting martingale sum.

Let $\mathcal F_t$ denote the information generated by the recorded
history $(U_1,\ldots,U_t;z_1,\ldots,z_t;z'_1,\ldots,z'_t)$.
Let $\mathsf R_t$ be the tensor product of the two normalized
conditional states. Given $\mathcal F_{t-1}$, the previous outcomes
and $\mathsf R_{t-1}$ are fixed and the new variables are $U_t,z_t,z'_t$.
Write $\Pi_{ab}:=\Pi_a\otimes\Pi_b$.
For a fixed new $U$ and next outcome pair $(c,e)$, the instrument in
Definition~\ref{def:hqcs-organized-transfer} gives the unnormalized
branch state $\Pi_{ce}U^{\otimes2}\mathsf R_{t-1}
(U^\dagger)^{\otimes2}\Pi_{ce}$, and
its probability is given by the trace.
The probability cancels this normalization when the outcomes are
averaged, giving
\begin{equation}
    \mathbb E_{z_t,z'_t\mid\mathcal F_{t-1},\,U_t=U}[\mathsf R_t]
    =\sum_{c,e}\Pi_{ce}U^{\otimes2}\mathsf R_{t-1}
    (U^\dagger)^{\otimes2}\Pi_{ce}
    =(\mathcal T_U\otimes\mathcal T_U)(\mathsf R_{t-1}).
    \label{eq:hqcs-inc-conditional-state-fixed}
\end{equation}
Averaging the fresh independent Haar block in
Eq.~\eqref{eq:hqcs-inc-conditional-state-fixed}, using the average
channel defined in Lemma~\ref{lem:hqcs-inc-dynamics}, yields
\begin{equation}
    \mathbb E_{U_t,z_t,z'_t\mid\mathcal F_{t-1}}[\mathsf R_t]
       =\overline{\mathcal T}^{(2)}(\mathsf R_{t-1}).
    \label{eq:hqcs-inc-conditional-state}
\end{equation}

To express the conditional mean of $H_t$, define the branch
probability for previous outcomes $(a,b)$ and next outcomes $(c,e)$:
\begin{equation}
    q_R(a,b;c,e\mid U)
    :=\operatorname{Tr}
       [\Pi_{ce}U^{\otimes2}\Pi_{ab}R\Pi_{ab}
                         (U^\dagger)^{\otimes2}\Pi_{ce}].
    \label{eq:hqcs-inc-drift-branch}
\end{equation}
The increments on this branch are $\delta_{ac}$ and $\delta_{be}$,
as defined in Eq.~\eqref{eq:hqcs-inc-reference}. The corresponding
Haar-averaged centered cost is the linear functional
\begin{equation}
    g(R):=\mathbb E_{U\sim{\rm Haar}}
       \sum_{a,b,c,e}
       \left(\varphi(|\delta_{ac}-\delta_{be}|)-\overline{h}_{QQ}\right)
           q_R(a,b;c,e\mid U).
    \label{eq:hqcs-inc-drift}
\end{equation}
For density matrices $R$, the probabilities in
Eq.~\eqref{eq:hqcs-inc-drift-branch} are nonnegative and sum to one.
The centered cost is bounded, and hence $|g(R)| \le \varphi_{\max}$.
The actual $\mathsf R_{t-1}$ is supported in its known previous
outcome block. In Eq.~\eqref{eq:hqcs-inc-drift-branch}, only that
block contributes. Averaging the costs in
Eq.~\eqref{eq:hqcs-inc-centered-pair-cost} therefore gives
\begin{equation}
    \mathbb E_{U_t,z_t,z'_t\mid\mathcal F_{t-1}}[H_t]=g(\mathsf R_{t-1}).
    \label{eq:hqcs-inc-centered-drift}
\end{equation}
For the maximally mixed reference input, however,
Eq.~\eqref{eq:hqcs-inc-reference} and tensor-product factorization
give, at the same fixed $U$,
\begin{equation}
    q_{I_2/d^2}(a,b;c,e\mid U)=p_U(a,c)p_U(b,e).
    \label{eq:hqcs-inc-reference-drift-branch}
\end{equation}
Substituting Eq.~\eqref{eq:hqcs-inc-reference-drift-branch} into
Eq.~\eqref{eq:hqcs-inc-drift}, and using the definition of
$\overline{h}_{QQ}$ in Eq.~\eqref{eq:hqcs-inc-local-costs}, shows that
\begin{equation}
    g(I_2/d^2)
    =\mathbb E_{U\sim{\rm Haar}}[h_{QQ}(U)]-\overline{h}_{QQ}=0.
    \label{eq:hqcs-inc-reference-drift-zero}
\end{equation}
The product in Eq.~\eqref{eq:hqcs-inc-reference-drift-branch} remains
inside the Haar average as the two trajectories share $U$.

We use mixing to compensate for the nonzero conditional mean in
Eq.~\eqref{eq:hqcs-inc-centered-drift}.
Applying the linear functional $g$ to the iteration formula
Eq.~\eqref{eq:hqcs-inc-average-iteration}, and using
Eq.~\eqref{eq:hqcs-inc-reference-drift-zero}, gives, for density
matrices $R$ and $m\ge1$,
\begin{equation}
    g((\overline{\mathcal T}^{(2)})^mR)
       =\kappa_1^{m-1}g(\overline{\mathcal T}^{(2)}R).
    \label{eq:hqcs-inc-drift-decay}
\end{equation}
The average channel preserves density matrices. Therefore
Eq~\eqref{eq:hqcs-inc-drift-decay} with the bound of $|g(R)|$ implies
\begin{equation}
    |g((\overline{\mathcal T}^{(2)})^mR)|
       \le\varphi_{\max}\kappa_1^{m-1}.
    \label{eq:hqcs-inc-drift-decay-bound}
\end{equation}
Since $0<\kappa_1<1$, this uniform geometric bound makes the series
defining the correction absolutely convergent:
\begin{equation}
    G(R):=\sum_{m=0}^\infty g((\overline{\mathcal T}^{(2)})^mR)
         =g(R)+\frac{g(\overline{\mathcal T}^{(2)}R)}{1-\kappa_1}.
    \label{eq:hqcs-inc-poisson-definition}
\end{equation}
The closed form is a deterministic linear functional on matrices.
Using the bound of $|g(R)|$ and
Eq.~\eqref{eq:hqcs-inc-drift-decay-bound}, its bound on density matrices
is precisely the constant from Eq.~\eqref{eq:hqcs-inc-constants}:
\begin{equation}
    |G(R)|\le\varphi_{\max}\left(1+\frac1{1-\kappa_1}\right)=C_G.
    \label{eq:hqcs-inc-poisson-bound}
\end{equation}
Shifting the absolutely convergent series in
Eq.~\eqref{eq:hqcs-inc-poisson-definition} gives the cancellation
identity
\begin{equation}
    G(\overline{\mathcal T}^{(2)}R)=G(R)-g(R).
    \label{eq:hqcs-inc-poisson}
\end{equation}

Motivated by Eq.~\eqref{eq:hqcs-inc-poisson}, we next define
\begin{equation}
    D_t:=H_t+G(\mathsf R_t)-G(\mathsf R_{t-1}).
    \label{eq:hqcs-inc-martingale-definition}
\end{equation}
The linearity established in
Eq.~\eqref{eq:hqcs-inc-poisson-definition} and the conditional state
identity Eq.~\eqref{eq:hqcs-inc-conditional-state} yield
\begin{equation}
    \mathbb E_{U_t,z_t,z'_t\mid\mathcal F_{t-1}}[G(\mathsf R_t)]
       =G(\mathbb E_{U_t,z_t,z'_t\mid\mathcal F_{t-1}}[\mathsf R_t])
       =G(\overline{\mathcal T}^{(2)}\mathsf R_{t-1}).
    \label{eq:hqcs-inc-conditional-correction}
\end{equation}
Equivalently, the uniform absolute convergence in
Eq.~\eqref{eq:hqcs-inc-drift-decay-bound} justifies averaging the
series term by term.
Because $G(\mathsf R_{t-1})$ is already determined by the past,
substituting Eqs.~\eqref{eq:hqcs-inc-centered-drift} and
\eqref{eq:hqcs-inc-conditional-correction} into
Eq.~\eqref{eq:hqcs-inc-martingale-definition}, and then using
Eq.~\eqref{eq:hqcs-inc-poisson}, proves
\begin{equation}
    \mathbb E_{U_t,z_t,z'_t\mid\mathcal F_{t-1}}[D_t]
       =g(\mathsf R_{t-1})
          +G(\overline{\mathcal T}^{(2)}\mathsf R_{t-1})
          -G(\mathsf R_{t-1})=0.
    \label{eq:hqcs-inc-martingale-zero}
\end{equation}
Separately, the definition Eq.~\eqref{eq:hqcs-inc-martingale-definition}
and the bounds of $|H_t|$ and Eq.~\eqref{eq:hqcs-inc-poisson-bound} give
\begin{equation}
    |D_t|\le |H_t|+|G(\mathsf R_t)|+|G(\mathsf R_{t-1})|
          \le\varphi_{\max}+2C_G=K_D,
    \label{eq:hqcs-inc-martingale-bound}
\end{equation}
where $K_D$ is defined in Eq.~\eqref{eq:hqcs-inc-constants}.
Summing Eq.~\eqref{eq:hqcs-inc-martingale-definition} telescopes the
correction terms:
\begin{equation}
    \sum_{t=2}^T H_t
       =\sum_{t=2}^T D_t+G(\mathsf R_1)-G(\mathsf R_T).
    \label{eq:hqcs-inc-martingale}
\end{equation}
Eq.~\eqref{eq:hqcs-inc-poisson-bound} controls the remaining
boundary term independently of $T$ as
\begin{equation}
    |G(\mathsf R_1)-G(\mathsf R_T)|\le2C_G.
    \label{eq:hqcs-inc-martingale-boundary}
\end{equation}

Set $M_0=0$ and $M_n:=\sum_{t=2}^{n+1}D_t$ for $1\le n\le T-1$.
Then $M_n=M_{n-1}+D_{n+1}$ and $M_{n-1}$ is
$\mathcal F_n$-measurable. By the tower rule and the zero conditional
mean in Eq.~\eqref{eq:hqcs-inc-martingale-zero},
\begin{equation}
    \mathbb E_{\mathbf U}\mathbb E_{\mathbf z,\mathbf z'\sim Q^{\rm out}_{\mathbf U}}[M_{n-1}D_{n+1}]
       =\mathbb E_{\mathbf U}\mathbb E_{\mathbf z,\mathbf z'\sim Q^{\rm out}_{\mathbf U}}[M_{n-1}\mathbb E_{U_{n+1},z_{n+1},z'_{n+1}\mid\mathcal F_n}[D_{n+1}]]=0.
    \label{eq:hqcs-inc-martingale-orthogonality}
\end{equation}
The same argument eliminates $\mathbb E_{\mathbf U}\mathbb E_{\mathbf z,\mathbf z'\sim Q^{\rm out}_{\mathbf U}}[M_{n-1}^3D_{n+1}]$.
Expanding $M_n^2$ and using
Eqs.~\eqref{eq:hqcs-inc-martingale-orthogonality} and
\eqref{eq:hqcs-inc-martingale-bound} gives
\begin{equation}
    \mathbb E_{\mathbf U}\mathbb E_{\mathbf z,\mathbf z'\sim Q^{\rm out}_{\mathbf U}}[M_n^2]
       =\mathbb E_{\mathbf U}\mathbb E_{\mathbf z,\mathbf z'\sim Q^{\rm out}_{\mathbf U}}[M_{n-1}^2]+\mathbb E_{\mathbf U}\mathbb E_{\mathbf z,\mathbf z'\sim Q^{\rm out}_{\mathbf U}}[D_{n+1}^2]
       \le\mathbb E_{\mathbf U}\mathbb E_{\mathbf z,\mathbf z'\sim Q^{\rm out}_{\mathbf U}}[M_{n-1}^2]+K_D^2\le nK_D^2,
    \label{eq:hqcs-inc-martingale-second}
\end{equation}
where the last inequality follows for iteration.
Cauchy--Schwarz inequality and Eq.~\eqref{eq:hqcs-inc-martingale-second} then imply
\begin{equation}
    \mathbb E_{\mathbf U}\mathbb E_{\mathbf z,\mathbf z'\sim Q^{\rm out}_{\mathbf U}}[|M_{n-1}|]\le\sqrt{\mathbb E_{\mathbf U}\mathbb E_{\mathbf z,\mathbf z'\sim Q^{\rm out}_{\mathbf U}}[M_{n-1}^2]}
                         \le K_D\sqrt{n-1}.
    \label{eq:hqcs-inc-martingale-first}
\end{equation}
For the fourth moment, expand $(M_{n-1}+D_{n+1})^4$ conditional on
$\mathcal F_n$. The term linear in $D_{n+1}$ vanishes by
Eq.~\eqref{eq:hqcs-inc-martingale-zero}, leaving
\begin{align}
    \mathbb E_{U_{n+1},z_{n+1},z'_{n+1}\mid\mathcal F_n}[M_n^4]
    &=M_{n-1}^4
       +6M_{n-1}^2\mathbb E_{U_{n+1},z_{n+1},z'_{n+1}\mid\mathcal F_n}[D_{n+1}^2]
    +4M_{n-1}\mathbb E_{U_{n+1},z_{n+1},z'_{n+1}\mid\mathcal F_n}[D_{n+1}^3] \nonumber\\
       &\quad +\mathbb E_{U_{n+1},z_{n+1},z'_{n+1}\mid\mathcal F_n}[D_{n+1}^4].
    \label{eq:hqcs-inc-martingale-fourth-expansion}
\end{align}
The conditional third moment need not vanish.
Bounding its absolute value, and the second and fourth moments,
using Eq.~\eqref{eq:hqcs-inc-martingale-bound} yields
\begin{equation}
    \mathbb E_{U_{n+1},z_{n+1},z'_{n+1}\mid\mathcal F_n}[M_n^4]
       \le M_{n-1}^4+6K_D^2M_{n-1}^2
                         +4K_D^3|M_{n-1}|+K_D^4.
    \label{eq:hqcs-inc-martingale-fourth-step}
\end{equation}
Taking the full expectation and applying
Eqs.~\eqref{eq:hqcs-inc-martingale-second} and
\eqref{eq:hqcs-inc-martingale-first} gives the recursion
\begin{equation}
    \mathbb E_{\mathbf U}\mathbb E_{\mathbf z,\mathbf z'\sim Q^{\rm out}_{\mathbf U}}[M_n^4]\le\mathbb E_{\mathbf U}\mathbb E_{\mathbf z,\mathbf z'\sim Q^{\rm out}_{\mathbf U}}[M_{n-1}^4]
       +K_D^4[6(n-1)+4\sqrt{n-1}+1].
    \label{eq:hqcs-inc-martingale-fourth-recursion}
\end{equation}
Summing Eq.~\eqref{eq:hqcs-inc-martingale-fourth-recursion} from
$1$ to $n$, with $M_0=0$, gives
\begin{equation}
    \mathbb E_{\mathbf U}\mathbb E_{\mathbf z,\mathbf z'\sim Q^{\rm out}_{\mathbf U}}[M_n^4]
       \le K_D^4\left[3n(n-1)+4\sum_{j=0}^{n-1}\sqrt j+n\right]
       \le7K_D^4n^2.
    \label{eq:hqcs-inc-martingale-fourth-sum}
\end{equation}
For the last inequality, use
$\sum_{j=0}^{n-1}\sqrt j\le n\sqrt n\le n^2$ for $n\ge1$.

Eq.~\eqref{eq:hqcs-inc-martingale} writes the original sum as
$M_{T-1}$ plus the bounded boundary term from
Eq.~\eqref{eq:hqcs-inc-martingale-boundary}.
Using $|a+b|^4\le8(|a|^4+|b|^4)$ and then the fourth-moment bound
Eq.~\eqref{eq:hqcs-inc-martingale-fourth-sum}, we obtain
\begin{align}
    \mathbb E_{\mathbf U}\mathbb E_{\mathbf z,\mathbf z'\sim Q^{\rm out}_{\mathbf U}}\!\left[\left(\sum_{t=2}^T H_t\right)^4\right]
    &\le8\mathbb E_{\mathbf U}\mathbb E_{\mathbf z,\mathbf z'\sim Q^{\rm out}_{\mathbf U}}[M_{T-1}^4]+128C_G^4\nonumber\\
    &\le8K_D^4\left[3(T-1)(T-2)+4\sum_{j=0}^{T-2}\sqrt j+(T-1)\right]
         +128C_G^4.
    \label{eq:hqcs-inc-pair-sum-fourth}
\end{align}
Finally, divide Eq.~\eqref{eq:hqcs-inc-pair-sum-fourth} by $(T-1)^4$
using Eq.~\eqref{eq:hqcs-inc-fourth-double-reduction}. This proves the first bound in
Eq.~\eqref{eq:hqcs-inc-fourth-double}. The estimate
$\mathbb E_{\mathbf U}\mathbb E_{\mathbf z,\mathbf z'\sim Q^{\rm out}_{\mathbf U}}[M_{T-1}^4]\le7K_D^4(T-1)^2$ in
Eq.~\eqref{eq:hqcs-inc-martingale-fourth-sum} gives its second bound.

\end{proof}

The mean errors and fourth-moment bounds now give the two ingredients
needed to control a Taylor remainder. The next lemma converts these
estimates into an expansion of a conditional exponential average.

\begin{lemma}[Conditional exponential linearization]
\label{lem:hqcs-inc-linearization}
A nonnegative random cost $S$ is a function of a sampling variable
$Z$ whose distribution conditional on $\mathbf U$ is $\nu_{\mathbf U}$.
For a deterministic $a\ge0$, suppose that
\begin{align}
    \mathbb E_{Z\sim\nu_{\mathbf U}}[S]
       =a+\eta(\mathbf U)+\varepsilon(\mathbf U),
    \label{eq:hqcs-inc-linearization-conditions}
\end{align}
with
\begin{equation}
	\sqrt{\mathbb E_{\mathbf U}[|\varepsilon|^2]}
	\le\frac{C_\varepsilon}{T-1},
	\qquad
	\mathbb E_{\mathbf U}\mathbb E_{Z\sim\nu_{\mathbf U}}[(S-a)^4]
	\le\frac{C_4}{(T-1)^2}.
\end{equation}
Then there is an exact representation
\begin{equation}
    \mathbb E_{Z\sim\nu_{\mathbf U}}[e^{-S}]
       =e^{-a}-e^{-a}\eta(\mathbf U)+r(\mathbf U),
    \qquad
    \sqrt{\mathbb E_{\mathbf U}[|r|^2]}
       \le\frac{C_\varepsilon+\tfrac12\sqrt{C_4}}{T-1}.
    \label{eq:hqcs-inc-linearization-result}
\end{equation}
\end{lemma}

\begin{proof}
The integral Taylor identity, valid for every $S,a\ge0$, is
\begin{align}
    e^{-S}
    &=e^{-a}-e^{-a}(S-a)
      +(S-a)^2\int_0^1(1-u)e^{-[a+u(S-a)]}\,du.
    \label{eq:hqcs-inc-taylor}
\end{align}
Every point $a+u(S-a)$ on the segment is nonnegative, then 
the last term therefore has absolute value at most
$(S-a)^2/2$.
Conditional averaging and
Eq.~\eqref{eq:hqcs-inc-linearization-conditions} define
$r(\mathbf U)$ as
\begin{equation}
	r(\mathbf U) \coloneqq -e^{-a} \varepsilon(\mathbf U) + \E_{Z \sim \nu_{\mathbf U}}\left[(S-a)^2\int_0^1(1-u)e^{-[a+u(S-a)]}\,du\right].
\end{equation}
Conditional Jensen and the scalar mean-square triangle inequality
give
\begin{align}
    \sqrt{\mathbb E_{\mathbf U}[|r|^2]}
    &\le e^{-a}\sqrt{\mathbb E_{\mathbf U}[|\varepsilon|^2]}
       +\frac12
        \sqrt{\mathbb E_{\mathbf U}
          \left[(\mathbb E_{Z\sim\nu_{\mathbf U}}[(S-a)^2])^2\right]}
       \nonumber\\
    &\le\sqrt{\mathbb E_{\mathbf U}[|\varepsilon|^2]}
       +\frac12
        \sqrt{\mathbb E_{\mathbf U}
                   \mathbb E_{Z\sim\nu_{\mathbf U}}[(S-a)^4]}
       \nonumber\\
    &\le\frac{C_\varepsilon+\tfrac12\sqrt{C_4}}{T-1}.
\end{align}
\end{proof}

The four lemmas now provide the local mean representations and
remainder estimates needed to complete the proof of
Theorem~\ref{thm:hqcs-inc-variance}. Applying the linearization to
both circuit-dependent loss terms gives the decomposition, from which
the variance bounds follow.

\begin{proof}[Proof of Theorem~\ref{thm:hqcs-inc-variance}]
Lemmas~\ref{lem:hqcs-inc-means} and~\ref{lem:hqcs-inc-fourth} verify
the conditions of Lemma~\ref{lem:hqcs-inc-linearization} with the
following choices (the target--model column is for each fixed $x$), summarized as follows.
\begin{equation}
\renewcommand{\arraystretch}{1.8}
\begin{array}{c|c|c}
& \text{Model--model} & \text{Target--model} \\ \hline
Z & (Y,Y') & Y \\
\nu_{\mathbf U} & Q_{\mathbf U}\otimes Q_{\mathbf U} & Q_{\mathbf U} \\
S & S_{QQ} & S_{xQ} \\
a & \overline{h}_{QQ}
& \displaystyle \frac1{T-1}\sum_{t=2}^T\bar h_{\Delta x_t} \\
\eta(\mathbf U)
& \displaystyle\frac1{T-1}\sum_{t=2}^T
\bigl(h_{QQ}(U_t)-\overline{h}_{QQ}\bigr)
& \displaystyle\frac1{T-1}\sum_{t=2}^T
\bigl(h_{\Delta x_t}(U_t)-\bar h_{\Delta x_t}\bigr) \\
\varepsilon(\mathbf U) & \varepsilon_B(\mathbf U) & \varepsilon_x(\mathbf U) \\
C_\varepsilon & 2\varphi_{\max}C_{\rm init}
& \varphi_{\max}C_{\rm init} \\
C_4 & 56K_D^4+128C_G^4 & 288\varphi_{\max}^4
\end{array}
\label{eq:hqcs-inc-linearization-choices}
\end{equation}
The resulting exact expansions are
\begin{align}
    B_T^\Delta(\mathbf U)
    &=e^{-\overline{h}_{QQ}}-\frac{e^{-\overline{h}_{QQ}}}{T-1}
            \sum_{t=2}^T(h_{QQ}(U_t)-\overline{h}_{QQ})+r_B(\mathbf U),
       \nonumber\\
    \mathbb E_{Y\sim Q_{\mathbf U}}[e^{-S_{xQ}}]
    &=e^{-m_T(x)}-\frac{e^{-\frac1{T-1}\sum_{t=2}^T\bar h_{\Delta x_t}}}{T-1}
            \sum_{t=2}^T(h_{\Delta x_t}(U_t)-\bar h_{\Delta x_t})
       +r_x(\mathbf U),
       \label{eq:hqcs-inc-kernel-expansions}
\end{align}
where the remainders are bounded as
\begin{align}
	\sqrt{\mathbb E_{\mathbf U}[|r_B|^2]}
	&\le\frac{2\varphi_{\max}C_{\rm init}
		+\tfrac12\sqrt{56K_D^4+128C_G^4}}{T-1},
	\nonumber\\
	\sup_x\sqrt{\mathbb E_{\mathbf U}[|r_x|^2]}
	&\le\frac{\varphi_{\max}C_{\rm init}
		+\tfrac12\sqrt{288}\,\varphi_{\max}^2}{T-1}.
	\label{eq:hqcs-inc-kernel-remainders}
\end{align}
The Jensen inequality and Eq.~\eqref{eq:hqcs-inc-kernel-remainders} give
\begin{align}
    \mathbb E_{\mathbf U}
       [|\mathbb E_{X\sim P_T}[r_X(\mathbf U)]|^2]
    &\le\mathbb E_{X\sim P_T}\mathbb E_{\mathbf U}
       [|r_X(\mathbf U)|^2]
       \nonumber\\
    &\le
       \frac{(\varphi_{\max}C_{\rm init}
                    +\tfrac12\sqrt{288}\,\varphi_{\max}^2)^2}{(T-1)^2}.
    \label{eq:hqcs-inc-target-remainder}
\end{align}
Combination with Eq.~\eqref{eq:hqcs-inc-loss} gives
Eq.~\eqref{eq:hqcs-inc-loss-expansion}, with the exact remainder
\begin{align}
    R_T(\mathbf U)=r_B(\mathbf U)
                   -2\mathbb E_{X\sim P_T}[r_X(\mathbf U)],
\end{align}
satisfying
\begin{align}
    \sqrt{\mathbb E_{\mathbf U}[|R_T|^2]}
       \le\sqrt{\mathbb E_{\mathbf U}[|r_B|^2]}
           +2\sqrt{\mathbb E_{\mathbf U}
               [|\mathbb E_{X\sim P_T}[r_X(\mathbf U)]|^2]}
        \le\frac{C_{\rm rem}}{T-1}.
    \label{eq:hqcs-inc-total-remainder-bound}
\end{align}
All correlations between the model--model and target--model
contributions are retained in $\mathcal{I}_{t,T}$.

The centering in Eq.~\eqref{eq:hqcs-inc-local-costs} implies
\begin{align}
    \mathbb E_{U\sim{\rm Haar}}[\mathcal{I}_{t,T}(U)]
       &=-e^{-\overline{h}_{QQ}}
             (\mathbb E_{U\sim{\rm Haar}}[h_{QQ}(U)]-\overline{h}_{QQ})
      +2\mathbb E_{X\sim P_T}
         \left[e^{-m_T(X)}
          (\mathbb E_{U\sim{\rm Haar}}[h_{\Delta X_t}(U)]
                                      -\bar h_{\Delta X_t})\right]
       =0.
    \label{eq:hqcs-inc-centered-leading}
\end{align}
As $|h_{QQ}(U) - \overline{h}_{QQ}|\le \varphi_{\max}$ and $|h_{\Delta X_t}(U)
-\bar h_{\Delta X_t}| \le \varphi_{\max}$, we have $|\mathcal{I}_{t,T}(U)|\le3\varphi_{\max}$. The definition in Eq.~\eqref{eq:hqcs-inc-leading} then gives
\begin{equation*}
    0\le\Gamma_T
    =\frac1{T-1}\sum_{t=2}^T
        \operatorname{Var}_{U\sim{\rm Haar}}[\mathcal{I}_{t,T}(U)]
    \le\frac1{T-1}\sum_{t=2}^T
        \mathbb E_{U\sim{\rm Haar}}[|\mathcal{I}_{t,T}(U)|^2]
    \le9\varphi_{\max}^2.
\end{equation*}
Each $\mathcal{I}_{t,T}(U_t)$ is a deterministic function of its own
independent Haar block. Independence and centering therefore give
\begin{align}
    \operatorname{Var}_{\mathbf U}
       \left[\frac1{T-1}\sum_{t=2}^T\mathcal{I}_{t,T}(U_t)\right]
       &=\frac1{(T-1)^2}\sum_{t=2}^T
          \operatorname{Var}_{U_t\sim{\rm Haar}}[\mathcal{I}_{t,T}(U_t)]
        =\frac{\Gamma_T}{T-1}.
    \label{eq:hqcs-inc-leading-variance}
\end{align}
Since $\ell_T$ is deterministic, expanding the variance of the exact
loss decomposition in Eq.~\eqref{eq:hqcs-inc-loss-expansion} gives
\begin{align}
    \operatorname{Var}_{\mathbf U}[\mathcal L_{\rm MMD}^{\Delta}]
    =\operatorname{Var}_{\mathbf U}
       \left[\frac1{T-1}\sum_{t=2}^T\mathcal I_{t,T}(U_t)\right]
       +\operatorname{Var}_{\mathbf U}[R_T]
    +2\operatorname{Cov}_{\mathbf U}
       \left(\frac1{T-1}\sum_{t=2}^T\mathcal I_{t,T}(U_t),R_T\right).
    \label{eq:hqcs-inc-full-variance-expansion}
\end{align}
The first term equals $\Gamma_T/(T-1)$ by
Eq.~\eqref{eq:hqcs-inc-leading-variance}.
For the second term, the mean-square remainder bound in
Eq.~\eqref{eq:hqcs-inc-total-remainder-bound} gives
\begin{equation}
    \operatorname{Var}_{\mathbf U}[R_T]
       \le\mathbb E_{\mathbf U}[|R_T|^2]
       \le\frac{C_{\rm rem}^2}{(T-1)^2}.
    \label{eq:hqcs-inc-remainder-variance-bound}
\end{equation}
For the covariance term, Cauchy--Schwarz together with
Eqs.~\eqref{eq:hqcs-inc-leading-variance} and
\eqref{eq:hqcs-inc-remainder-variance-bound} yields
\begin{equation}
    \left|\operatorname{Cov}_{\mathbf U}
        \left(\frac1{T-1}\sum_{t=2}^T\mathcal I_{t,T}(U_t),R_T\right)\right|
       \le\sqrt{\frac{\Gamma_T}{T-1}}\,
          \sqrt{\operatorname{Var}_{\mathbf U}[R_T]}
       \le\frac{C_{\rm rem}\sqrt{\Gamma_T}}{(T-1)^{3/2}}.
    \label{eq:hqcs-inc-covariance-bound}
\end{equation}
Applying the triangle inequality to
Eq.~\eqref{eq:hqcs-inc-full-variance-expansion} and using
Eqs.~\eqref{eq:hqcs-inc-remainder-variance-bound} and
\eqref{eq:hqcs-inc-covariance-bound} gives
\begin{equation*}
    \left|\operatorname{Var}_{\mathbf U}[\mathcal L_{\rm MMD}^{\Delta}]
          -\frac{\Gamma_T}{T-1}\right|
       \le\frac{2C_{\rm rem}\sqrt{\Gamma_T}}{(T-1)^{3/2}}
          +\frac{C_{\rm rem}^2}{(T-1)^2},
\end{equation*}
which proves Eq.~\eqref{eq:hqcs-inc-variance-error}. Substituting
$\Gamma_T\le9\varphi_{\max}^2$ proves
Eq.~\eqref{eq:hqcs-inc-variance-upper-bound}.
When $\Gamma_T\ge\gamma>0$ and $T-1 \ge 25 C_{\rm rem}^2/\gamma$, division of the error bound by
$\Gamma_T/(T-1)$ gives
\begin{equation}
    \frac{2C_{\rm rem}}{\sqrt{\Gamma_T(T-1)}}
       +\frac{C_{\rm rem}^2}{\Gamma_T(T-1)}
       \le\frac25+\frac1{25}=\frac{11}{25}<\frac12
\end{equation}
proving
Eq.~\eqref{eq:hqcs-inc-nondegenerate}.
\end{proof}

For the Gaussian discrepancy, the leading variance coefficient can
be evaluated explicitly. This gives a condition on the target
increment energy under which Theorem~\ref{thm:hqcs-inc-variance}
provides variance bounds of order $1/(T-1)$.

For fixed $\sigma>0$, the choice
$\varphi(r)=r^2/(2\sigma^2)$ gives
\begin{equation}
    K_T^\Delta(x,\tilde x)
       =\exp[-\|\Delta x-\Delta\tilde x\|_2^2/(2(T-1)\sigma^2)],
    \qquad \varphi_{\max}=2/\sigma^2.
    \label{eq:hqcs-inc-gaussian-kernel}
\end{equation}
For readout map $f$, we define 
\begin{align}
    \mu&:=\frac1{d_B}\sum_z f(z),
    \qquad v:=\frac1{d_B}\sum_z(f(z)-\mu)^2,
    \qquad J_T(x):=\frac1{T-1}\sum_{t=2}^T(\Delta x_t)^2,
       \nonumber\\
    \alpha_T&:=\frac2{\sigma^2}\left(
       e^{-2v/\sigma^2}
       -e^{-v/\sigma^2}
          \mathbb E_{X\sim P_T}[e^{-J_T(X)/(2\sigma^2)}]
       \right).
    \label{eq:hqcs-inc-gaussian-centers}
\end{align}
For the proof, define the centered readout operator
$Z:=\sum_z(f(z)-\mu)\Pi_z$ and the local reference covariance
$c(U):=d^{-1}\operatorname{Tr}(ZUZU^\dagger)$.

\begin{corollary}[Two-sided variance bounds for the Gaussian kernel]
\label{cor:hqcs-inc-gaussian}
Under the assumptions of Theorem~\ref{thm:hqcs-inc-variance}, use the
Gaussian kernel in Eq.~\eqref{eq:hqcs-inc-gaussian-kernel} with fixed
$\sigma>0$ and a nonconstant readout map, so that $v>0$.
Suppose that the target distribution $P_T$ and readout map $f$ satisfy, for a constant
$\varepsilon_0>0$ independent of $T$,
\begin{equation}
    \left|
      \mathbb E_{X\sim P_T}[e^{-J_T(X)/(2\sigma^2)}]
      -e^{-v/\sigma^2}
    \right|\ge\varepsilon_0.
    \label{eq:hqcs-inc-gaussian-target-separation}
\end{equation}
Then the leading coefficient obeys
\begin{equation}
    \Gamma_T=\frac{\alpha_T^2v^2}{d^2-1}\ge\gamma,
    \qquad
    \gamma:=\frac{4v^2e^{-2v/\sigma^2}\varepsilon_0^2}
                     {\sigma^4(d^2-1)}>0.
    \label{eq:hqcs-inc-gaussian-gamma}
\end{equation}
For every $T$ satisfying $T-1\ge25C_{\rm rem}^2/\gamma$,
where $C_{\rm rem}$ is evaluated with $\varphi_{\max}=2/\sigma^2$,
\begin{equation}
    \frac{\Gamma_T}{2(T-1)}
       \le\operatorname{Var}_{\mathbf U}[\mathcal L_{\rm MMD}^{\Delta}]
       \le\frac{3\Gamma_T}{2(T-1)}.
    \label{eq:hqcs-inc-gaussian-two-sided}
\end{equation}
\end{corollary}

\begin{proof}
The centered operator satisfies
$\operatorname{Tr}Z=0$ and $\operatorname{Tr}Z^2=dv$.
Expansion using Eq.~\eqref{eq:hqcs-inc-reference} gives
\begin{align}
    c(U)&=\sum_{a,b}p_U(a,b)(f(a)-\mu)(f(b)-\mu),
       \nonumber\\
    \mathbb E_{(a,b)\sim p_U}[\delta_{ab}]&=0,
    \qquad
    \mathbb E_{(a,b)\sim p_U}[\delta_{ab}^2]=2v-2c(U).
    \label{eq:hqcs-inc-gaussian-local-moments}
\end{align}
Here $U$ is fixed and the expectations average only local reference
outcomes. The second line follows from the uniform marginals in
Eq.~\eqref{eq:hqcs-inc-uniform-marginals}:
both centered marginal second moments equal $v$, and their cross
moment is $c(U)$.
Two independent local outcome pairs under that same $U$ have
zero mean increments. Substituting the Gaussian cost into the
local cost definitions in Eq.~\eqref{eq:hqcs-inc-local-costs}
and using Eq.~\eqref{eq:hqcs-inc-gaussian-local-moments} gives
\begin{align}
    h_\xi(U)
       &=\frac{\xi^2-2\xi
                 \mathbb E_{(a,b)\sim p_U}[\delta_{ab}]
                  +\mathbb E_{(a,b)\sim p_U}[\delta_{ab}^2]}
               {2\sigma^2}
        =\frac{\xi^2+2v-2c(U)}{2\sigma^2},
       \nonumber\\
    h_{QQ}(U)&=\frac1{2\sigma^2}
       \mathbb E_{(a,b),(c,e)
                            \sim p_U}
       [(\delta_{ab}-\delta_{ce})^2]
        =\frac{4v-4c(U)}{2\sigma^2}.
    \label{eq:hqcs-inc-gaussian-local-costs}
\end{align}
The first Haar moment gives
$\mathbb E_{U\sim{\rm Haar}}[c(U)]=0$.
For its second moment, Eq.~\eqref{eq:hqcs-inc-twirl} with
$R=Z\otimes Z$ gives
\begin{align}
    \mathbb E_{U\sim{\rm Haar}}
       [(UZU^\dagger)\otimes(UZU^\dagger)]
       &=\frac{dv}{d^2-1}(F-I_2/d),
       \nonumber\\
    \mathbb E_{U\sim{\rm Haar}}[c(U)^2]
       &=\frac1{d^2}\operatorname{Tr}
          \left[(Z\otimes Z)\frac{dv}{d^2-1}(F-I_2/d)\right]
       \nonumber\\
       &=\frac1{d^2}\frac{dv}{d^2-1}
                   \operatorname{Tr}Z^2
        =\frac{v^2}{d^2-1}.
    \label{eq:hqcs-inc-gaussian-covariance-variance}
\end{align}
The vanishing identity contribution uses $\operatorname{Tr}Z=0$,
and the swap contribution uses
$\operatorname{Tr}[(Z\otimes Z)F]=\operatorname{Tr}Z^2$.
Eqs.~\eqref{eq:hqcs-inc-gaussian-local-costs}
and~\eqref{eq:hqcs-inc-local-costs} then imply
\begin{align}
    \overline{h}_{QQ}&=\frac{2v}{\sigma^2},
    \qquad
    \bar h_\xi=\frac{\xi^2+2v}{2\sigma^2},
    \qquad
    m_T(x)=\frac{J_T(x)+2v}{2\sigma^2},
       \nonumber\\
    h_{QQ}(U)-\overline{h}_{QQ}&=-\frac{2}{\sigma^2}c(U),
    \qquad
    h_\xi(U)-\bar h_\xi=-\frac1{\sigma^2}c(U),
       \nonumber\\
    \mathcal{I}_{t,T}(U)
       &=\frac{2e^{-2v/\sigma^2}}{\sigma^2}c(U)
          -\frac2{\sigma^2}
               \mathbb E_{X\sim P_T}[e^{-m_T(X)}]c(U)
        =\alpha_Tc(U).
   \label{eq:hqcs-inc-gaussian-centered-costs}
\end{align}
Substituting Eq.~\eqref{eq:hqcs-inc-gaussian-centered-costs} into
the definition of $\Gamma_T$ in Eq.~\eqref{eq:hqcs-inc-leading},
and using $\mathbb E_{U\sim{\rm Haar}}[c(U)]=0$ and
Eq.~\eqref{eq:hqcs-inc-gaussian-covariance-variance}, gives
\begin{align}
   \Gamma_T
      &=\frac1{T-1}\sum_{t=2}^T
          \operatorname{Var}_{U\sim{\rm Haar}}[\mathcal I_{t,T}(U)]
       =\frac{\alpha_T^2}{T-1}\sum_{t=2}^T
          \operatorname{Var}_{U\sim{\rm Haar}}[c(U)]\nonumber\\
      &=\alpha_T^2\mathbb E_{U\sim{\rm Haar}}[c(U)^2]
       =\frac{\alpha_T^2v^2}{d^2-1}.
   \label{eq:hqcs-inc-gaussian-gamma-derivation}
\end{align}
Theorem~\ref{thm:hqcs-inc-variance} gives the loss expansion and
Eq.~\eqref{eq:hqcs-inc-gaussian-variance}.

In particular, the variance approximation is
\begin{align}
    \left|\operatorname{Var}_{\mathbf U}
          [\mathcal L_{\rm MMD}^\Delta]
           -\frac{\alpha_T^2v^2}{(d^2-1)(T-1)}\right|
       \le
       \frac{2C_{\rm rem}|\alpha_T|v}
             {\sqrt{d^2-1}\,(T-1)^{3/2}}
          +\frac{C_{\rm rem}^2}{(T-1)^2}.
    \label{eq:hqcs-inc-gaussian-variance}
\end{align}
The target separation in
Eq.~\eqref{eq:hqcs-inc-gaussian-target-separation} and the definition
of $\alpha_T$ in Eq.~\eqref{eq:hqcs-inc-gaussian-centers} imply
\begin{equation}
    |\alpha_T|
       =\frac{2e^{-v/\sigma^2}}{\sigma^2}
          \left|e^{-v/\sigma^2}
             -\mathbb E_{X\sim P_T}[e^{-J_T(X)/(2\sigma^2)}]\right|
       \ge\frac{2e^{-v/\sigma^2}}{\sigma^2}\varepsilon_0.
    \label{eq:hqcs-inc-gaussian-alpha-lower}
\end{equation}
Together with
Eq.~\eqref{eq:hqcs-inc-gaussian-gamma-derivation}, this proves
$\Gamma_T\ge\gamma$ in Eq.~\eqref{eq:hqcs-inc-gaussian-gamma}.
The length condition then allows
Eq.~\eqref{eq:hqcs-inc-nondegenerate} of
Theorem~\ref{thm:hqcs-inc-variance} to give
Eq.~\eqref{eq:hqcs-inc-gaussian-two-sided}.
Finally, $\gamma\le\Gamma_T\le9\varphi_{\max}^2$ and the uniform
remainder constant give the stated $\Theta(1/T)$ scaling.
\end{proof}

\section{Proof of Theorem~\ref{thm:temporal_resource_separation_main} in the main text}
\label{proof:resource separation}

In this section, we provide the formal version and proof of Theorem~\ref{thm:temporal_resource_separation_main} in the main text. The formal theorem is
\begin{theorem}
\label{thm:temporal_resource_separation}
({\bf Resource separation from temporal parameterization})
For every $T\ge 3$, there exists a nonstationary first-order binary
Markov process $P_T$ with uniform single-time marginals such that:
\begin{enumerate}
\item HQCS generates $P_T$ exactly using one physical qubit,
$T-1$ independently adjustable parameters, one elementary gate per
temporal block, and $O(T)$ total gate count and trajectory depth.

\item Every generator with shared unitary block across temporal steps that generates
$P_T$ exactly requires at least $T-1$ independent continuous
parameter settings in its shared block.
If this block has $G$ elementary gates, depth $L$, and $N$ qubits,
then
\begin{equation}
    G\ge \frac{T-1}{c},
    \qquad
    NL\ge \frac{T-1}{c},
    \label{eq:existence_width_depth}
\end{equation}
where $c$ is the number of parameter in each elementary gate.
Consequently, its prescribed repeated-block execution requires
$\Omega(T^2)$ total gates and a trajectory depth of
$\Omega(T^2/N)$.
\end{enumerate}
\end{theorem}
We prove Theorem~\ref{thm:temporal_resource_separation} in a generic
form: the process $P_T$ can be taken to be any member of an explicit
$(T-1)$-parameter family, outside a Lebesgue-null set. We then prove the
two statements made after the theorem, on clock-controlled rotations
and on the attainability of the gate-count bound. 

Throughout, a \emph{layer} is a set of elementary gates acting on
pairwise disjoint qubits, so that a layer on $N$ qubits contains at most
$N$ gates. For a generator of length-$T$ trajectories, $G_{\rm traj}$
denotes the total number of elementary gates executed to produce one
trajectory, and the trajectory depth $D_{\rm traj}$ is the number of
layers of the resulting sequential circuit. Initialization,
measurement, and classical feed-forward occupy additional steps and can
only increase both quantities.

\subsection{Target family}

For a schedule
$\boldsymbol q=(q_2,\ldots,q_T)\in\mathcal S_T:=(1/4,3/4)^{T-1}$, let
$P_{\boldsymbol q}$ be the distribution of the binary process
\begin{equation}
    X_1\sim\operatorname{Bernoulli}(1/2),
    \qquad
    X_t=X_{t-1}\oplus B_t,
    \label{eq:rs-flip-process}
\end{equation}
for $t=2,\ldots,T$, where $B_2,\ldots,B_T$ are independent of each
other and of $X_1$, with $B_t\sim\operatorname{Bernoulli}(q_t)$. With
$b_t(\boldsymbol x)=x_t\oplus x_{t-1}$, its trajectory distribution is
\begin{equation}
    P_{\boldsymbol q}(\boldsymbol x)
    =\frac12\prod_{t=2}^{T}
     q_t^{\,b_t(\boldsymbol x)}(1-q_t)^{1-b_t(\boldsymbol x)},
    \qquad \boldsymbol x\in\{0,1\}^T .
    \label{eq:rs-law}
\end{equation}
The conditional distribution of $X_t$ given $X_1,\ldots,X_{t-1}$ depends only on
$X_{t-1}$, so $P_{\boldsymbol q}$ is a first-order Markov process.
Since $X_1$ is uniform and $B_t$ is independent of $X_{t-1}$, every
$X_t$ is uniform. For any distribution $Q$ on $\{0,1\}^T$, define the
flip probabilities
\begin{equation}
    F_t(Q)=\sum_{\boldsymbol x\in\{0,1\}^T}b_t(\boldsymbol x)\,Q(\boldsymbol x),
    \qquad t=2,\ldots,T .
    \label{eq:rs-flip-map}
\end{equation}
Then $F_t(P_{\boldsymbol q})=q_t$, or equivalently
$\mathbb E_{P_{\boldsymbol q}}[(-1)^{X_{t-1}+X_t}]=1-2q_t$, so the
schedule is determined by the process. If the $q_t$ are not all equal,
the joint distribution of $(X_{t-1},X_t)$ changes with $t$ and
$P_{\boldsymbol q}$ is nonstationary. The interval $(1/4,3/4)$ keeps all
trajectory probabilities positive; any open set of schedules would
serve equally well.

\subsection{Item 1: HQCS construction}

Consider HQCS with $N_A=0$ and $N_B=1$, whose single qubit is
initialized in $\lvert0\rangle$, and choose the blocks
\begin{equation}
    U_1=H,
    \qquad
    U_t=R_y(\alpha_t),
    \qquad
    \sin^2\frac{\alpha_t}{2}=q_t,
    \label{eq:rs-hqcs-gates}
\end{equation}
for $t=2,\ldots,T$, with $R_y(\alpha)=e^{-i\alpha Y/2}$. After each
block the qubit is
measured in the computational basis, and the post-measurement state is
reused. The first outcome is uniform. Because
\begin{equation}
    R_y(\alpha)\lvert x\rangle
    =\cos\frac{\alpha}{2}\lvert x\rangle
    +(-1)^{x}\sin\frac{\alpha}{2}\lvert x\oplus1\rangle,
    \qquad x\in\{0,1\},
\end{equation}
the outcome at step $t\ge2$ differs from $x_{t-1}$ with probability
$\sin^2(\alpha_t/2)=q_t$, independently of $x_{t-1}$ and of the earlier
outcomes, which influence the future only through the post-measurement
state $\lvert x_{t-1}\rangle$. The trajectory distribution is therefore
Eq.~\eqref{eq:rs-law}. The circuit uses $T$ elementary gates, one per
block, and $T$ measurements, so the total gate count and the trajectory
depth are both $O(T)$. The map $(\alpha_2,\ldots,\alpha_T)\mapsto
\boldsymbol q$ is a diffeomorphism from $(\pi/3,2\pi/3)^{T-1}$ onto
$\mathcal S_T$, so the $T-1$ angles are independently adjustable and
cover the whole family.

\subsection{Item 2: dimension counting}

Fix a layout $\ell$ of a shared generator. After imposing the
parameter-sharing constraints, the parameters range over a
second-countable smooth manifold $\mathcal M_\ell$ of dimension
$p_\ell$. The continuous settings of the $G$ gates of the block that
are not fixed by the layout are given by a smooth map
$s_\ell:\mathcal M_\ell\to\mathbb R^{k_\ell}$ with $k_\ell\le cG$; for
instance, parameter sharing assigns the same coordinate of
$\boldsymbol\vartheta$ to several settings. We assume that the classical encoding sets each
executed gate setting to a fixed smooth function of these settings and
of the previously sampled symbols. Let $Q_{\ell,\boldsymbol\vartheta}$
denote the resulting distribution of the trajectory
$\boldsymbol x\in\{0,1\}^T$.

\begin{lemma}[Smoothness]
\label{lem:rs-smooth}
For every $\boldsymbol x\in\{0,1\}^T$, the probability
$Q_{\ell,\boldsymbol\vartheta}(\boldsymbol x)$ is a smooth function of
$\boldsymbol\vartheta$ that depends on $\boldsymbol\vartheta$ only
through $s_\ell(\boldsymbol\vartheta)$.
\end{lemma}

\begin{proof}
Fix a complete record $r$ of all measurement outcomes produced while
generating a trajectory. Along $r$, the previously sampled symbols are
determined at every step, so each executed gate is a fixed gate or a
member of one of the smooth gate families with settings that are fixed
smooth functions of $s_\ell(\boldsymbol\vartheta)$. The probability of
$r$ has the form
$\Tr[\mathcal E_{r_K}\circ\mathcal V_K\circ\cdots\circ
\mathcal E_{r_1}\circ\mathcal V_1(\rho_{\rm init})]$, where the
$\mathcal V_i$ are conjugations by products of these gates and the
$\mathcal E_{r_i}$ are the fixed measurement operations. It is
therefore a polynomial in the matrix elements of the gates and a smooth
function of $s_\ell(\boldsymbol\vartheta)$. Finally,
$Q_{\ell,\boldsymbol\vartheta}(\boldsymbol x)$ is the finite sum of
these probabilities over the records that the fixed readout maps to
$\boldsymbol x$.
\end{proof}

Define the flip map $F_\ell:\mathcal M_\ell\to\mathbb R^{T-1}$ by
$F_\ell(\boldsymbol\vartheta)=\bigl(F_2(Q_{\ell,\boldsymbol\vartheta}),
\ldots,F_T(Q_{\ell,\boldsymbol\vartheta})\bigr)$. By
Lemma~\ref{lem:rs-smooth}, $F_\ell=\widetilde F_\ell\circ s_\ell$ for a
smooth map $\widetilde F_\ell$ defined on an open subset of
$\mathbb R^{k_\ell}$. If the generator produces $P_{\boldsymbol q}$
exactly at $\boldsymbol\vartheta^\star$, then
Eq.~\eqref{eq:rs-flip-map} gives
\begin{equation}
    F_\ell(\boldsymbol\vartheta^\star)=\boldsymbol q .
    \label{eq:rs-exact}
\end{equation}
In local parameter coordinates $\boldsymbol\vartheta$, write the Jacobians
of $F_\ell$ and $s_\ell$ as $\partial_{\boldsymbol\vartheta}F_\ell$ and
$\partial_{\boldsymbol\vartheta}s_\ell$, respectively; their ranks are
independent of the chosen coordinates.
A point $\boldsymbol\vartheta$ is critical for $F_\ell$ if
$\operatorname{rank}\partial_{\boldsymbol\vartheta}F_\ell(\boldsymbol\vartheta)<T-1$, and
$\operatorname{Crit}_\ell\subset\mathbb R^{T-1}$ denotes the set of
critical values, i.e., images of critical points. By Sard's
theorem~\cite{lee2013smooth}, $\operatorname{Crit}_\ell$ has Lebesgue
measure zero. If $p_\ell<T-1$, every point is critical, and the entire
image $F_\ell(\mathcal M_\ell)$ is null. Since the set of layouts is
countable,
\begin{equation}
    \mathcal N_T
    :=\bigcup_\ell\operatorname{Crit}_\ell
    \;\cup\;
    \{\boldsymbol q\in\mathcal S_T:\ q_2=\cdots=q_T\}
    \label{eq:rs-null-set}
\end{equation}
is a Lebesgue-null subset of $\mathcal S_T$; the second set is a line
segment, and $T-1\ge2$.

Now let $\boldsymbol q\in\mathcal S_T\setminus\mathcal N_T$, and suppose
that a shared generator with layout $\ell$ generates $P_{\boldsymbol q}$
exactly at $\boldsymbol\vartheta^\star$. Because
$\boldsymbol q\notin\operatorname{Crit}_\ell$, the point
$\boldsymbol\vartheta^\star$ is regular, and the chain rule gives
\begin{align}
    T-1
    &=\operatorname{rank}\partial_{\boldsymbol\vartheta}F_\ell(\boldsymbol\vartheta^\star)
    \le\min\bigl\{\operatorname{rank}\partial_{\boldsymbol\vartheta}s_\ell(\boldsymbol\vartheta^\star),\,
    p_\ell\bigr\}
    \nonumber\\
    &\le\min\{k_\ell,\,p_\ell\}.
    \label{eq:rs-rank}
\end{align}
Hence $p_\ell\ge T-1$. Moreover, $F_\ell$ is a submersion near
$\boldsymbol\vartheta^\star$, so it maps every neighborhood of
$\boldsymbol\vartheta^\star$ onto a neighborhood of
$\boldsymbol q$~\cite{lee2013smooth}: near an exact solution, the $T-1$
flip probabilities can be varied independently by varying
$\boldsymbol\vartheta$. This is the sense in which the shared block
contains at least $T-1$ independent continuous parameter settings; it
is a dimension count of the type used in
Ref.~\cite{funcke2021dimensional}.

The countability of layouts is essential. If fixed gate settings could
be arbitrary real numbers that are not counted as parameters, then the
parameter-free clock-controlled layouts of
Sec.~\ref{sec:rs-attain}, with the angles fixed, would already form an
uncountable family that generates every $P_{\boldsymbol q}$.

\paragraph{Gates, width, and depth.}
Eq.~\eqref{eq:rs-rank} and $k_\ell\le cG$ give $G\ge(T-1)/c$. Each
layer of the block contains at most $N$ gates, so $G\le NL$ and hence
$NL\ge(T-1)/c$, which proves Eq.~\eqref{eq:existence_width_depth}.
Under complete-block repetition, all $G$ gates are executed at each of
the $T-1$ transitions, so
\begin{equation}
    G_{\rm traj}\ge(T-1)G\ge\frac{(T-1)^2}{c}.
    \label{eq:rs-gtraj}
\end{equation}
Every layer of the circuit that produces a trajectory also contains at
most $N$ gates, so
\begin{equation}
    D_{\rm traj}\ge\frac{G_{\rm traj}}{N}\ge\frac{(T-1)^2}{cN}.
    \label{eq:rs-dtraj}
\end{equation}
This bound holds even if consecutive executions of the block overlap in
time.

Together with item~1, this proves
Theorem~\ref{thm:temporal_resource_separation} for every
$\boldsymbol q\in\mathcal S_T\setminus\mathcal N_T$. Since
$\mathcal N_T$ is null, such a $\boldsymbol q$ exists for every $T\ge3$,
and $P_T:=P_{\boldsymbol q}$ has the stated properties.

\subsection{Clock-controlled rotations}

Let $m\ge\lceil\log_2(T-1)\rceil$, and consider an $m$-qubit clock
register $C$ and a data qubit $D$. The uniformly controlled rotation
\begin{equation}
    W_{\boldsymbol\alpha}
    =\sum_{j=0}^{2^m-1}\lvert j\rangle\langle j\rvert_C
     \otimes R_y(\alpha_j)_D
    \label{eq:rs-ucr}
\end{equation}
encodes the schedule if $\sin^2(\alpha_j/2)=q_{j+2}$ for
$0\le j\le T-2$; the remaining angles are arbitrary.

\begin{lemma}[Clock-controlled rotations]
\label{lem:rs-ucr}
Consider circuits over the gate families of the shared generators of
Appendix~\ref{proof:resource separation}, with countably many
layouts and possibly with ancilla qubits that start and end in a fixed
state. For every $\boldsymbol q\in\mathcal S_T$ outside a
Lebesgue-null set, every such circuit that implements an operator
$W_{\boldsymbol\alpha}$ encoding $\boldsymbol q$ exactly, up to a global
phase, has at least $T-1$ independent continuous settings and hence at
least $(T-1)/c$ elementary gates.
\end{lemma}

\begin{proof}
Let a layout $\ell$ implement $V_\ell(\boldsymbol\vartheta)$ and let
$\lvert a_0\rangle$ be the fixed ancilla state. For $0\le j\le T-2$,
define
$h_{\ell,j}(\boldsymbol\vartheta)=
|\langle j,1,a_0|V_\ell(\boldsymbol\vartheta)|j,0,a_0\rangle|^2$.
Exact implementation requires
$h_{\ell,j}(\boldsymbol\vartheta)=\sin^2(\alpha_j/2)=q_{j+2}$. The map
$h_\ell=(h_{\ell,0},\ldots,h_{\ell,T-2})$ is smooth and depends on
$\boldsymbol\vartheta$ only through the continuous gate settings, so the
argument leading to Eq.~\eqref{eq:rs-rank}, with $F_\ell$ replaced by
$h_\ell$, applies unchanged.
\end{proof}

\subsection{Attainability of the gate-count bound}
\label{sec:rs-attain}

\begin{proposition}[Clock-based shared generator]
\label{prop:rs-attain}
Suppose that the gate families include arbitrary single-qubit
rotations and the CNOT gate, and let $m=\lceil\log_2(T-1)\rceil$. For
every $T\ge3$ and every $\boldsymbol q\in\mathcal S_T$, there is a shared
generator with $N=m+1$ qubits and
\begin{equation}
    G=2^{m+1}+5m^2-2m<4(T-1)+5m^2
    \label{eq:rs-attain-count}
\end{equation}
elementary gates per block that generates $P_{\boldsymbol q}$ exactly.
Hence $G=O(T)$ and $G_{\rm traj}=O(T^2)$, matching
Eq.~\eqref{eq:rs-gtraj}, and $D_{\rm traj}=O(T^2)$, within a factor
$O(\log T)$ of Eq.~\eqref{eq:rs-dtraj} at this width.
\end{proposition}

\begin{proof}
\emph{Protocol.} The clock $C$ is initialized in
$\lvert0\rangle^{\otimes m}$. The data qubit $D$ is prepared by a
Hadamard gate and measured, which produces a uniformly distributed $x_1$. At each
transition, the block applies $W_{\boldsymbol\alpha}$ of
Eq.~\eqref{eq:rs-ucr}, with $\sin^2(\alpha_j/2)=q_{j+2}$ for
$j\le T-2$ and $\alpha_j=0$ otherwise, followed by the increment
$\mathrm{INC}:\lvert j\rangle_C\mapsto\lvert j+1\bmod 2^m\rangle_C$.
Then $D$ is measured, the outcome is read out as $x_t$, and the
post-measurement state is reused.

\emph{Correctness.} By induction, before transition $t$ the clock is in
the basis state $\lvert t-2\rangle$, unentangled with $D$, and $D$ is in
$\lvert x_{t-1}\rangle$. The block therefore acts on $D$ as
$R_y(\alpha_{t-2})$, which flips it with probability $q_t$, and the
trajectory distribution is Eq.~\eqref{eq:rs-law}. Because $2^m\ge T-1$, the
clock values $0,\ldots,T-2$ used at the $T-1$ transitions are distinct.

\emph{Uniformly controlled rotation.} Let
$g(i)=i\oplus\lfloor i/2\rfloor$ be the binary reflected Gray code, and
let $c_i$ be the clock bit in which $g(i)$ and $g(i+1\bmod 2^m)$
differ. For $i=0,1,\ldots,2^m-1$ in turn, apply $R_y(\theta_i)$ to $D$
followed by a CNOT from clock bit $c_i$ to $D$, where
\begin{equation}
    \boldsymbol\theta=2^{-m}S^{\mathsf T}\boldsymbol\alpha,
    \qquad
    S_{ji}=(-1)^{j\cdot g(i)},
\end{equation}
and $j\cdot g(i)$ is the bitwise inner product modulo two. This circuit
implements $W_{\boldsymbol\alpha}$~\cite{bergholm2005quantum}. Indeed,
for clock value $j$, the CNOT after $R_y(\theta_i)$ acts on $D$ as
$X^{j_{c_i}}$, where $j_c$ denotes bit $c$ of $j$. Using
$XR_y(\theta)X=R_y(-\theta)$, all $X$ factors can be moved to the end:
$R_y(\theta_i)$ acquires the sign
$(-1)^{\sum_{i'<i}j_{c_{i'}}}=(-1)^{j\cdot g(i)}$, and the product of
all $X$ factors is the identity, because each bit changes an even number
of times along the cyclic Gray code. The net operation on $D$ is
$R_y\bigl(\sum_iS_{ji}\theta_i\bigr)=R_y(\alpha_j)$, since
$SS^{\mathsf T}=2^mI$. This part uses $2^m$ rotations and $2^m$ CNOTs.

\emph{Increment.} Write $\mathrm{INC}=\mathcal F^\dagger\Lambda\mathcal F$,
where $\mathcal F$ is the quantum Fourier transform on $C$ and
$\Lambda=\sum_y e^{2\pi iy/2^m}\lvert y\rangle\langle y\rvert$ is a
product of $m$ single-qubit phase gates~\cite{draper2000addition}. We
omit the final bit-reversal swaps of $\mathcal F$ and relabel the
qubits on which the phases of $\Lambda$ act. Then $\mathcal F$ consists
of $m$ Hadamard gates and $m(m-1)/2$ controlled phase gates, each of
which equals two CNOTs and three single-qubit phase rotations up to a
global phase. Hence $\mathrm{INC}$ uses $5m^2-2m$ elementary gates and
no ancilla. Phases are irrelevant here, because the clock always
remains in a computational-basis state.

\emph{Counting.} The block contains the $2^{m+1}$ gates of
$W_{\boldsymbol\alpha}$ and the $5m^2-2m$ gates of $\mathrm{INC}$, which
gives Eq.~\eqref{eq:rs-attain-count} because $2^m<2(T-1)$. Its depth is
at most $G$. The settings fixed by the layout are the phases of
$\mathrm{INC}$, which are rational multiples of $\pi$, and the $2^m$
angles $\theta_i$ are the continuous parameters. The generator
therefore belongs to the class of
Theorem~\ref{thm:temporal_resource_separation}, consistently with
$p=2^m\ge T-1$. Executing the block $T-1$ times gives
$G_{\rm traj}=O(T^2)$ and $D_{\rm traj}=O(T^2)$.
\end{proof}

\section{Additional numerical results}
\label{app:numerical_detail}

In this section, we provide additional numerical details for HQCS.

\begin{figure}[t]
    \centering
    \includegraphics[width=\textwidth]{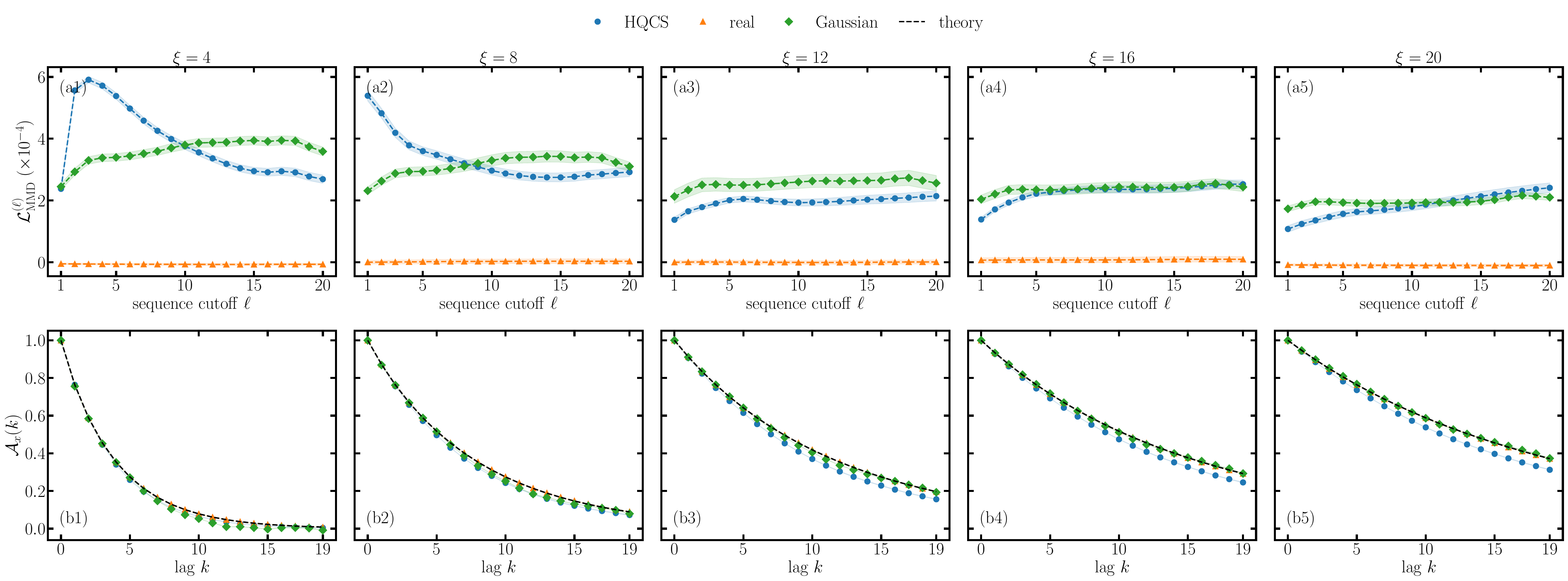}
    \caption{Numerical details in the generative learning of AR(1) with various temporal correlation lengths from $\xi=4$ to $20$ (left to right). We plot HQCS (blue), real data baseline (orange) and Gaussian benchmark (green). The dashed black line indicates the theoretical autocorrelation.   \label{fig:ar1_addtion}}
\end{figure}

In Fig.~\ref{fig:ar1_addtion}, we provide the performance of HQCS in generative learning of AR(1) with various correlation lengths $\xi = 4$ to $20$ (left to right). With only 1,000 training samples, we find close alignment in the MMD between HQCS (circles) and Gaussian model (diamonds) shown in the top panel of Fig.~\ref{fig:ar1_addtion}. We also observe a dependence of MMD on the correlation length $\xi$. For relatively short-range correlation of $\xi/T\ll 1$, i.e. $\xi = 4$, $\calL_{\rm MMD}^{(\ell)}$ increases with $\ell$ with a peak at about $\ell = 3$ which can be related to the temporal correlation in real data. The further decrease of MMD at larger $\ell$ may be due to dilution of the distributional mismatch caused by adding new, nearly independent data. Meanwhile, for relatively long-range temporal correlations of $\xi/T \simeq 1$, i.e. $\xi=20$, HQCS exhibits a relatively stable performance in MMD. Note that the classification of short and long-range temporal correlation here is defined with respect to $T$. 
In the bottom panel of Fig.~\ref{fig:ar1_addtion}, we provide the results for ACF, and as the correlation length $\xi$ grows, the gap between HQCS and the theoretical result appears at long lag $k$. Table~\ref{tab:ar1_acf_xi} reports the fitted correlation parameters for HQCS and the Gaussian model alongside the theoretical values.

Tables~\ref{tab:vix_fig5_mmd} and~\ref{tab:vix_fig6_mmd} report the MMD loss and increment MMD underlying Figs.~\ref{fig:vix_main} and~\ref{fig:experimental_validation} in the main text at representative window lengths. Entries give the mean and standard error across $20$ sample sets of $20{,}000$ trajectories. A common scale of $10^{-3}$ is used in both tables to facilitate comparison, including uncertainties that are difficult to resolve in the shaded areas. The experimental sample sets are disjoint shot batches rather than independent hardware runs, and small negative values of the unbiased MMD estimator are retained.

\begin{table}[H]
    \centering
    \caption{Correlation parameters extracted from the normalized AR(1) ACF for different target correlation lengths $\xi$. The theoretical parameter is $\phi=\exp(-1/\xi)$. For HQCS and the Gaussian model, we fit $\mathcal{A}_x(k)=6\arcsin(\phi^k/2)/\pi$ to the mean empirical ACF over lags $k=1,\ldots,19$, using the ACF standard errors across $20$ independent batches of $20{,}000$ length-$20$ trajectories as weights. Entries report $\hat\phi\pm\mathrm{SE}$ from the weighted least-squares fit; these fit standard errors do not account for covariance between lags.}
    \label{tab:ar1_acf_xi}
    \begin{ruledtabular}
    \begin{tabular}{cccc}
        Target $\xi$ & Theory: $\phi$ & HQCS: $\hat\phi\pm\mathrm{SE}$ & Gaussian: $\hat\phi\pm\mathrm{SE}$ \\
        \hline
         4 & $0.77880078$ & $0.77312336\pm0.00012879$ & $0.77103240\pm0.00012234$ \\
         8 & $0.88249690$ & $0.87505699\pm0.00006315$ & $0.87912508\pm0.00006579$ \\
        12 & $0.92004441$ & $0.91143192\pm0.00005455$ & $0.91863064\pm0.00003970$ \\
        16 & $0.93941306$ & $0.93237578\pm0.00004264$ & $0.93907650\pm0.00003571$ \\
        20 & $0.95122942$ & $0.94375837\pm0.00003854$ & $0.95141409\pm0.00003148$ \\
    \end{tabular}
    \end{ruledtabular}
\end{table}

\begin{table}[H]
    \centering
    \caption{Numerical values for Fig.~\ref{fig:vix_main} in the main text. Entries are mean $\pm$ standard error of the unbiased MMD loss, in units of $10^{-3}$, evaluated on raw VIX values in (a) and their increments in (b).}
    \label{tab:vix_fig5_mmd}
    \begin{ruledtabular}
    \begin{tabular}{lccc}
        \multicolumn{4}{c}{(a) MMD loss ($\times 10^{-3}$)} \\
        Model / data & $\ell=1$ & $\ell=10$ & $\ell=20$ \\
        \hline
        HQCS & $1.2978\pm0.0319$ & $0.7679\pm0.0330$ & $0.6565\pm0.0324$ \\
        Real VIX & $-0.0035\pm0.0058$ & $-0.0039\pm0.0058$ & $-0.0025\pm0.0059$ \\
        Discretized real VIX & $0.4236\pm0.0098$ & $0.3842\pm0.0112$ & $0.3950\pm0.0126$ \\
        Gaussian model & $23.0271\pm0.2019$ & $22.5620\pm0.2113$ & $22.8033\pm0.2108$ \\
        TimeGAN & $26.2712\pm0.1527$ & $33.9937\pm0.1350$ & $45.4816\pm0.1640$ \\
        \hline
        \multicolumn{4}{c}{(b) Increment MMD ($\times 10^{-3}$)} \\
        Model / data & $r=1$ & $r=9$ & $r=19$ \\
        \hline
        HQCS & $0.4308\pm0.0097$ & $8.2051\pm0.0371$ & $8.3681\pm0.0693$ \\
        Real VIX & $0.0008\pm0.0015$ & $-0.0020\pm0.0025$ & $-0.0034\pm0.0020$ \\
        Discretized real VIX & $0.4887\pm0.0105$ & $0.3357\pm0.0174$ & $0.5407\pm0.0242$ \\
        Gaussian model & $30.6168\pm0.1085$ & $42.3761\pm0.2019$ & $35.9101\pm0.1815$ \\
        TimeGAN & $73.4907\pm0.2037$ & $122.4177\pm0.4102$ & $157.3860\pm0.4474$ \\
    \end{tabular}
    \end{ruledtabular}
\end{table}

\begin{table}[H]
    \centering
    \caption{Numerical values for Fig.~\ref{fig:experimental_validation} in the main text. Entries are mean $\pm$ standard error of the unbiased MMD loss, in units of $10^{-3}$, evaluated on raw VIX values in (a) and their increments in (b).}
    \label{tab:vix_fig6_mmd}
    \begin{ruledtabular}
    \begin{tabular}{lccc}
        \multicolumn{4}{c}{(a) MMD loss ($\times 10^{-3}$)} \\
        Model / data & $\ell=1$ & $\ell=10$ & $\ell=20$ \\
        \hline
        HQCS (ideal) & $1.3053\pm0.0310$ & $0.7606\pm0.0336$ & $0.6463\pm0.0333$ \\
        Experiment (IBM Phoenix) & $8.6411\pm0.1967$ & $12.4984\pm0.2016$ & $16.6244\pm0.1906$ \\
        Gaussian model & $23.0271\pm0.2019$ & $22.5620\pm0.2113$ & $22.8033\pm0.2108$ \\
        TimeGAN & $26.2712\pm0.1527$ & $33.9937\pm0.1350$ & $45.4816\pm0.1640$ \\
        \hline
        \multicolumn{4}{c}{(b) Increment MMD ($\times 10^{-3}$)} \\
        Model / data & $r=1$ & $r=9$ & $r=19$ \\
        \hline
        HQCS (ideal) & $0.4214\pm0.0120$ & $8.2615\pm0.0464$ & $8.4212\pm0.0704$ \\
        Experiment (IBM Phoenix) & $4.1755\pm0.0439$ & $28.2708\pm0.2740$ & $57.7443\pm0.4645$ \\
        Gaussian model & $30.6168\pm0.1085$ & $42.3761\pm0.2019$ & $35.9101\pm0.1815$ \\
        TimeGAN & $73.4907\pm0.2037$ & $122.4177\pm0.4102$ & $157.3860\pm0.4474$ \\
    \end{tabular}
    \end{ruledtabular}
\end{table}

\begin{figure}[t]
	\centering
	\includegraphics[width=0.72\textwidth]{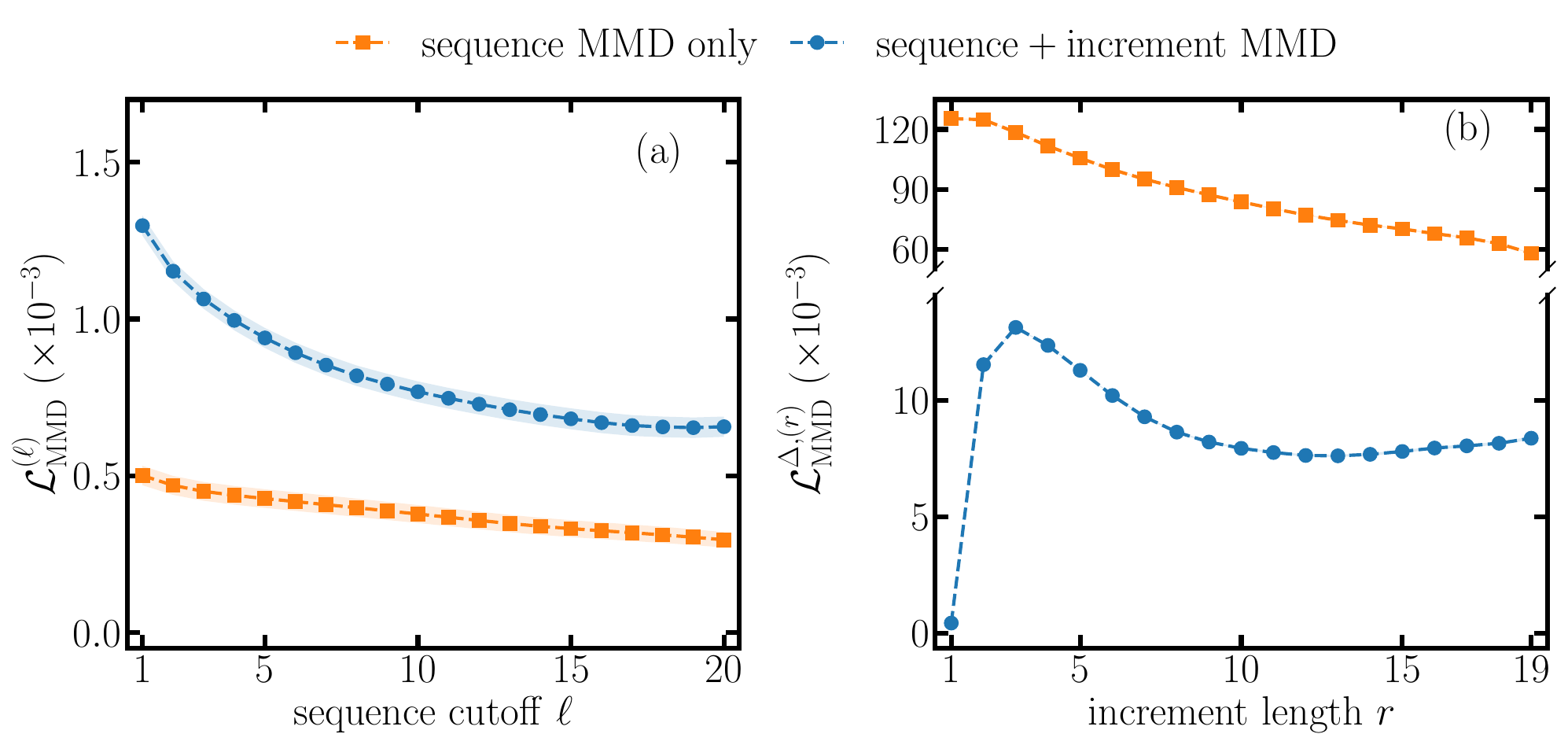}
	\caption{Comparison of HQCS trained with sequence MMD alone (orange squares) and with sequence MMD plus an increment marginal MMD term (\(r=1\); blue circles). Both models are evaluated against training-period real VIX data in raw space. (a) Sequence MMD \(\mathcal{L}_{\mathrm{MMD}}^{(\ell)}\) versus sequence cutoff length \(\ell\). (b) Increment MMD \(\mathcal{L}_{\mathrm{MMD}}^{\Delta,(r)}\) versus increment vector length \(r\), with a broken vertical axis. 	\label{fig:vix_deltax}}
\end{figure}

In Fig.~\ref{fig:vix_deltax}, we compare HQCS performance with/out additional increment MMD training loss. Including the increment marginal MMD term substantially improves agreement with real VIX increments, shown in Fig.~\ref{fig:vix_deltax}(b). The discrepancy decreases by about two orders of magnitude at \(r=1\) and about one order of magnitude at \(r=19\), suggesting that matching single-step changes also benefits the joint increment distribution. At the same time, their performance in the original sequence-space MMD remains comparable as they are both close to the finite-precision real data baseline (Fig.~\ref{fig:vix_main} in the main text).

\begin{figure}[t]
    \centering
    \includegraphics[width=0.72\textwidth]{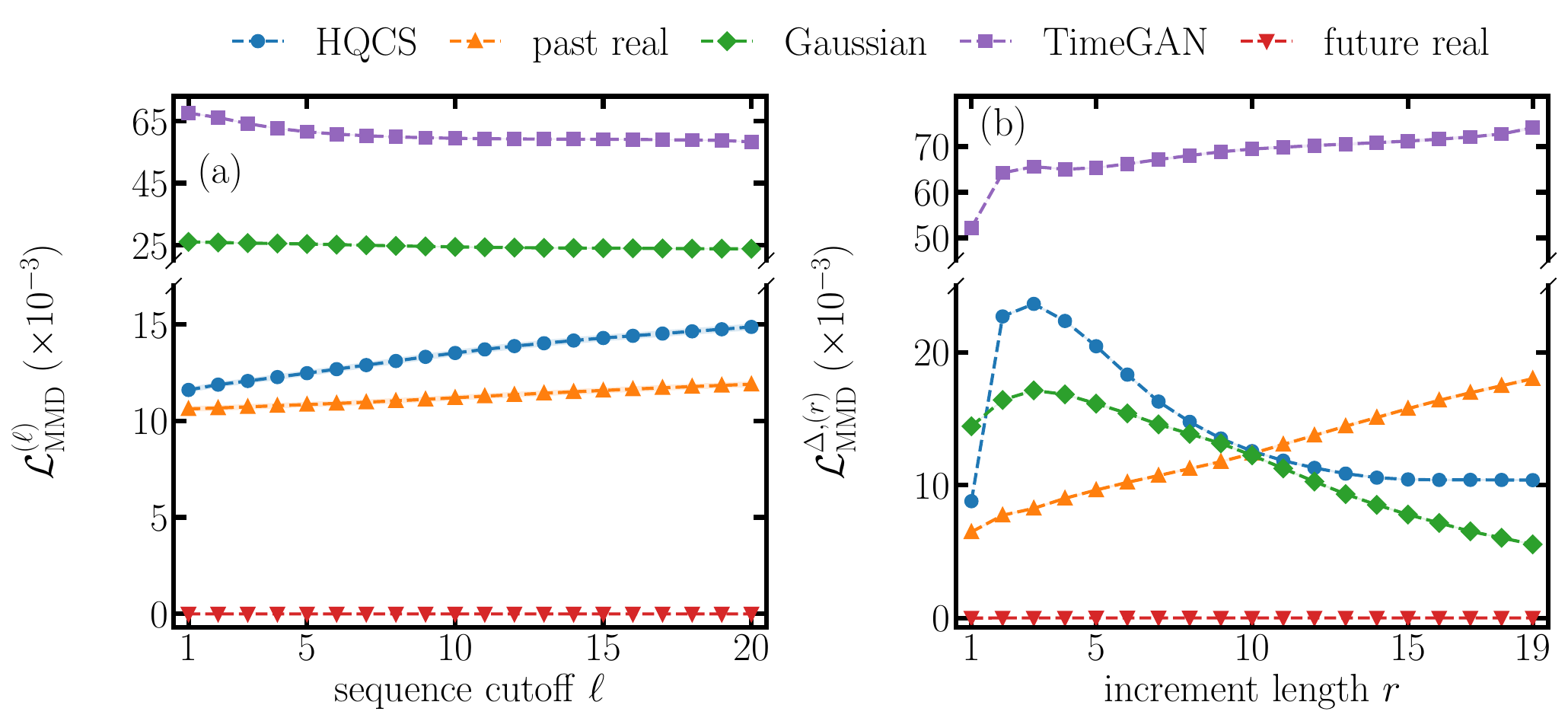}
    \caption{MMD losses versus sequence cutoff length $\ell$ for HQCS (circles), real VIX data from training period (up triangles), real future VIX data (down triangles), Gaussian model (diamonds) and TimeGAN (squares).
    Shaded areas show the standard error of the mean across 20 iid sample sets.
    Here we choose $N_A=3, N_B=5$ with a total length of $T=20$ steps. In each step the circuit depth is $8$. The bandwidth of Gaussian kernel is $\sigma=6.28$. }
    \label{fig:vix_future}
\end{figure}
In addition to the VIX benchmarking of Fig.~\ref{fig:vix_main} in the main text, we also evaluate the trained generators against future VIX data to examine their performance beyond the training period, which can be regarded as out-of-distribution generation. The training data covers 1990/09/18–2018/06/28, while the future data covers 2018/06/29–2026/06/12. Fig.~\ref{fig:vix_future} shows the generalization performance evaluated by the MMD loss to the future period data. We use 20 sample sets of 20,000 length-\(20\) sequences for each model and sample future reference sequences uniformly from all available windows. Each model is compared with the same future reference sets. We also compare real past sequences and independently-sampled real future sequences with these references as a baseline. HQCS achieves lower MMD loss than both Gaussian and TimeGAN across all sequence cutoffs to maintain the learning advantage from training period. The real past–future comparison also yields a similar performance compared to HQCS, indicating a distributional difference between the two periods, whereas the future–future baseline is consistent with near zero. Thus, HQCS retains the strongest agreement with future data among the tested generators, although its discrepancy exceeds the real past–future reference. The future increment benchmark provides a complementary assessment of local temporal variation. HQCS (circles) and Gaussian benchmark (diamonds) share similar performance in increment MMD and both are close to the performance of real past data (up triangles). All of them outperform TimeGAN (squares). Its discrepancy also remains appreciably larger than the future-real baseline, indicating a consistent residual difference from the future increment distribution.

\begin{figure}[t]
    \centering
    \includegraphics[width=0.72\textwidth]{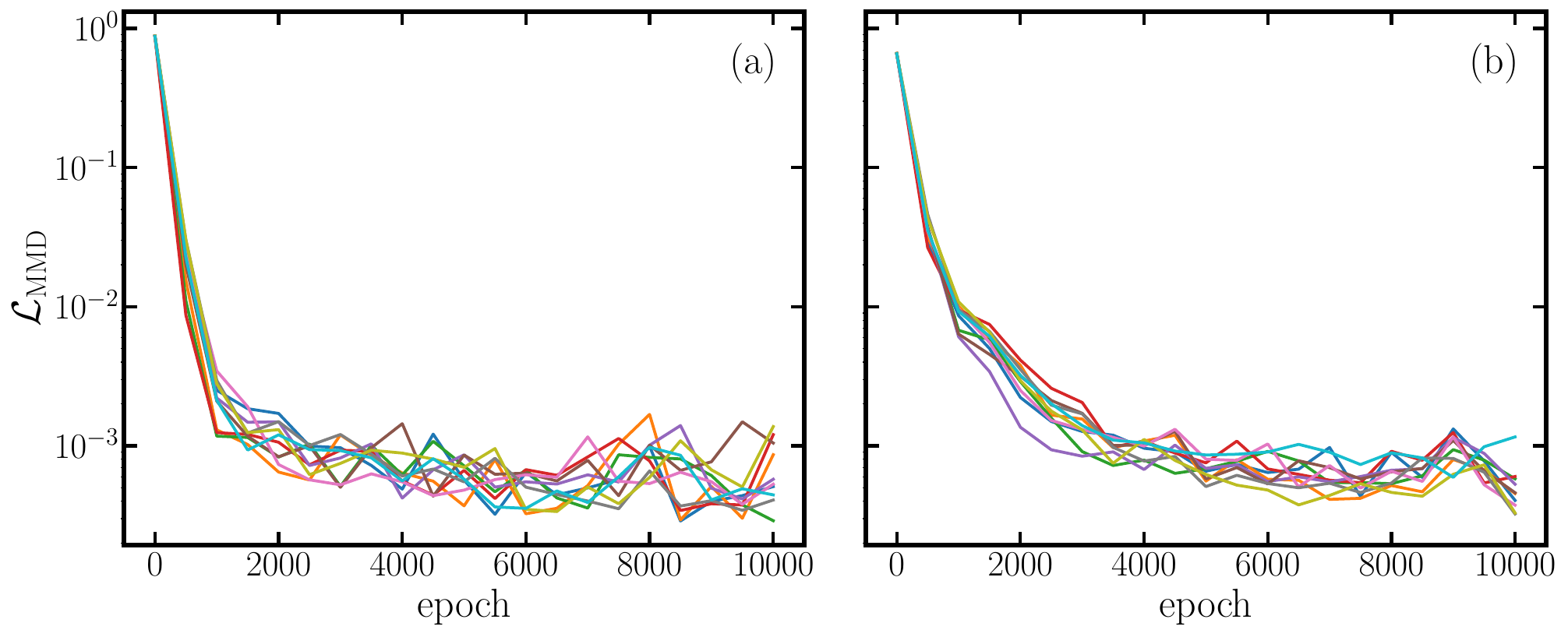}
    \caption{Training history of MMD loss of HQCS for 10 random initializations. In (a) and (b), we show the training for AR(1) and VIX index.    \label{fig:training_history}}
\end{figure}

In Fig.~\ref{fig:training_history}, we plot the training history of MMD loss for 10 randomly-initialized HQCS models for (a) AR(1) random process with $\xi = 16$ and (b) VIX index. Overall, all randomly-initialized HQCS models approximately converge to the same level of generative performance. The fluctuation in the MMD loss at late epochs may be due to the limited training samples.

\begin{figure}[t]
    \centering
    \includegraphics[width=0.72\textwidth]{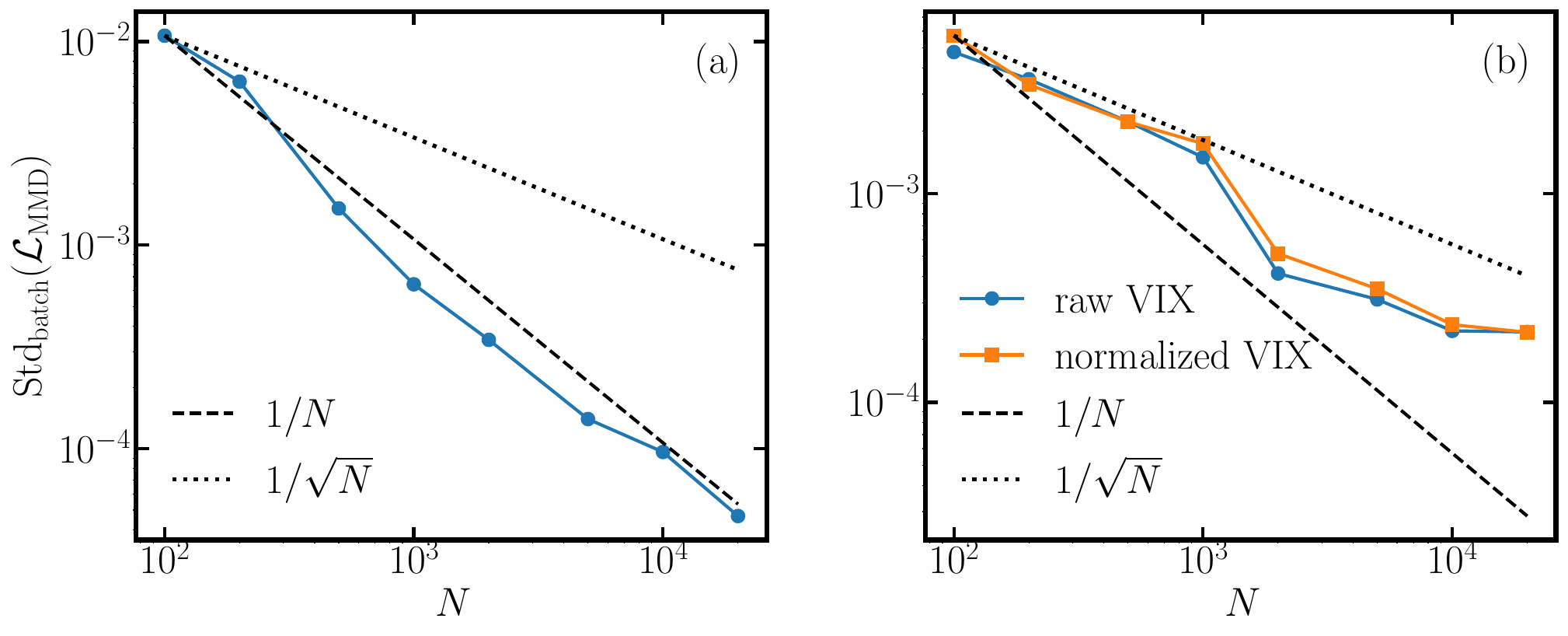}
    \caption{Sampling fluctuations of the full-sequence MMD loss for fixed trained HQCS models. We plot the sample standard deviation of the unbiased estimator of squared MMD, $\calL_{\rm MMD}$, across repeated evaluations using $N$ target and $N$ generated trajectories of length $T=20$, with $N$ ranging from $100$ to $20{,}000$. (a) CDF-normalized AR(1) with $\xi=20$, evaluated over $30$ batches. (b) VIX in raw units (circles) and after empirical-CDF normalization (squares), evaluated over $20$ batches. Target VIX windows are sampled uniformly with replacement from the historical dataset. The Gaussian-kernel bandwidth is $\sigma = 0.3$ and $6.28$ for normalized and raw data. Black dashed and dotted lines indicate $1/N$ and $1/\sqrt{N}$, respectively.  \label{fig:mmd_sample_complexity}}
\end{figure}

In Fig.~\ref{fig:mmd_sample_complexity}, we quantify the sampling uncertainty in evaluating a fixed trained generator. For AR(1), the batch standard deviation decreases approximately as $1/N$ over the sampled range from $N = 100$ to $2\times 10^4$, while the VIX curves decrease more slowly, broadly following the $1/\sqrt{N}$. These trends are consistent with the two fluctuation regimes of the unbiased squared-MMD estimator: its standard deviation scales as $O(N^{-1})$ when the two distributions coincide, whereas a nonzero leading variance term generally gives $O(N^{-1/2})$ for distinct distributions~\cite{gretton2012kernel}. The faster decrease for AR(1) is compatible with a small distributional discrepancy over this finite range, and increasing $N$ improves the precision of the benchmark. Here the VIX fluctuations are conditional on resampling the historical window distribution, and the real distribution for real VIX data is only empirical, which can also enlarge the discrepancy. Overall, these numerical results provide guidance on the number of samples required for benchmarking.

\clearpage

\section{On-device training of HQCS}
\label{app:on-device-training}

\subsection{Hardware setup and training objective}
We demonstrate on-device fine-tuning of HQCS on the IBM Heron r2 processor \texttt{ibm\_marrakesh}. Starting from a classically pretrained univariate VIX model, we optimize its circuit parameters using losses estimated from hardware-generated trajectories. The model has $N_A=3$ memory qubits and $N_B=5$ data qubits, with circuit depth $L=8$ per temporal block and sequence length $T=5$. Here $L$ denotes circuit depth, as in Appendix~\ref{app:details}. The implementation uses Qiskit dynamic circuits with mid-circuit measurement, feed-forward and qubit reuse, without reset. Each measured data-register outcome is decoded to a raw VIX value using the pretrained model's fixed $2^{N_B}=32$-level lookup table.

The training objective is the full-sequence MMD loss $\calL_{\rm MMD}^{(T)}$ with $\ell=T=5$ in Eq.~\eqref{eq:window_mmd_def} in the main text, evaluated between generated and historical VIX trajectories in raw units with a Gaussian-kernel bandwidth $\sigma=6.28$. For this on-device experiment, we use the \emph{unbiased estimator of squared MMD} during optimization, rather than the biased estimator used for the simulator training described in the main text. Each loss evaluation uses $S=2000$ measured circuit shots, each yielding one complete trajectory. The quantum processor supplies the trajectory samples, while the loss calculation and parameter updates are performed classically.

\subsection{SPSA updates and resource use}
We use simultaneous perturbation stochastic approximation (SPSA)~\cite{spall1992multivariate}, which has also been applied to quantum learning~\cite{zhuang2019physical}. Let $\Theta_k$ be the circuit-parameter vector at iteration $k=0,\ldots,49$. For each of two independent perturbation directions, indexed by $r=1,2$, we draw a vector $\Delta_{k,r}$ whose components independently take the values $+1$ and $-1$ with equal probability. The gradient estimate and update are
\begin{align}
\widehat{g}_{k,j}
&=\frac{1}{2}\sum_{r=1}^{2}
\frac{\widehat{\calL}(\Theta_k+c_k\Delta_{k,r})
      -\widehat{\calL}(\Theta_k-c_k\Delta_{k,r})}
     {2c_k\Delta_{k,r,j}},
\label{eq:on-device-spsa-gradient}\\
\Theta_{k+1}&=\Theta_k-a_k\widehat{g}_k,
\label{eq:on-device-spsa-update}
\end{align}
where $\widehat{\calL}$ is the hardware-estimated full-sequence loss. We use the gain schedules
\begin{equation}
a_k=\frac{a}{(k+1+A)^\alpha},\qquad
c_k=\frac{c}{(k+1)^\gamma},
\label{eq:on-device-spsa-gains}
\end{equation}
with $a=0.05$, $c=0.10$, $A=5$, $\alpha=0.602$ and $\gamma=0.101$. The offset $A$ is $10\%$ of the total number of iterations. Each perturbation direction requires two loss evaluations, so averaging over two directions requires four evaluations per parameter update, independently of the number of circuit parameters. The $50$ updates therefore require $200$ perturbed loss evaluations. Including monitoring evaluations, the run comprises approximately $210$ circuit evaluations, each comprising $2000$ shots, and approximately $630\,\mathrm{s}$ of billed QPU time.

\clearpage
\subsection{Observed improvement}
\begin{figure}[htbp]

    \centering
    \includegraphics[width=0.82\textwidth]{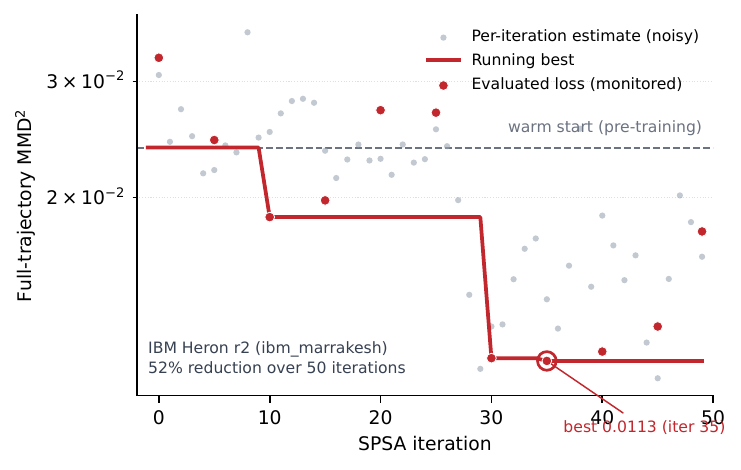}
    \caption{On-device SPSA fine-tuning of HQCS on IBM Heron r2 (\texttt{ibm\_marrakesh}). The full-sequence squared-MMD loss is evaluated in raw VIX units for trajectories of length $T=5$, using a Gaussian kernel with bandwidth $\sigma=6.28$ and $2000$ shots per loss evaluation. Gray points show noisy per-iteration estimates, red circles show monitored losses, and the solid red line shows the running best, including the initial hardware baseline. The dashed line marks the loss $0.0238$ of the classically pretrained parameters evaluated on hardware. The lowest monitored loss is $0.0113$ at iteration $35$, a reduction of approximately $52\%$ relative to this baseline. This minimum is the best observed value during the $50$-iteration run, rather than the final monitored loss.}
    \label{fig:on-device-spsa}
\end{figure}

Figure~\ref{fig:on-device-spsa} shows the on-device optimization history. The hardware-estimated loss at the pretrained parameters is $0.0238$. During on-device optimization, the lowest monitored loss reaches $0.0113$ at iteration $35$, corresponding to a reduction of approximately $52\%$ relative to the initial hardware value. The running-best monitored loss improves from iteration $10$ onward and retains this minimum through the end of the $50$-iteration run. These values describe the best monitored shot-based estimates; the running-best statistic is nonincreasing by construction and does not imply that every parameter update lowers the loss.

This experiment demonstrates that hardware feedback can improve the measured generative loss of a pretrained HQCS model despite device noise and finite-shot fluctuations. It complements the ideal-ensemble loss-concentration bounds of Theorem~\ref{thm:hqcs-main-concentration} in the main text, which do not establish convergence of a noisy hardware optimizer. The demonstration concerns fine-tuning from a classically pretrained initialization at $T=5$; the eight-qubit circuit remains classically simulable, and these results do not constitute a demonstration of quantum computational advantage.

\clearpage
\subsection{Generative performance after fine-tuning}
\begin{figure}[htbp]

    \centering
    \includegraphics[width=0.90\textwidth]{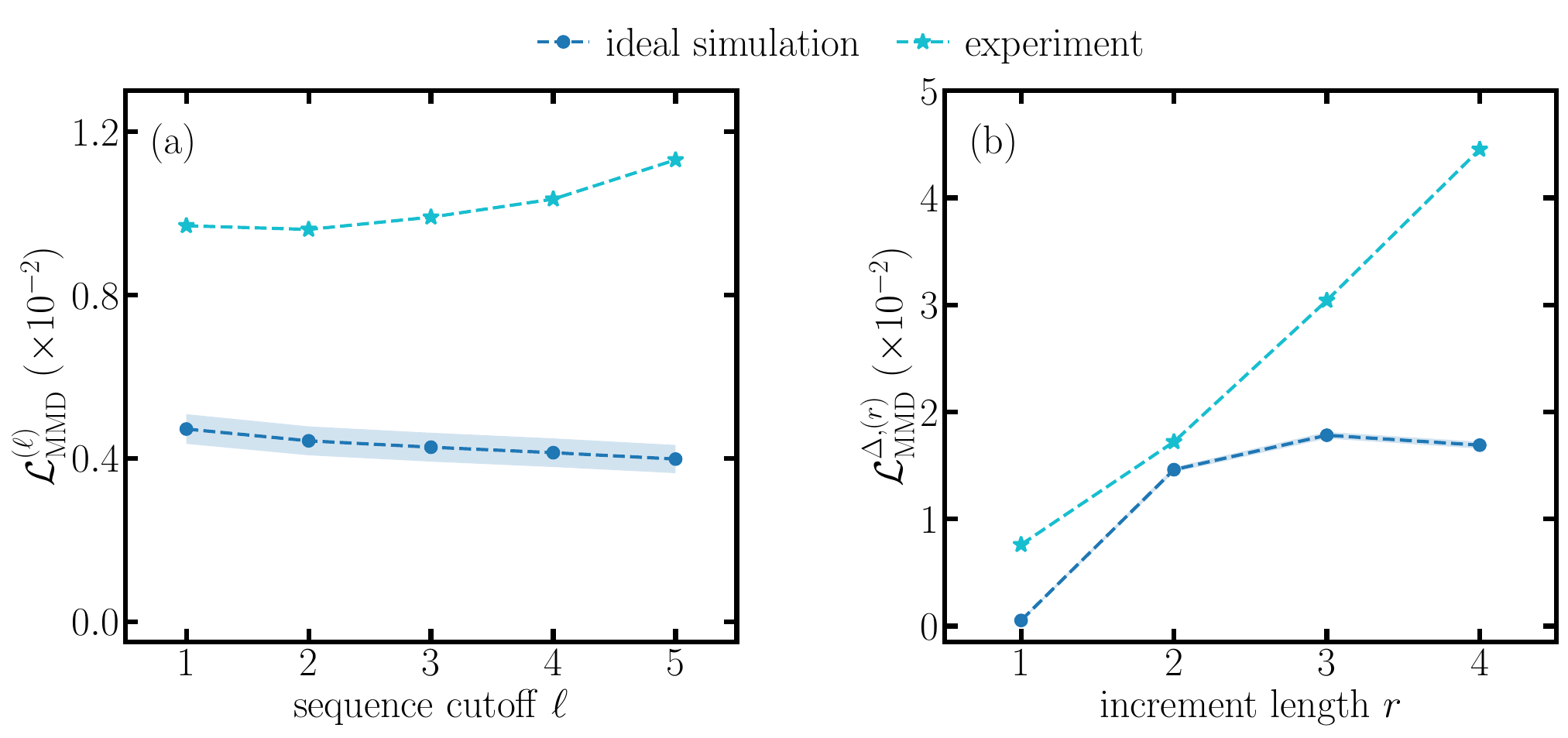}
    \caption{Generative performance of HQCS after on-device fine-tuning on IBM Heron r2 (\texttt{ibm\_marrakesh}). (a) Sequence-space MMD loss versus cutoff length $\ell$ and (b) increment MMD versus increment length $r$ for ideal simulation (circles) and experiment (stars), using the same fine-tuned parameters. The model has $N_A=3$, $N_B=5$, circuit depth $L=8$ per temporal block and sequence length $T=5$, and is fine-tuned from a classically pretrained initialization using $50$ SPSA iterations. Each MMD evaluation uses $2000$ generated trajectories. Ideal points and shaded bands show the mean and standard error across ten disjoint groups of $2000$ trajectories. The Gaussian-kernel bandwidths are $\sigma=6.28$ in (a) and $\sigma=0.96$ in (b), as in Fig.~\ref{fig:experimental_validation} in the main text.}
    \label{fig:heron_ondevice_training}
\end{figure}

To evaluate generative performance after the $50$-iteration SPSA run, we compare hardware measurements with ideal simulation of the same fine-tuned parameters in Figure~\ref{fig:heron_ondevice_training}. Each comparison uses $2000$ generated trajectories, with ideal results averaged over ten disjoint groups of this size. Figure~\ref{fig:heron_ondevice_training}(a) shows that the experimental sequence MMD (stars) is approximately two to three times the ideal value (circles) across the tested cutoff lengths. The experimental increment MMD continues to increase with increment length $r$, as shown in Figure~\ref{fig:heron_ondevice_training}(b).

Together with the monitored loss reduction in Figure~\ref{fig:on-device-spsa}, these results support the potential of hardware-in-the-loop fine-tuning while showing that a residual hardware--ideal gap persists. This experiment uses IBM Marrakesh and $T=5$, whereas Fig.~\ref{fig:experimental_validation} in the main text reports results on IBM Phoenix at $T=20$. The smaller sequence-MMD gap here therefore cannot be attributed to fine-tuning alone.

\end{document}